\documentclass[12pt]{article}
\usepackage{textcomp}
\usepackage[latin9]{inputenc}
\usepackage{geometry}
\usepackage{float}
\usepackage{booktabs}
\usepackage{amsmath}
\usepackage{amsthm}
\usepackage{amssymb}
\usepackage{graphicx}
\usepackage{setspace}
\usepackage[pdfusetitle,
 bookmarks=true,bookmarksnumbered=false,bookmarksopen=false,
 breaklinks=false,pdfborder={0 0 1},backref=false,colorlinks=false]
 {hyperref}

\makeatletter

\providecommand{\tabularnewline}{\\}

\theoremstyle{plain}
\newtheorem{assumption}{\protect\assumptionname}
\theoremstyle{remark}
\newtheorem{rem}{\protect\remarkname}
\theoremstyle{definition}
 \newtheorem{example}{\protect\examplename}
\theoremstyle{plain}
\newtheorem{lem}{\protect\lemmaname}
\theoremstyle{plain}
\newtheorem{prop}{\protect\propositionname}
\theoremstyle{plain}
\newtheorem{thm}{\protect\theoremname}
\theoremstyle{plain}
\newtheorem{cor}{\protect\corollaryname}
\theoremstyle{plain}
\newtheorem{lemapp}{\protect\lemmaname}

\AtBeginDocument{\epstopdfDeclareGraphicsRule{.eps}{pdf}{.pdf}{repstopdf --outfile="\OutputFile" "\SourceFile"}}
\usepackage{pgfplots}
\usepackage{pgfplotstable}
\usepackage{xcolor}
\usepackage{tikz}
\usepackage{tkz-graph} 
\usetikzlibrary{positioning,matrix,calc}
\usetikzlibrary{patterns}
\usetikzlibrary{decorations.pathmorphing}
\usepackage{lscape}  
\usepackage{array}  
\usepackage{caption}  
\usepackage{siunitx}  
\usepackage{threeparttable}  
\usepackage{algorithm}
\usepackage{algpseudocode} 

\newcolumntype{C}[1]{>{\centering\arraybackslash}p{#1}}
\PassOptionsToPackage{position=top}{subfig}

\makeatother

\usepackage[bibstyle=authoryear,citestyle=authoryear-comp,maxnames=4, uniquename=false, uniquelist=false]{biblatex}
\providecommand{\assumptionname}{Assumption}
\providecommand{\corollaryname}{Corollary}
\providecommand{\examplename}{Example}
\providecommand{\lemmaname}{Lemma}
\providecommand{\propositionname}{Proposition}
\providecommand{\remarkname}{Remark}
\providecommand{\theoremname}{Theorem}

\begin{document}
\title{Repairing Locally Misspecified GMM: An Empirical Bayes Approach}
\author{Patrick Kline\thanks{This manuscript was prepared for the 2026 Sargan lecture. I am grateful
to Isaiah Andrews, Kevin Chen, Andres Santos, and Jinglin Yang for
helpful feedback on earlier versions of this draft. Luan Borelli, Claude Opus, and Refine.ink provided outstanding research assistance on this project.}\\
UC Berkeley}
\maketitle
\begin{abstract}
\begin{singlespace}
Econometric models offer parsimonious but inexact approximations to
data generating processes. This paper studies the generalized method
of moments (GMM) when exchangeable specification errors of order $n^{-1/2}$
contaminate the moment conditions. I develop estimators for the mean
and variance of these specification errors, establishing their consistency
in an asymptotic framework where the number of overidentifying restrictions
grows with the sample size. These hyperparameter estimates are used
to develop a feasible bias-corrected estimator of target
parameters. I also propose an empirical Bayes estimator that weakly improves precision by subtracting a
best linear predictor of the first-order estimation error from the bias-corrected estimator. 
Using a combinatorial central limit theorem, I establish asymptotic normality of both estimators and provide
variance estimators that enable misspecification-aware frequentist
inference. Simulation exercises indicate the procedures can meaningfully
improve on standard two-stage least squares estimation when exclusion
violations are present. Revisiting the influential study of \textcite{angrist1991does}, I
consider an instrument set where exchangeable excludability violations are plausible.
Repairing the two-stage least squares estimates of the returns to
schooling moves them in the direction of ordinary least squares and
reduces sensitivity to the specification of controls.

\medskip
\noindent\textbf{Keywords:} Misspecification, exchangeability, empirical Bayes.

\noindent\textbf{JEL codes:} C11, C13, C26, C52.
\end{singlespace}
\end{abstract}
\newpage{}

The generalized method of moments (GMM) is a widely used technique
for fitting parsimonious semi-parametric models to data. The procedure
consists of specifying a function $g:\mathbb{R}^{p}\to\mathbb{R}^{m}$
whose zeros define $m$ moment restrictions on a parameter vector
$\theta\in\mathbb{R}^{p}$. A distinguishing feature of GMM, relative
to classical method of moments estimators, is that the number of moment
restrictions can exceed the number of parameters $\left(m\geq p\right)$.
The case where $m>p$ is known as an overidentified model \parencite{anderson1949estimation}.
In practice, GMM models often contain many overidentifying restrictions.
For instance, \textcite{angrist1991does} consider a specification
where $m-p=179$. The primary specification in \textcite{dellavigna2012testing}
has $m-p=55$, while \textcite{gourinchas2002consumption} consider
a specification with $m-p=36$.

Foundational work by \textcite{sargan1958estimation} and \textcite{hansen1982large}
developed a framework for testing overidentifying restrictions
by comparing the  $J$-statistic to its limiting distribution under
proper specification. In principle, $J$-tests impose valuable discipline on empirical work. Yet researchers can be wedded to their models, making statistical rejections awkward to explain to journal editors and referees. This awkwardness may explain why $J$-statistics are rarely reported in modern empirical work \parencite{andrews2025purpose}. When they are reported, rejections of overidentifying restrictions are often met with a degree of
ambivalence.\footnote{This ambivalence is not unique to the  $J$-test. In \textcite{hansen_sargent_2014_uncertainty},
Thomas Sargent recalls that ``Our good friend Robert E. Lucas, Jr.
told us in the early 1980s that our likelihood ratio tests and moment
matching tests were rejecting too many good models.'' \textcite{andrews2026misspecification} provide a decision-theoretic approach to dealing with potentially misspecified likelihood functions.} Rather than discard the model, researchers typically proceed with estimation, demurring that econometric models are ultimately approximations. However, the nature of these approximation errors and their influence on downstream conclusions are rarely quantified
precisely. 

This paper explores how to use the information in overidentifying
restrictions to build improved estimates of the model parameters  $\theta$ and their standard errors
that account for specification error. Following \textcite{andrews2017measuring},
I adopt a local misspecification framework in which the $m$
moment conditions are perturbed by errors that shrink with the sample
size. Unlike past work on local misspecification, I treat these errors as exchangeable random
variables. I propose estimators of the mean and variance of the specification
errors, the latter of which captures a form of overdispersion in the moment conditions that inflates the $J$-statistic.
I then establish consistency of these hyperparameter estimates in
an asymptotic regime where both $m$ and $n$ grow large but
$m^{2}/n\rightarrow0$, which rules out settings where
many-instruments biases become first order \parencite{altonji1996small,newey2004higher}.
This moderately overidentified regime arguably captures a large share
of settings where GMM is employed.

I use the proposed hyperparameter estimates to develop two improved
estimators of structural parameters.
The first estimator is a plug-in bias correction that subtracts an estimate of the
GMM estimator's first-order misspecification bias. In the special case of linear instrumental variables models, this correction is closely related to ``visual IV'' diagnostics \parencite{angrist2009mostly}, wherein reduced forms are regressed on first stages. The bias-corrected estimator allows for a non-zero intercept in this relationship, which effectively purges the fitted slope of an omitted-variables bias.

The second estimator is an empirical Bayes shrinkage
adjustment that subtracts a best linear predictor of the first-order
estimation error from the bias-corrected estimator. I show that the shrinkage adjustment weakly reduces leading-term estimation risk relative to bias correction by removing the predictable structure of the specification errors. This finding closely parallels the inadmissibility result of \textcite{brown1990ancillarity}, who established that adjusting for a high-dimensional nuisance regressor
can lower risk for a low-dimensional target. 

Leveraging the exchangeable structure of the specification errors,
I establish asymptotic normality of both estimators via a combinatorial central limit theorem. I use these results to develop consistent standard error
estimators that enable \emph{misspecification-aware} frequentist inference. Confidence intervals based on these standard errors are shown to yield asymptotically valid unconditional coverage (i.e., averaged over the unknown distribution of the specification errors).

A few papers take related approaches. One is \textcite{andrews2024true}
who propose an approach to inference on structural parameters when priors over moment misspecification are rotation invariant.
The distributional assumptions on specification errors considered
here nest rotation-invariance but require bounded fourth moments.
In particular, rotation-invariant priors must have mean zero, a restriction
I relax. However, \textcite{andrews2024true} deliver fixed-$m$
Bayesian inference while the frequentist misspecification-aware inferences
proposed here require that $m$ grows with $n$.

A second related contribution is \textcite{chernozhukov2025plausible} who
propose a quasi-Bayes approach to GMM in which priors are specified
over both the structural parameters and the specification errors.
As in this paper, they consider an environment in which $m$ grows
with $n$. While their Bayesian calibration and coverage statements
are contingent on an assumed prior over misspecification (e.g., Gaussianity),
I rely on exchangeability of the specification errors, alongside standard moment and regularity conditions, for estimation and inference.

Finally, earlier work by \textcite{kolesar2015identification} studies linear instrumental variables models featuring excludability violations in an environment where the number of instruments can grow with the sample size. They propose a bias-corrected k-class estimator predicated on these violations being orthogonal to the first-stage coefficients. The exchangeability assumption I invoke implies non-correlation in expectation but not almost surely. While their approach relies on homoscedastic errors in the first-stage and structural equations, I work in a more general GMM framework that accommodates both nonlinear moment conditions and heteroscedasticity.

Though the framework developed here offers new opportunities for misspecification-aware
estimation and inference, the methods proposed are not entirely automatic.
Moment conditions need to be scaled carefully to ensure violations
of the nominal model are plausibly exchangeable. These scaling decisions
are inherently context dependent, requiring commitment to a particular
model of misspecification. This limitation reflects the general observation
that misspecification can only be studied with respect to larger encompassing
models that plausibly hold in the setting under study \parencite{armstrong2025misspecification}. I develop diagnostics for detecting departures from exchangeability that can aid
in assessing the suitability of a given encompassing model.

To study the finite-sample performance of the methods, I conduct a series of simulation exercises that
consider exclusion failures in linear IV specifications. Concerns about excludability are often heightened in applied work using multiple instruments. For example, \textcite{angrist2017leveraging,angrist2024credible} considered a parametric empirical Bayes approach to remedying such failures.
The simulation design I study departs from the usual i.i.d.\ assumptions
employed in empirical Bayes modeling by considering specification
errors that are exchangeable but not independent. I find that both
of the proposed corrections can yield sizable risk improvements even when
the degree of overidentification is modest. To assess the theory's
asymptotic predictions, I simulate a setting where the number of moments
grows with the sample size. The simulations confirm that tests based
on the proposed misspecification-aware variance estimators have correct
asymptotic size. 

I then revisit the influential study of \textcite{angrist1991does},
examining schooling instruments based on season of birth by state
of birth interactions that plausibly satisfy exchangeability. I find
that moment conditions based on these instruments exhibit substantial
mean bias. The changing direction of the bias across control sets
is consistent with exogeneity failures of the sort posited by \textcite{rosenzweig2000natural}
and \textcite{buckles2013season}. The corrections consistently shift
the two-stage least squares (TSLS) estimates in the direction of ordinary
least squares (OLS), substantially narrowing the dispersion of point
estimates across alternative sets of controls. The corrected estimates
also exhibit larger standard errors than TSLS, reflecting the misspecification-aware
nature of the proposed variance estimators.

The rest of the paper is organized as follows. Section \ref{sec:GMM-preliminaries} reviews the basics of GMM under
proper specification and under local misspecification. Section \ref{sec:local_misspecification}
introduces the exchangeable misspecification framework and Section \ref{sec:Examples} discusses
some examples of models that fall into this framework.
Section \ref{sec:Identification-of-hyperparameter} studies identification
of the mean and variance of the specification errors. Section \ref{sec:A-differencing-approach}
discusses a GMM estimator based on differenced moment conditions,
which serves as a pedagogical bridge to the empirical Bayes estimators
of interest. Section \ref{sec:estimation} proposes estimators of
the mean and variance of the specification error distribution and
establishes their consistency. Section \ref{sec:bias-estimation}
proposes a plug-in bias-corrected estimator and an empirical Bayes
shrinkage estimator of target parameters that account for the specification
errors in different ways. Section \ref{sec:distribution-theory} develops
distribution theory for these estimators enabling misspecification-aware
inference. Section \ref{sec:monte-carlo} evaluates the finite-sample
performance of the estimators in a series of simulation exercises.
Section \ref{sec:empirical} revisits the study of \textcite{angrist1991does}.
Section \ref{sec:Conclusion} concludes with directions for future
work.

\section{GMM preliminaries \protect\label{sec:GMM-preliminaries}}

In this section I briefly review GMM estimation theory and state some standard
regularity conditions that are maintained throughout the analysis.
Suppose there exists a unique parameter vector $\theta_{0}\in\Theta$
obeying $g\left(\theta_{0}\right)=0$, where the vector $g:\mathbb{R}^{p}\rightarrow\mathbb{R}^{m}$ of population moment
conditions is continuously
differentiable. Assume the parameter space $\Theta$ is compact and
convex and that the true $\theta_{0}$ lies strictly in its interior.
Let the $m\times p$ matrix  $G$ denote the Jacobian $\nabla_{\theta}g\left(\theta\right)$
evaluated at  $\theta_{0}$ and
assume that $G$ has rank $p\leq m$. I will focus on the overidentified
case, assuming $m-p\geq2$. The smoothness and rank conditions are maintained
throughout my analysis. However, the nominal moment condition $g(\theta_{0})=0$
will be relaxed by allowing for local misspecification.

GMM estimates are obtained by constructing empirical moment conditions
$\hat{g}\left(\theta\right)$ that converge in probability to $g\left(\theta\right)$.
The GMM estimator $\hat{\theta}$ minimizes the quadratic form  $\hat{g}\left(\theta\right)'\hat{W}\hat{g}\left(\theta\right)$,
where $\hat{W}$ is some positive definite weighting matrix that converges
in probability to a nonsingular $m\times m$ matrix $W$. I assume
that $G'WG$ is non-singular. Letting $\hat{G}$ denote the Jacobian
of $\hat{g}\left(\theta\right)$ evaluated at $\hat{\theta}$, the
GMM estimator satisfies the first order condition $\hat{G}'\hat{W}\hat{g}\left(\hat{\theta}\right)=0.$

\subsection{Influence and sensitivity \protect\label{subsec:Influence-and-sensitivity}}

Suppose that the scaled empirical moments, when evaluated at the true
parameter vector  $\theta_{0}$, converge to normals:  
\[
\sqrt{n}\hat{g}\left(\theta_{0}\right)\overset{d}{\rightarrow}N\left(0,V\right).
\]
Standard arguments \parencite[e.g.,][Theorem 3.4]{newey1994large}
yield the influence function representation of GMM estimators 
\[
\sqrt{n}\left(\hat{\theta}-\theta_{0}\right)=-\left(G'WG\right)^{-1}G'W\sqrt{n}\hat{g}\left(\theta_{0}\right)+o_{p}\left(1\right).
\]
Hence, the GMM estimator is, to first order, a linear combination
of normals. \textcite{andrews2017measuring} term these combination
weights the ``sensitivity matrix''
\[
\Lambda:=-\left(G'WG\right)^{-1}G'W,
\]
which summarizes the local influence of the empirical moment conditions
on the parameter estimates. The limiting distribution of GMM under
proper specification can be written in terms of the sensitivity matrix
as 
\[
\sqrt{n}\left(\hat{\theta}-\theta_{0}\right)\overset{d}{\rightarrow}N\left(0,\Lambda V\Lambda'\right).
\]

Note that premultiplying a vector of moments by  $-\Lambda$ amounts
to a weighted least squares regression using the columns of the Jacobian
 $G$ as regressors. It will be useful to consider the corresponding
``hat'' matrix $H:=-G\Lambda.$ Premultiplying a vector of moments
by $H$ gives predicted values of the moment conditions generated
from a generalized least squares (GLS) fit. $H$ is an \textit{oblique}
projection matrix: it is idempotent but not generally symmetric. The
complementary residual maker matrix, which features prominently in
later sections, will be denoted $M:=I_{m}-H$. The sample analogs
of $\Lambda$,  $H$, and $M$ will be denoted by $\hat{\Lambda}=-\left(\hat{G}'\hat{W}\hat{G}\right)^{-1}\hat{G}'\hat{W}$,
 $\hat{H}=-\hat{G}\hat{\Lambda}$, and $\hat{M}=I_{m}-\hat{H}$ respectively.

\subsection{Local and global misspecification\protect\label{subsec:Misspecification-as-contaminatio}}

Describing misspecification requires enlarging the set of data generating
processes (DGPs) under consideration to include those that violate
the researcher's posited moment conditions. Fix a parameter vector
$\theta_{0}\in\Theta$ and denote by $\mathcal{P}_{0}$ the family
of DGPs under which $g\left(\theta_{0}\right)=0$.
For any DGP $P_{0}\in\mathcal{P}_{0}$ one can think of
the estimand $\theta_{0}$ as a functional $\theta_{0}=\theta\left(P_{0}\right)$.
Misspecification arises when the actual DGP lies outside $\mathcal{P}_{0}$,
in which case $g\left(\theta_{0}\right)\neq0$.

\textit{Local misspecification} refers to the case in which $g\left(\theta_{0}\right)=b/\sqrt{n}$ for a fixed vector $b\in\mathbb{R}^{m}$ of specification errors. To motivate this form, consider a family of DGPs that lie within an $n^{-1/2}$ neighborhood of $\mathcal{P}_{0}$. Let $P_{n}$ denote a DGP in this neighborhood, indexed by the sample size, and let $P_{0}\in\mathcal{P}_{0}$ denote the DGP toward which it drifts. Evaluated under $P_{n}$, the moment conditions, their Jacobian, and their covariance matrix all converge to their $P_{0}$ values. 

\textcite{andrews2017measuring} establish, using regularity conditions analogous to those in Section \ref{subsec:Influence-and-sensitivity},
that under the drifting sequence $P_n$, $\sqrt{n}(\hat{\theta}-\theta_{0})$ converges in distribution to $\Lambda$ times the limit in distribution of $\sqrt{n}\hat{g}(\theta_{0})$. With $g(\theta_{0})=b/\sqrt{n}$, write $\sqrt{n}\hat{g}(\theta_{0})=b+\sqrt{n}(\hat{g}(\theta_{0})-g(\theta_{0}))$. A central limit theorem for the second term gives $\sqrt{n}(\hat{g}(\theta_{0})-g(\theta_{0}))\overset{d}{\rightarrow}\varepsilon\sim N(0,V)$, which implies
\[
\sqrt{n}\left(\hat{\theta}-\theta_{0}\right)\overset{d}{\rightarrow}N\left(\Lambda b,\Lambda V\Lambda'\right).
\]
Hence, the limiting distribution of GMM is centered at $\Lambda b$ rather than zero, but its variance is unchanged. If the violations were to shrink at a slower rate (e.g., $g\left(\theta_{0}\right)=b/n^{1/4}$),
the scaled bias would explode relative to the variance, precluding
convergence in distribution. Thus, local misspecification is a tool
for studying specification errors that generate modest first-order
biases in GMM comparable in magnitude to sampling variation. This
sort of misspecification is difficult to detect with a standard
$J$-test: power does not grow with sample size. 

In contrast, under \textit{global misspecification}, the violations do not shrink
with the sample size. Formally, $g\left(\theta_{0}\right)=b$ for a fixed nonzero
vector $b$. In this case, GMM is generally inconsistent for $\theta_{0}$,
converging instead to a pseudo-true value that depends on the weighting
matrix. Likewise, the $J$-test generally has power that grows with the sample size.

Since the first-order bias of $\Lambda b/\sqrt{n}$ under local misspecification
vanishes as $n$ grows large, GMM remains consistent for the estimand $\theta_{0}$. Thus, the local misspecification
framework offers the potential to conduct inference on the true parameter
vector, as opposed to a pseudo-true target, by augmenting
measures of sampling uncertainty to account for specification errors.

\section{An exchangeable model of local misspecification \protect\label{sec:local_misspecification}}

Thus far, the specification errors $b$ have been treated as fixed constants, which is the approach taken in most of the literature following \textcite{andrews2017measuring}. My key departure is to model these errors as exchangeable random variables. Specifically, I assume $g\left(\theta_{0}\right)=b/\sqrt{n}$ for a \emph{random} vector $b\in\mathbb{R}^{m}$.

Conditional on a realization of $b$, the data are generated by a drifting
DGP $P_{n}$ of the kind described in
Section~\ref{subsec:Misspecification-as-contaminatio}. Unconditionally, the data are generated by a mixture of such drifting
DGPs, with mixing distribution given by the law of $b$. In what follows, expectations without a subscript, $E[\cdot]$, treat $b$ as random. In contrast, I will use
$E_{n}[\cdot]$ to denote expectations under $P_{n}$, which condition on
$b$, and $E_{0}[\cdot]$ to denote expectations under the limit DGP $P_0$, which do not depend on $b$.

The scaled sample moment conditions
decompose as 
\[
r_{n}:=\sqrt{n}\hat{g}\left(\theta_{0}\right)=b+\varepsilon_{n},
\]
where $\varepsilon_{n}:=\sqrt{n}\left(\hat{g}\left(\theta_{0}\right)-g\left(\theta_{0}\right)\right)$
captures noise in the moment conditions attributable to sampling variation. The vector $b$ encodes the population-level structure of misspecification
across moment conditions. Exchangeability implies that all entries $b_j$ of $b$ have the same marginal distribution and that all pairs $(b_j, b_k)$,
for $j \neq k$, have the same joint distribution. This is a reasonable perspective
to take when the moment conditions share the same basic form ---
i.e., they are all measured in the same units and involve averages
over similar mathematical objects --- and one does not know which
moment conditions are likely to be misspecified a priori.

The following assumption formalizes this perspective and adds regularity conditions used in later sections. In what follows, $\Vert \cdot \Vert$ denotes the Euclidean norm.
\begin{assumption}[Independence and exchangeability]
\label{assu:random_b}The pair $\left(b,\varepsilon_{n}\right)$
obeys the following conditions:

i) \textup{$b\perp\varepsilon_{n}$ for each $n$, $\Vert E[\varepsilon_{n}]\Vert\to0$, and $\varepsilon_{n}\overset{d}{\rightarrow}\varepsilon\sim N\left(0,V\right)$,}

ii) \textup{the specification errors $b=\left(b_{1},\dots,b_{m}\right)'$ are exchangeable with bounded fourth moments.}
\end{assumption}
Assumption \ref{assu:random_b}.i imposes two conditions on the relationship between
specification errors and the moment-condition noise. First, $b$ and
$\varepsilon_{n}$ are independent at every sample size. This is a
substantive economic restriction that would be violated if, for example,
a data-dependent moment selection rule induced correlation between
$b$ and $\varepsilon_{n}$. Second, the moment-condition noise is asymptotically unbiased and normal with covariance $V$.

Assumption \ref{assu:random_b}.ii is the key exchangeability restriction. The bounded fourth moments assumption will prove useful for the asymptotic analysis of later sections. Together, these conditions ensure that the specification errors share a common marginal mean $\bar \mu:=E[b_j]$ and variance $\bar \sigma^2 := \mathrm{Var}(b_j)$, both of which are finite.

\subsection{Defining hyperparameters}
The marginal mean $\bar \mu$ and variance $\bar \sigma^2$ of the specification errors describe the properties of a hypothetical superpopulation from which the specification errors were sampled. Since the goal of this analysis will be to improve GMM estimates given the specification errors actually faced by the researcher, I will condition on the mean and variance of the realized specification errors. 

Denote the realized mean and variance of the specification errors by
\[
\mu:=\frac{1}{m}\sum_{j=1}^{m}b_{j},\qquad \sigma^{2}:=\frac{1}{m-1}\sum_{j=1}^{m}\left(b_{j}-\mu\right)^{2}.
\]
In subsequent sections, I treat these objects as hyperparameters to be targeted in estimation. The following lemma derives the first two moments of the specification errors $b$ conditional on these hyperparameters. 
\begin{lem}
\label{lem:cond-moments}Assumption \ref{assu:random_b}.ii implies that:
\[
E\left[b\mid\mu,\sigma^{2}\right]=\mu\,1_{m},\qquad \mathrm{Var}\left[b\mid\mu,\sigma^{2}\right]=\sigma^{2}Q,\qquad Q:=I_{m}-\frac{1}{m}1_{m}1_{m}'.
\]
\end{lem}
\begin{proof}
Since $\mu$ and $\sigma^{2}$ are symmetric functions of $b$, conditioning on them preserves exchangeability. Write $b=\mu1_{m}+Qb$ and note that $Qb=b-\mu1_{m}$ is conditionally exchangeable and orthogonal to $1_{m}$. Hence, the conditional mean of $Qb$ must be both proportional to and orthogonal to $1_{m}$. Consequently, its conditional mean equals zero, implying $E[b\mid\mu,\sigma^{2}]=\mu1_{m}$. 

Likewise, the conditional second moment $E[Qb(Qb)'\mid\mu,\sigma^{2}]$ must be a permutation-invariant matrix---i.e., a linear combination of $I_{m}$ and $1_{m}1_{m}'$. Since $1_{m}'Qb=0$, this second moment matrix must also annihilate $1_{m}$, forcing it to be proportional to $Q$. Hence, $\mathrm{Var}[b\mid\mu,\sigma^{2}]=\mathrm{Var}[Qb\mid\mu,\sigma^{2}] \propto Q$. Since $(Qb)'Qb=\sum_{j}(b_{j}-\mu)^{2}=(m-1)\sigma^{2}$ is a function of the conditioning variables, $\mathrm{tr}(\mathrm{Var}[b\mid\mu,\sigma^{2}])=E[(Qb)'Qb\mid\mu,\sigma^{2}]=(m-1)\sigma^{2}$. Direct calculation yields $\mathrm{tr}(Q)=m-1$. Therefore, the constant of proportionality must be $\sigma^2$.
\end{proof}

The conditional pairwise correlation between entries of $b$ is $-1/(m-1)$, which reflects the adding-up constraint imposed by conditioning on $\mu$. Thus, conditioning removes the nuisance parameter governing the marginal correlation between the entries of $b$. While conditioning on $(\mu,\sigma^{2})$ pins down the first two moments of $b$, the centered errors $b-\mu1_{m}$ may remain conditionally dependent at higher moments in ways that depend on additional parameters.

\begin{rem}[Equivariant linear transformations]\label{rem:equivariant}
    Permutation-equivariant linear transformations preserve exchangeability. Hence, if $b$ obeys Assumption~\ref{assu:random_b}.ii, so does $(c_1 I_m+c_2 1_m1_m')b$ for constants $c_1\neq0$ and $c_2$. This transformation maps $\mu$ to $(c_1+mc_2)\mu$ and $\sigma^2$ to $c_1^2\sigma^2$.
\end{rem}

\subsection{Leading term and bias}

Assumption \ref{assu:random_b}, in conjunction with the standard
GMM regularity conditions discussed in Section \ref{sec:GMM-preliminaries},
implies the asymptotic representation 
\[
\sqrt{n}\left(\hat{\theta}-\theta_{0}\right)=\Lambda r_{n}+o_{p}\left(1\right)\overset{d}{\rightarrow}\Lambda r,\qquad r:=b+\varepsilon.
\]
Conditional on the hyperparameters, the leading term $\Lambda r=\Lambda b+\Lambda\varepsilon$ has mean
\[
E\left[\Lambda r\mid\mu,\sigma^{2}\right]=\mu\,\Lambda1_{m}.
\]
When the specification errors have nonzero realized mean, the asymptotic distribution of $\sqrt{n}(\hat{\theta}-\theta_{0})$ is centered at $\mu\Lambda1_{m}$ rather than at zero. Its variance splits into two terms,
\[
\mathrm{Var}\left[\Lambda r\mid\mu,\sigma^{2}\right]=\underbrace{\Lambda V\Lambda'}_{\text{sampling uncertainty}}+\underbrace{\sigma^{2}\Lambda Q\Lambda'}_{\text{specification error}}=\Lambda\Sigma\Lambda',\qquad \Sigma:=V+\sigma^{2}Q.
\]

In Section~\ref{sec:bias-estimation}, I propose estimators that mitigate both the bias and the variance attributable to specification errors by leveraging estimates of $\mu$ and $\sigma^{2}$. To this end, I condition on $(\mu,\sigma^{2})$ throughout the subsequent analysis, treating them as fixed parameters. For brevity, I leave this conditioning implicit wherever possible.

\section{Examples \protect\label{sec:Examples}}

It is useful to walk through some examples of misspecified models that can be fit into the framework of the previous section. My discussion will center on
how Assumption~\ref{assu:random_b}.ii can be violated and
how to design moment conditions under which it can plausibly be satisfied.

\begin{example}[Equal means]
\label{exa:mean}Let  $Y_{n}$ be an $m\times1$ vector of independent
sample means with $\mathrm{Var}\left(Y_{n}\right)=\mathrm{diag}\left(v_{1},\dots,v_{m}\right)/n=V/n$,
where  $n$ is the total sample size. The researcher considers a model
that restricts the $m\times1$ vector of population means $\vartheta_{n}$
to equal a common scalar $\theta$. 

To map this example to the local misspecification framework, let $\vartheta_{n}:=\theta_{0}1_{m}+b/\sqrt{n}$
for some $\theta_{0}\in\Theta$ and specification errors $b\in\mathbb{R}^{m}$.
The $m$-vector of population moment conditions then takes the form
$g\left(\theta\right)=\vartheta_{n}-\theta1_{m}$, which implies the
population Jacobian is $G=-1_{m}$. The
scaled population moment conditions $\sqrt{n}g\left(\theta_{0}\right)=b$
are the specification errors. Since the misspecification is local, GMM is consistent. 

If $b$ is exchangeable with bounded fourth moments, then Assumption~\ref{assu:random_b}.ii
is satisfied. In contrast, Assumption~\ref{assu:random_b}.ii
will be violated if either $E\left[b_{j}\right]$ or $\mathrm{Var}\left(b_{j}\right)$
varies with $j$. For instance, if noisier means have larger deviations
then one might expect $E\left[b_{j}\right]\propto\sqrt{v_{j}}$. In this case, residualizing the means against the noise levels---e.g., by using the methods proposed by \textcite{chen2026empirical}---can potentially restore exchangeability.
\end{example}

The following example considers an overidentified instrumental variables
model in which the plausibility of the exchangeability assumption hinges on moment scaling.
\begin{example}[Exclusion violations]
\label{exa:IV_example}A researcher has cross-sectional data
on  $m>2$ instruments  $Z_{i1},\dots,Z_{im}$ for a scalar endogenous
treatment  $T_{i}$,  where  $i$ indexes individual-level observations.
Collect the instruments into the vector $Z_{i}:=\left(Z_{i1},\dots,Z_{im}\right)^{\prime}$
and denote the outcome by  $Y_{i}$. The triples  $\left\{ Y_{i},T_{i},Z_{i}'\right\} _{i=1}^{n}$
are i.i.d.\ draws from the locally misspecified DGP $P_{n}$.

The researcher's model stipulates the population moment restrictions
\[
g\left(\theta\right)=E_{n}\left[\left(Y_{i}-\theta T_{i}\right)Z_{i}\right]
\]
for scalar  $\theta$, implying $G=-E_{n}\left[T_{i}Z_{i}\right]$.
Suppose instead that for some vector $\gamma\in\mathbb{R}^{m}$ of direct effects
\[
E_{n}\left[\left(Y_{i}-\theta_{0}T_{i}-\frac{1}{\sqrt{n}}Z_{i}'\gamma\right)Z_{i}\right]=0.
\]
Excludability violations of this form were considered by \textcite{conley2012plausibly}, who studied estimation and inference given priors on the direct effects. \textcite{kolesar2015identification} considered the case where the entries of $\gamma$ are orthogonal to the instrument first stages. 

Suppose the direct effects  $\left\{ \gamma_{j}\right\} _{j=1}^{m}$
are exchangeable. The associated scaled
population moment conditions satisfy 
\[
\sqrt{n}g\left(\theta_{0}\right)=E_{n}\left[Z_{i}Z_{i}'\right]\gamma = E_{0}\left[Z_{i}Z_{i}'\right]\gamma=b.
\]
In this case, $b$ is generally not exchangeable unless $E_{0}\left[Z_{i}Z_{i}'\right]$ takes the permutation-equivariant form described in Remark \ref{rem:equivariant}. Since $E_{0}\left[Z_{i}Z_{i}'\right]$ equals $\mathrm{Var}_{0}(Z_i)+E_{0}[Z_i]E_{0}[Z_i]'$, this effectively requires the instruments to share a common mean and exhibit an exchangeable (equicorrelated and homoscedastic) variance structure.

This difficulty can be circumvented by working with rescaled moment
conditions taking the form 
\[
g\left(\theta\right)=E_{n}\left[Z_{i}Z_{i}'\right]^{-1}E_{n}\left[\left(Y_{i}-\theta T_{i}\right)Z_{i}\right]=\delta - \theta \pi,
\]
for $\delta:=E_{n}\left[Z_{i}Z_{i}'\right]^{-1}E_{n}\left[Z_{i}Y_{i}\right]$ the $m$-vector of population reduced form coefficients and $\pi:=E_{n}\left[Z_{i}Z_{i}'\right]^{-1}E_{n}\left[Z_{i}T_{i}\right]$ the corresponding vector of first-stage coefficients. These rescaled moments have $b=\gamma$, satisfying Assumption \ref{assu:random_b}.ii whenever $\gamma$ is exchangeable with bounded fourth moments.
\end{example}

Because excludability violations are often cited as a first-order
concern in applied work with instrumental variables, I will focus
on variants of Example \ref{exa:IV_example} in
the applications of Sections \ref{sec:monte-carlo} and \ref{sec:empirical}. As the following example illustrates, however, the methods developed in this paper are equally applicable to nonlinear models. 

\begin{example}[A binary choice restriction]\label{exa:binarychoice}
    Consider the binary choice model $\Pr\left(Y_{i}=1|X_{i}\right)=F\left(X_{i}'\beta\right),$ where $\beta\in\mathbb{R}^{m}$ and the link function $F:\mathbb{R}\to(0,1)$ is strictly increasing and continuously differentiable with derivative $f:\mathbb{R}\to\mathbb{R}_{>0}$. Maximum likelihood estimation of such a model is equivalent to method of moments on the $m$-vector of log-likelihood scores, the population versions of which can be written
    \[
    U(\beta)= E_{n}\left[\frac{Y_{i}-F\left(X_{i}'\beta\right)}{F\left(X_{i}'\beta\right)\left[1-F\left(X_{i}'\beta\right)\right]}f\left(X_{i}'\beta\right)X_{i}\right].
    \]
 
    Now suppose a researcher imposes the coefficient restriction $\beta=\theta 1_{m}$  for $\theta\in\mathbb{R}$. For example, in the discrete choice experiment of \textcite{mas2019labor}, job applicants were randomly presented with one of $m$ pairs of job options involving distinct earnings and hours bundles. The outcome $Y_i$ measured whether individual $i$ chose the bundle with longer hours, while the vector $X_i$ measured the interaction between an indicator for the assigned menu and the implicit wage of the additional hours required by the longer option. In this case, the restriction would require the marginal rate of substitution between leisure and consumption to be equal across the $m$ distinct choice menus. 
    
    It is natural to estimate this restricted model by GMM using the moment conditions
    \[
    U\left(\theta 1_{m}\right)=E_{n}\left[\frac{Y_{i}-F\left(\theta X_{i}'1_{m}\right)}{F\left(\theta X_{i}'1_{m}\right)\left[1-F\left(\theta X_{i}'1_{m}\right)\right]}f\left(\theta X_{i}'1_{m}\right)X_{i}\right].
    \]
    Suppose however that the true choice probability obeys a drifting DGP
    \[
        \Pr\left(Y_{i}=1|X_{i}\right)=F\left(\theta_{0}X_{i}'1_{m}\right)+\frac{1}{\sqrt{n}}X_{i}'b\cdot F\left(\theta_{0}X_{i}'1_{m}\right)\left[1-F\left(\theta_{0}X_{i}'1_{m}\right)\right],
    \]
    where $b\in\mathbb{R}^{m}$. When $X_i'b$ has bounded support (e.g., the menu indicators described above) this expression is almost surely a valid probability for large enough $n$. By iterated expectations,
    \[
        U\left(\theta_{0}1_{m}\right)=\frac{1}{\sqrt{n}}E_{n}\left[\left(X_{i}'b\right)f\left(\theta_0 X_{i}'1_{m}\right)X_{i}\right]=\frac{1}{\sqrt{n}}E_{n}\left[f\left(\theta_0 X_{i}'1_{m}\right) X_{i}X_{i}'\right]b.
    \]
    
    Since $f>0$, when $E_{n}[X_{i}X_{i}']$ is positive definite, the weighted matrix $E_{n}[f(\theta X_{i}'1_{m})X_{i}X_{i}']$ is also invertible. One can then work with the rescaled moment condition
    \[
     g\left(\theta\right)=E_{n}\left[f\left(\theta X_{i}'1_{m}\right) X_{i}X_{i}'\right]^{-1}U(\theta 1_{m}),
    \]
    which yields $\sqrt{n} g\left(\theta_{0}\right)=b$.
    When $b$ is exchangeable with bounded fourth moments, Assumption~\ref{assu:random_b}.ii is satisfied.
\end{example}

\section{Identification of hyperparameters \protect\label{sec:Identification-of-hyperparameter}}

This section studies identification of the hyperparameters $\left(\mu,\sigma^{2}\right)$, which are treated as fixed (i.e., conditioned on) throughout.
Recall that the GMM estimator sets $p$ linear combinations of the
sample moments equal to zero. Assumption~\ref{assu:random_b}
implies the scaled moment conditions can be written 
\[
\hat{r}:=\sqrt{n}\hat{g}\left(\hat{\theta}\right)=Mr_{n}+o_{p}\left(1\right)\overset{d}{\rightarrow}Mr,
\]
where the sample and population residual maker matrices were defined
in Section \ref{subsec:Influence-and-sensitivity}. Intuitively, the
process of fitting GMM masks $p$ linear combinations of the  $b_{j}$'s
that lie in the null space of  $M$. The exchangeability assumption
allows one to infer moments of these hidden linear combinations
from the  $m-p$ identified linear combinations that lie in  $M$'s
column space. 

\subsection{Moment equations}

Assumption~\ref{assu:random_b} implies 
\[
E\left[Mr\right]=\mu M1_{m},\quad \mathrm{Var}\left[Mr\right]=M\Sigma M'.
\]
Provided $\left\Vert M1_{m}\right\Vert ^{2}=\left(M1_{m}\right)'\left(M1_{m}\right)>0$,
it follows that 
\begin{equation}
E\left[\frac{\left(M1_{m}\right)'Mr}{\left\Vert M1_{m}\right\Vert ^{2}}\right]=\mu,\quad E\left[\frac{\left\Vert Mr-\mu M1_{m}\right\Vert ^{2}-\mathrm{tr}\left(MVM'\right)}{\mathrm{tr}\left(M'MQ\right)}\right]=\sigma^{2}.\label{eq:identification}
\end{equation}
Note that $\left\Vert M1_{m}\right\Vert ^{-2}\left(M1_{m}\right)'Mr$
is the coefficient from an OLS regression of $Mr$ on $M1_{m}$. The
ratio $\left[\left\Vert Mr-\mu M1_{m}\right\Vert ^{2}-\mathrm{tr}\left(MVM'\right)\right]/\mathrm{tr}\left(M'MQ\right)$
is a degrees-of-freedom adjusted estimator of the variance of the
specification errors. It subtracts the expected sampling noise $\mathrm{tr}\left(MVM'\right)$
from the residual sum of squares, then divides by the effective residual degrees of freedom $\mathrm{tr}\left(M'MQ\right)=\mathrm{tr}\left(MM'\right)-\left\Vert M1_{m}\right\Vert ^{2}/m$.
When $m-p\ge2$, this denominator is strictly positive and the ratio is well-defined. These population moment equations motivate the hyperparameter
estimators that will be proposed in Section \ref{sec:estimation},
which replace $Mr$ with the sample residuals $\hat{r}$ and population
matrices with their sample analogs.

\subsection{Identification failures \protect\label{subsec:Identification-failures}}

Some overidentified models imply $M1_{m}=0$, leading to non-identification
of  $\mu$. This problem arises whenever  $1_{m}$ lies in the column
space of $G$.

\subsubsection{Constant Jacobians}

Recall that $G=-1_{m}$ in Example \ref{exa:mean}, implying the mean
bias loads directly onto the parameter of interest: $E\left[\Lambda r\right]=\mu$.
Here the hyperparameter  $\mu$ is absorbed by the GMM fitting
process because it is collinear with the parameter being estimated.
Consequently, the $J$-test has no power to detect non-zero values
of $\mu$ \parencite{newey1985gmm}.

The same issue arises in the binary choice model of Example~\ref{exa:binarychoice} when the link function $F$ is logistic. Differentiating the rescaled moments at $\theta_{0}$ yields the Jacobian
\[
G=-E_{n}\left[f\left(\theta_{0}X_{i}'1_{m}\right)X_{i}X_{i}'\right]^{-1}E_{n}\left[\frac{f\left(\theta_{0}X_{i}'1_{m}\right)^{2}}{F\left(\theta_{0}X_{i}'1_{m}\right)\left[1-F\left(\theta_{0}X_{i}'1_{m}\right)\right]}X_{i}X_{i}'\right]1_{m},
\]
up to terms that vanish under the drifting DGP. For the logistic link, $f=F\left(1-F\right)$, which yields $G=-1_{m}$. Hence, the mean $\mu$ is unidentified in the logit model. In contrast, link functions for which the ratio $f/[F(1-F)]$ varies with the index generically produce a Jacobian with non-constant entries, preserving identification of $\mu$.

\subsubsection{Intercepts in linear models \protect\label{subsec:Intercepts-in-linear}}

In overidentified linear models, $\mu$ typically remains identified
even when an intercept is included. Consider a $p=2$ extension of Example~\ref{exa:IV_example} where  $\theta=\left(\alpha,\beta\right)'$,
with  $\alpha$ capturing an intercept and $\beta$ a slope. Suppose also that
an intercept is added to the  $m$ fundamental instruments  $Z_{i}$. Thus,
the rescaled moment conditions take the form  $g\left(\theta\right)=S^{-1} E_{n}\left[\left(Y_{i}-\left(1,T_{i}\right)\theta\right)\left(1,Z_{i}'\right)'\right]$, where $S:=E_{n}[\left(1,Z_{i}'\right)'\left(1,Z_{i}'\right)]$ is the second moment matrix of the augmented instruments.
In this case 
\[
G=-S^{-1}\left[\begin{array}{cc}
1 & E_{n}\left[T_{i}\right]\\
E_{n}\left[Z_{i}\right] & E_{n}\left[T_{i}Z_{i}\right]
\end{array}\right] = - \pi,
\]
where $\pi$ is an $(m+1) \times 2$ matrix of first-stage coefficients. The first column of $\pi$ is necessarily a unit vector $(1, 0_m')'$ capturing the first-stage coefficients for the constant vector. Its second column collects the population coefficients from regressing $T_i$ on the augmented instruments: an intercept $\pi_{12}$ and slopes $\pi_{22},\dots,\pi_{m+1,2}$ on the $m$ fundamental instruments. Identification of $\mu$ fails only when these $m$ slopes share a common nonzero value. 

In contrast, identification fails in an extension of Example~\ref{exa:IV_example}
with $\theta=\left(\alpha,\beta\right)'$ and mutually exclusive binary
instruments  $Z_{ij}\in\left\{ 0,1\right\} $ obeying $\sum_{j=1}^{m}Z_{ij}=1$,
a design that arises often in the recent literature on judge IV \parencite{chyn2025examiner}.
To understand the problem, consider the rescaled moment conditions
\[
g\left(\theta\right)=E_{n}\left[Z_{i}Z_{i}'\right]^{-1}E_{n}\left[\left(Y_{i}-\left(1,T_{i}\right)\theta\right)Z_{i}\right].
\]
Letting  $q=E_{n}\left[Z_{i}\right]$, the Jacobian takes the form
\[
G=-\mathrm{diag}\left(q\right)^{-1}\left(E_{n}\left[Z_{i}\right],E_{n}\left[T_{i}Z_{i}\right]\right)=-\left(1_{m},\mathrm{diag}\left(q\right)^{-1}E_{n}\left[T_{i}Z_{i}\right]\right).
\]
From the first column,  $1_{m}$ lies in the span of $G$. Therefore
$M1_{m}=0$, implying $\mu$ is not identified. This failure stems fundamentally from the adding-up constraint $\sum_{j}Z_{ij}=1$: a direct effect common to all judges is constant and is absorbed by the intercept. 

\subsubsection{Which parameters are biased?}

When $\mu$ is not identified, GMM's leading bias term $\mu\Lambda1_{m}$
cannot be consistently estimated. Non-identification of $\mu$ therefore
implies that some linear combination of the entries in  $\hat{\theta}$
has a leading bias of order $O\left(n^{-1/2}\right)$ that cannot
be removed. However, other linear combinations may have no leading
bias. In the judge-IV example,  $\mu$
loads directly on the intercept parameter  $\alpha$. But  $\beta$
is insensitive to  $\mu$, yielding no leading-term bias in the GMM
point estimate $\hat{\beta}$: 
\[
E\left[\Lambda_{\beta}r\right]=\mu\Lambda_{\beta}1_{m}=-\mu\Lambda_{\beta}G_{\alpha}=0.
\]
Here, $\Lambda_{\beta}$ is the row of $\Lambda$ corresponding to
$\beta$ and $G_{\alpha}=-1_{m}$ is the intercept column of $G$.
The final equality uses $\Lambda_{\beta}G_{\alpha}=0$, which follows
from $\Lambda G=-I_{p}$. 

Though GMM exhibits no leading-term bias for $\beta$ in this example,
it does exhibit conditional bias $E\left[\Lambda_{\beta}r\mid b\right]=\Lambda_{\beta}b$.
Averaging over the exchangeable specification errors, this conditional
bias contributes variance  $\sigma^{2}\Lambda_{\beta}\Lambda_{\beta}'$
to the asymptotic distribution of $\hat{\beta}$: 
\[
\mathrm{Var}\left[\Lambda_{\beta}r\right]=\underbrace{\Lambda_{\beta}V\Lambda_{\beta}'}_{\text{sampling error}}+\underset{\text{specification error}}{\underbrace{\sigma^{2}\Lambda_{\beta}\Lambda_{\beta}'}}.
\]
The specification-error term is $\sigma^{2}\Lambda_{\beta}Q\Lambda_{\beta}'$, which reduces to $\sigma^{2}\Lambda_{\beta}\Lambda_{\beta}'$ because $\Lambda_{\beta}1_{m}=0$.
In this scenario, the job of ``repairing'' GMM simplifies to reducing
the noise in  $\hat{\beta}$ attributable to specification error using
estimates of  $\sigma^{2}$, a task that can be accomplished via empirical
Bayes shrinkage.

\section{The limits of differencing \protect\label{sec:A-differencing-approach}}

Before moving on to empirical Bayes approaches that rely on estimating
the hyperparameters $\left(\mu,\sigma^{2}\right)$, it is useful to
consider what can be achieved when the hyperparameters are treated
as nuisance parameters. A natural strategy for dealing with the exchangeable
structure of the errors in the moment conditions is to difference
their shared mean $\mu$ out. As in Section \ref{sec:Identification-of-hyperparameter},
distributional statements throughout this section are conditional
on $\left(\mu,\sigma^{2}\right)$.

\subsection{Identification}

Let $D$ be an $\left(m-1\right)\times m$ differencing matrix obeying
$D1_{m}=0$ and $DD'=I_{m-1}$ (e.g., a Helmert contrast matrix).
The differenced sample moment conditions $D\hat{g}\left(\theta\right)$,
when evaluated at $\theta_{0}$, obey 
\[
\sqrt{n}D\hat{g}\left(\theta_{0}\right)\overset{d}{\rightarrow}Dr=Db+D\varepsilon.
\]
Under Assumption \ref{assu:random_b}
\[
E\left[Dr\right]=\mu D1_{m}=0,\quad \mathrm{Var}\left(Dr\right)=\sigma^{2}DD'+DVD'=\sigma^{2}I_{m-1}+DVD'.
\]
Thus, differencing removes the mean bias of the moment conditions and consequently the leading-term bias of the resulting GMM estimator, at the potential cost of inflating its variance. In fact, when $b$
is exchangeable, the prior expectation of the differenced violations
vanishes regardless of their scale: $E[Db]=DE[b]=0$. Hence, differencing
yields moment conditions whose superpopulation mean is zero even
under global misspecification of the sort described in Section \ref{subsec:Misspecification-as-contaminatio}.

Mirroring the discussion in Section \ref{subsec:Identification-failures},
one concern is that differencing may undermine identification of $\theta_{0}$
itself (e.g., by leading an intercept to drop out of the original
moment conditions). This concern is ruled out by the condition $M1_{m}\neq0$,
which ensures the differenced moment Jacobian $DG$ has full column
rank $p$. Under this rank condition, $\theta_{0}$ is locally identified from the differenced moments $Dg\left(\theta\right)$. These moments have mean zero at $\theta_{0}$, since $E[Dr]=\mu D1_{m}=0$. Global identification additionally requires that $\theta_{0}$ be the unique such solution.

\subsection{Leading term}

Let $\hat{\theta}_{D}$ denote the estimator obtained from conducting
GMM on the differenced moment conditions with $\left(m-1\right)\times\left(m-1\right)$
weighting matrix $W_{D}$. The sensitivity matrix of the differenced
GMM estimator is $\Lambda_{D}:=-\left(G'D'W_{D}DG\right)^{-1}G'D'W_{D}$.
The arguments above imply that when $M1_{m}\neq0$, 
\[
\sqrt{n}\left(\hat{\theta}_{D}-\theta_{0}\right)=\Lambda_{D}Dr_{n}+o_{p}\left(1\right).
\]
Since $E\left[Dr\right]=0$, the leading-term bias has been removed.
However, the differenced estimator need not converge to a normal limit as
$n$ grows large. 

To see why normality is not ensured in the asymptotic limit, observe
that 
\[
\Lambda_{D}Dr=\underbrace{\Lambda_{D}D\varepsilon}_{\text{sampling variability}}+\underbrace{\Lambda_{D}Db}_{\text{specification errors}}.
\]
The first term, capturing sampling variability, is $N\left(0,\Lambda_{D}DVD'\Lambda_{D}'\right)$.
In contrast, the second term depends on the unknown exchangeable distribution
of the specification errors $b$.

\subsection{Limitations}

Recall that the asymptotics thus far have let $n$ grow large with
the number of moment conditions $m$ fixed. When $m$ is small, the
linear combination $\Lambda_{D}Db$ can be far from normal.
Assumption \ref{assu:random_b} guarantees that this term has
variance $\sigma^{2}\Lambda_{D}\Lambda_{D}'$, that four of its moments
exist, and that it is independent of the first term. However, a large
family of exchangeable distributions satisfies these restrictions. 

In the next section, I develop an asymptotic framework in which the number
of moment conditions grows with the sample size $n$, which allows
me to consistently estimate the hyperparameters $\left(\mu,\sigma^{2}\right)$
at rates suitable for inference on $\theta_{0}$. Estimating the first-step
hyperparameters can provide two advantages beyond inference. The first is
that knowledge of the hyperparameters can enable additional reductions
in the variability of the estimator. In particular, knowledge of $\sigma^{2}$
can be used to develop shrinkage estimators with certain efficiency
advantages. 

Second, the hyperparameters are often interesting in their own right,
providing economically interpretable information about the nature
of the specification errors. As will be illustrated by the empirical
application of Section \ref{sec:empirical}, it is reasonable for
researchers to scrutinize the hyperparameter estimates as a way of
gauging the plausibility of the exchangeability framework on which
the estimators are predicated.

\section{\protect\label{sec:estimation}Estimating hyperparameters}

This section considers how to estimate the hyperparameters  $\left(\mu,\sigma^{2}\right)$
in an environment where  $m$ and $n$ grow jointly, while  $p$ remains
fixed. Intuitively, these results leverage the high degree of overidentification
 $m-p$ found in many empirical studies, which provide opportunities
to detect overdispersion in the empirical moment conditions attributable
to misspecification.

The results in this section are established under a joint asymptotic
framework in which $m$ grows more slowly than $\sqrt{n}$. This asymptotic regime,
which was also considered by \textcite{newey1990efficient} and \textcite{chernozhukov2025plausible},
allows me to state explicit rate conditions under which the feasible
estimators approximate their oracle counterparts. I collect the regularity
conditions used throughout below, using $\left\Vert \cdot\right\Vert _{2}$
to denote the spectral operator norm.
\begin{assumption}
\label{assu:regime}(Asymptotic Regime and Regularity). The following
conditions hold as  $n\to\infty$:

i)  $m=m(n)\to\infty$ with  $m^{2}/n\to0$,  and  $p$ is fixed.

ii) $\|W\|_{2}=1$ and $\|V\|_{2}=1$ (normalizations). The row sums
of $W$ and $V$ are uniformly bounded: $\max_{1\le j\le m}\sum_{k}|W_{jk}|=O(1)$
and $\max_{1\le j\le m}\sum_{k}|V_{jk}|=O(1)$.

iii) The first-stage estimators satisfy, uniformly in  $m$: $\|\hat{M}-M\|_{2}=O_{p}(\sqrt{m/n})$,
 $\|\hat{V}-V\|_{2}=O_{p}(\sqrt{m/n})$,  $\|\hat{\Lambda}-\Lambda\|_{2}=O_{p}(1/\sqrt{n})$,
and the scaled linearization remainder $\sqrt{n}\,[\hat{g}(\hat{\theta})-\hat{g}(\theta_{0})-\hat{G}(\hat{\theta}-\theta_{0})]$ has norm $O_{p}(\sqrt{m/n})$.

iv) $MVM'$ has rank $m-p$, and its smallest nonzero eigenvalue, denoted $\lambda_{\min}^{+}(MVM')$, satisfies $\lambda_{\min}^{+}(MVM')\cdot n/m\to\infty$.

v) The entries of $G$ are uniformly bounded, $\sup_{j,k}|G_{jk}|<\infty$, and $\left\Vert M\right\Vert _{2}=O(1)$.

vi) $m^{-1}\left\Vert M1_{m}\right\Vert ^{2}\to\kappa\in(0,\infty)$.

vii) $V_{n}:=\mathrm{Var}(\varepsilon_{n})$ exists with $\left\Vert V_{n}-V\right\Vert _{2}=O(\sqrt{m/n})$,
and there is a constant $C$ not depending on $n$ such that $\mathrm{Var}(\varepsilon_{n}'M'M\varepsilon_{n})\le C\cdot\mathrm{tr}((M'MV)^{2})$ for all $n$.
\end{assumption}
Condition (i) specifies that the number of moment conditions grows
with the sample size but at a rate slower than $\sqrt{n}$, ensuring
that first-stage estimation errors remain asymptotically negligible.
When moment conditions are added at a faster rate, a many-instruments
bias emerges that compromises both GMM and the sensitivity-based
corrections that will be proposed. Thus, the asymptotic results that
follow apply to use cases where models are moderately overidentified,
but not so complex that techniques for finite-dimensional modeling
break down. In Section \ref{sec:empirical}, I explore a variant of
a specification considered by \textcite{angrist1991does} that fits
these criteria. 

The first part of condition (ii) is without loss of generality: GMM
is invariant to rescaling of either the moment conditions $\hat g(\theta)$ or the weighting matrix $W$ by a scalar. These scalars set the operator norms to one but leave the condition numbers unchanged. The order of $V$'s smallest eigenvalue $\lambda_{\min}(V)$, and hence of $\lambda_{\min}^+(MVM')$, is a feature of the design, not the normalization. Accordingly, the eigenvalue conditions below are stated as rates. The bounded row-sum conditions hold
automatically when $W$ and $V$ are diagonal and more generally
when each of their rows has sparse or rapidly decaying entries. For
$V$, this condition controls pairwise dependence among the moment errors.

Condition (iii) requires that the first-stage estimators converge
uniformly in $m$. The rate $O_{p}(\sqrt{m/n})$ for the $m\times m$ matrix $\hat{V}$ is delivered
by standard concentration inequalities under sub-Gaussian tails or
bounded support. Appendix~\ref{app:verify} verifies the parametric rate required of $\hat{\Lambda}$ in a cell-IV design, where the rows of $\Lambda$ average the first-stage errors across the $m$ moments. The rate required of $\hat{M}$ follows from the identity $M=I_{m}+G\Lambda$. When $\hat{G}$ converges at the rate of $\hat{V}$, the error $\hat{M}-M=(\hat{G}-G)\hat{\Lambda}+G(\hat{\Lambda}-\Lambda)$ is $O_{p}(\sqrt{m/n})$ in operator norm, via $\|G\|_{2}=O(\sqrt{m})$ from condition (v). Finally, the linearization remainder condition ensures
that the scaled residuals $\hat{r}=\sqrt{n}\hat{g}(\hat{\theta})$
are well approximated by $\hat{M}\sqrt{n}\hat{g}(\theta_{0})$, which holds automatically for linear moment conditions
and follows from uniformly bounded Hessians in the nonlinear case.

Condition (iv) keeps the residualized moment covariance nondegenerate as the number of
moments grows. The rank-$(m-p)$ requirement makes $MVM'$ positive definite on the residualized subspace $\mathrm{col}(M)$ but does not require $\Sigma$ to be positive definite elsewhere. Because $M\Sigma M'=MVM'+\sigma^{2}MQM'$ and the second term is positive semidefinite, $M\Sigma M'$ has the same rank as $MVM'$ and its smallest nonzero eigenvalue is at least as large, for every $\sigma^{2}\ge0$, including $\sigma^{2}=0$.

Condition (v) requires uniformly bounded entries of the Jacobian $G$,
ensuring that individual moment conditions do not exert disproportionate
influence on the parameter estimates, and additionally requires that
the projection matrix $M$ act as a bounded operator.
Condition (vi) requires that $1_{m}$ not lie asymptotically in the
column space of $G$, ensuring identification of the mean specification
error $\mu$.

Condition (vii) has two clauses. The first requires that the finite-sample covariance $V_{n}=\mathrm{Var}(\varepsilon_{n})$ exists and converges to its limit $V$ in operator norm at the rate $\sqrt{m/n}$ of the first-stage estimators in condition (iii). The second clause bounds the variance of the quadratic form $\varepsilon_{n}'M'M\varepsilon_{n}$ that arises in the consistency argument for $\widehat{\sigma^{2}}$ (Lemma~\ref{lem:sigma-consistency}). Under i.i.d.\ moment contributions with finite fourth moments this bound follows from standard moment computations on quadratic forms. With weakly dependent or martingale-difference contributions it follows from analogous moment-cumulant bounds.

In the remainder of this section, I will establish consistency of
estimators of the hyperparameters $\left(\mu,\sigma^{2}\right)$, which are treated as fixed throughout.
When Assumption \ref{assu:random_b} is invoked jointly with
Assumption \ref{assu:regime}, it is understood to hold at each $m$
in the sequence $m\left(n\right)$, with a fourth-moment bound that does not depend on $m$. The consistency arguments in this
section rely on the quadratic-form concentration of Lemma~\ref{lem:perm-quadratic}, while the empirical Bayes results of Section \ref{sec:bias-estimation} rely on the
consistency of $\widehat{\sigma^{2}}$ established below. In Section
\ref{sec:distribution-theory}, I leverage Assumption \ref{assu:random_b}
to establish asymptotic normality of both estimators.

\subsection{Estimating $\mu$ \protect\label{subsec:Estimating-mu}}

Suppose that $\left\Vert \hat{M}1_{m}\right\Vert>0$. As discussed in Section \ref{subsec:Identification-failures}, this
condition fails if a linear combination of the columns of $\hat{G}$
equals $1_{m}$. A plug-in estimator of $\mu$, motivated by \eqref{eq:identification},
is the regression coefficient 
\[
\hat{\mu}=\left\Vert \hat{M}1_{m}\right\Vert ^{-2} \left(\hat{M}1_{m}\right)'\hat{r}.
\]

Under Assumption~\ref{assu:regime}.iii, the first-stage estimator  $\hat{M}$
converges to its population counterpart, while Assumption~\ref{assu:regime}.vi guarantees that $\left\Vert M1_{m}\right\Vert ^{2}$ is of order $m$ (i.e., that $\kappa>0$). The
following Lemma establishes the convergence rate of $\hat \mu$ when $m$ and  $n$ grow large.
\begin{lem}
\label{lem:mu_consistency}Under Assumptions~\ref{assu:random_b}
and \ref{assu:regime},
\[
\hat{\mu}-\mu=O_{p}\left(\frac{1}{\sqrt{m}}\right)+O_{p}\left(\sqrt{\frac{m}{n}}\right).
\]
\end{lem}
\begin{proof}
See appendix.
\end{proof}
The first term captures the error of an infeasible estimator that
uses the population $M$ and residual $r_{n}$, rather than $\hat{M}$ and $\hat{r}$, to estimate $\mu$. This term shrinks
at a $1/\sqrt{m}$ rate because it averages across $m$ moment conditions.
The second term captures the feasibility cost of replacing these population
quantities with their sample analogs, primarily the replacement of $M$ by $\hat{M}$, which converges to
$M$ in operator norm at rate $\sqrt{m/n}$ by Assumption \ref{assu:regime}.iii.
Thus,  $\hat{\mu}$ is consistent for $\mu$ when $m$ grows more
slowly than  $n$. Under Assumption \ref{assu:regime}.i, $m^2/n\to0$, implying the second term is dominated by the first.
\begin{rem}[Weak identification of $\mu$]
\label{rem:strong_id}The requirement $\kappa>0$ in Assumption~\ref{assu:regime}.vi places $\hat{\mu}$ in an asymptotic regime where its target $\mu$ is \emph{strongly identified}. An interesting question for future work is how to characterize the behavior of $\hat{\mu}$ in the setting where $\Vert M1_m \Vert^2=O\left(1\right)$, under which the hyperparameter is \emph{weakly identified}.
\end{rem}

\subsection{Estimating  $\sigma^{2}$}

From \eqref{eq:identification},  a plug-in estimator of  $\sigma^{2}$
takes the form 
\[
\widetilde{\sigma^{2}}=\frac{\left\Vert \hat{r}-\hat{\mu}\hat{M}1_{m}\right\Vert ^{2}-\mathrm{tr}\left(\hat{M}\hat{V}\hat{M}'\right)}{\mathrm{tr}\left(\hat{M}'\hat{M}Q\right)},
\]
 where $\hat{V}$ is an estimator of $V$ satisfying condition (iii) of Assumption~\ref{assu:regime}. Centering on the fitted $\hat{\mu}$ leaves $\widetilde{\sigma^{2}}$ with a small downward bias of order $O(1/m)$, negligible relative to its sampling error. To ensure the
estimate is non-negative, one can rely on the truncated estimator
\[
\widehat{\sigma^{2}}=\max\left\{ 0,\widetilde{\sigma^{2}}\right\} .
\]
The max operator introduces a mild upward bias near $\sigma^{2}=0$ of the same order as the estimator's sampling error.

\subsubsection{Connection to the J-statistic}

The estimator  $\widetilde{\sigma^{2}}$ bears a close connection
to the usual  $J$-statistic. Assuming the variance estimator $\hat{V}$
is nonsingular, this statistic can be written 
\[
\hat{J}=\hat{r}'\hat{V}^{-1}\hat{r},
\]
 where $\hat{r}=\sqrt{n}\hat{g}\left(\hat{\theta}\right)$ and $\hat{\theta}$
is constructed using $\hat{W}=\hat{V}^{-1}$. 

In the special case where  $\hat{W}=\hat{V}=I$,  one can write 
\[
\widetilde{\sigma^{2}}=\frac{\hat{J}-\left(m-p\right)-\hat{\mu}^{\,2}\,1_{m}'\hat{M}'\hat{M}1_{m}}{m-p-\tfrac{1}{m}1_{m}'\hat{M}'\hat{M}1_{m}}.
\]
Provided that the true moment covariance is also the identity ($V=I$), $m-p$ is the large-$n$ expected value of $\hat{J}$ under correct
specification. The term $-\hat{\mu}^{\,2}\,1_{m}'\hat{M}'\hat{M}1_{m}$
subtracts the additional noncentrality in $\hat{J}$ that would arise if every specification error equaled
$\hat{\mu}$. Hence, in this case, $\widetilde{\sigma^{2}}$ captures overdispersion in the moment conditions attributable to specification
error.

Because  $\widetilde{\sigma^{2}}$ weights moment deviations equally,
this intuition does not extend directly to other choices of  $\hat{W}$
and  $\hat{V}$, as the  $J$-statistic weights moments by their precision.
However, if one restricts attention to the optimally weighted case
in which  $\hat{W}=\hat{V}^{-1}$,  the following alternative estimator
of  $\sigma^{2}$ is also consistent, provided $\lambda_{\min}(V)$ is bounded away from zero:
\[
\widetilde{\sigma_{J}^{2}}:=\frac{\hat{J}-\left(m-p\right)-\hat{\mu}_W^{\,2}\,1_{m}'\hat{M}'\hat{V}^{-1}\hat{M}1_{m}}{\mathrm{tr}(\hat{V}^{-1}\hat{M}\hat{M}')-\tfrac{1}{m}1_{m}'\hat{M}'\hat{V}^{-1}\hat{M}1_{m}},
\]
where $\hat \mu_W:=(1_{m}' \hat{M}' \hat W \hat{M}1_{m}) ^{-1} 1_{m}' \hat{M}' \hat W \hat{r} $ is the $\hat W$-weighted analog of $\hat \mu$. Note that when  $\hat{V}=I$,  $\widetilde{\sigma_{J}^{2}}=\widetilde{\sigma^{2}}$. As with $\widetilde{\sigma^{2}}$, negative values are possible and the truncation $\max\{0,\widetilde{\sigma_{J}^{2}}\}$ would need to be applied in practice.

While the connection
of  $\widetilde{\sigma_{J}^{2}}$ to optimally-weighted GMM is appealing,
weighting squared moment deviations by  $\hat{V}^{-1}$ need not yield
uniformly more precise estimates of  $\sigma^{2}$ than  $\widetilde{\sigma^{2}}$.
Moreover, inverse variance weighting can yield finite-sample biases
due to correlation between the weights and moment conditions \parencite{altonji1996small,newey2004higher}.

\subsubsection{Consistency}

The following Lemma establishes the convergence rate of the variance
estimator under joint asymptotics.
\begin{lem}
\label{lem:sigma-consistency}Under Assumptions~\ref{assu:random_b}
and \ref{assu:regime}, 
\[
\widehat{\sigma^{2}}-\sigma^{2}=O_{p}\left(\frac{1}{\sqrt{m}}\right)+O_{p}\left(\sqrt{\frac{m}{n}}\right).
\]
\end{lem}
\begin{proof}
See appendix.
\end{proof}
As in Lemma \ref{lem:mu_consistency}, the rate is decomposable
into two terms. The first term is attributable to the performance
of an infeasible estimator using the oracle estimate of $\mu$ described after Lemma~\ref{lem:mu_consistency} along with the population
$M$ and $V$ matrices. The second term captures the first-stage estimation
errors $\hat{M}-M$ and $\hat{V}-V$, along with the deviation of $\hat{\mu}$ from its oracle counterpart. 

Write $\hat{\Sigma}-\Sigma=(\hat{V}-V)+(\widehat{\sigma^{2}}-\sigma^{2})Q$. Assumption~\ref{assu:regime}.iii
controls the first term and Lemma~\ref{lem:sigma-consistency} the second. Both vanish as $m\to\infty$ and $m^{2}/n\to0$. Since $\|Q\|_{2}=1$, it follows immediately that $\hat{\Sigma}=\hat{V}+\widehat{\sigma^{2}}Q$ is consistent
for $\Sigma$.
\begin{rem}[Hyperparameter consistency with $m=o(n)$]
\label{rem:weaker-rate}While Assumption~\ref{assu:regime} stipulates
that  $m^{2}/n\to0$, the rates in Lemmas~\ref{lem:mu_consistency}
and~\ref{lem:sigma-consistency} require only  $m/n\to0$, under
which both the oracle error  $O_{p}(1/\sqrt{m})$ and the first-stage
error  $O_{p}(\sqrt{m/n})$ vanish. The stronger condition  $m^{2}/n\to0$
is driven by the  $m\times m$ matrix estimation problems that arise
with the shrinkage estimator studied in Section~\ref{sec:bias-estimation},
for which the  $O(m^{2})$ entries of  $\hat{M}$ must be estimated
precisely enough that they contribute downstream estimation error
of order $o_{p}\left(1/\sqrt{n}\right)$.
\end{rem}

\section{Bias correction and shrinkage \protect\label{sec:bias-estimation}}

In this section, I introduce two estimators that repair the damage misspecification does to GMM. The first estimator removes the leading-order bias by using the hyperparameter estimate $\hat \mu$. The second estimator uses the estimated hyperparameter $\widehat{\sigma^2}$ to form an empirical Bayes shrinkage prediction of the noise in the bias-corrected estimator. Subtracting this prediction from the bias-corrected estimator reduces its sensitivity to any remaining idiosyncratic specification errors. Each estimator is shown to retain GMM's $\sqrt{n}$ convergence rate despite reliance on plug-in hyperparameters. I show that the second estimator weakly improves on the first and provide conditions characterizing when the improvement is strict.

\subsection{Bias-corrected estimator}
Recall from Section \ref{sec:local_misspecification} that the leading misspecification bias of $\hat{\theta}$ is $\mu\Lambda1_{m}/\sqrt{n}$. Plugging in $\hat{\mu}$ for $\mu$ and $\hat{\Lambda}$ for $\Lambda$, then subtracting this leading term, yields the bias-corrected (BC) estimator
\[
\hat{\theta}_{BC}:=\hat{\theta}-\frac{\hat{\mu}}{\sqrt{n}}\cdot\hat{\Lambda}1_{m}.
\]

The BC estimator is a close cousin of the differencing estimator introduced in Section \ref{sec:A-differencing-approach}. Rather than differencing the leading term out, the BC estimator estimates it and subtracts it. When the moment conditions are linear and identity weights are used $(W=I_m, W_D=I_{m-1})$, the estimators $\hat \theta_{BC}$ and $\hat \theta_D$ are numerically equivalent.

The BC estimator is also closely related to a two-step minimum distance approach that \textcite{angrist2009mostly} term ``visual IV'' (VIV), wherein the reduced-form coefficients associated with dummy instruments are regressed on their corresponding first-stage estimates. The conventional TSLS fit is a weighted least squares regression of the reduced forms on the first stages, constrained to go through the origin. In contrast, the VIV estimator described below allows for an intercept, which should be approximately zero if the exclusion restriction is satisfied.

\begin{example}[BC as Visual IV \protect \label{exa:VIV}] 
    Continuing the overidentified IV setting of Example~\ref{exa:IV_example}, suppose the instruments $Z_{i}\in\mathbb{R}^{m}$ are mutually exclusive cell dummies. In this case, the first-stage and reduced-form cell means can be written
    \[
        \hat{\pi} = \hat{W}^{-1}\sum_{i} Z_{i} T_{i}, \quad
        \hat{\delta} = \hat{W}^{-1}\sum_{i} Z_{i} Y_{i}, \quad
        \hat{W} := \sum_{i} Z_{i} Z_{i}',
    \]
    where $\hat{W}$ is the diagonal matrix of cell counts. The cell-dummy TSLS estimator
    \[
        \hat{\theta} = (\hat{\pi}'\hat{W}\hat{\pi})^{-1}\hat{\pi}'\hat{W}\hat{\delta}
    \]
    is the cell-size--weighted least squares slope of the reduced forms $\hat{\delta}$ on the first stages $\hat{\pi}$, constrained to pass through the origin. The sample sensitivity matrix of TSLS is $\hat{\Lambda} = (\hat{\pi}'\hat{W}\hat{\pi})^{-1}\hat{\pi}'\hat{W}$, while the vector $\hat{\Lambda}1_{m} = (\hat{\pi}'\hat{W}\hat{\pi})^{-1}\hat{\pi}'\hat{W}1_{m}$ gives the through-origin WLS coefficient of $1_{m}$ on $\hat{\pi}$. 
    
    The omitted-variables bias formula relates the through-origin slope $\hat \theta$ to its intercept-inclusive counterpart $\hat \theta_{\mathrm{VIV}}$,
    \[
        \hat{\theta} = \hat{\theta}_{\mathrm{VIV}} + \frac{\hat{\mu}_{W}}{\sqrt{n}}\,\hat{\Lambda}1_{m},
    \]
    where $\hat{\theta}_{\mathrm{VIV}}$ is the slope and $\hat{\mu}_{W}/\sqrt{n}$ the intercept of the WLS regression of $\hat{\delta}$ on $(1_{m},\hat{\pi})$ with weight $\hat{W}$. The VIV intercept $\hat{\mu}_{W}$ agrees with $\hat \mu$ numerically when the cell counts are equal (i.e., when $\hat W \propto I_m$). In that case, the BC estimator $\hat \theta_{BC}$ also coincides with the slope $\hat{\theta}_{\mathrm{VIV}}$ of the VIV regression. Under unequal cell counts, the gap between the two estimators can be written
\[
\hat{\theta}_{BC}-\hat{\theta}_{\mathrm{VIV}}=\frac{\hat{\mu}_{W}-\hat{\mu}}{\sqrt{n}}\,\hat{\Lambda}1_{m}.
\]
Hence, the gap is entirely attributable to the two intercept estimates and scales with the vector $\hat{\Lambda}1_{m}$.
\end{example}

The connection between BC and VIV extends to nonlinear models but requires working with linearizations of the moment conditions.

\begin{example}[BC as VIV in a binary choice model \protect\label{exa:binarychoice_BC}]
    Recall the binary choice model of Example~\ref{exa:binarychoice}, with rescaled moments $g(\theta)=E_{n}[f(\theta X_{i}'1_{m})X_{i}X_{i}']^{-1}U(\theta1_{m})$ obeying $\sqrt{n}\,g(\theta_{0})=b$. The first-order expansion of these moment conditions around $\theta_{0}$ takes an IV-like form. Coordinate $j$ of the expansion can be written 
    \[
        g_{j}(\theta)= \delta_{j}-\theta\,\pi_{j}+o(n^{-1/2}),
    \]
    where $\pi_j$ is the $j$th entry in the $m$-vector of population first stages, evaluated under the limiting law $P_{0}$ of Section~\ref{subsec:Misspecification-as-contaminatio},
    \[
         \pi:=-G=E_{0}[f(\theta_{0}X_{i}'1_{m}) X_{i}X_{i}']^{-1}E_{0}\!\left[\frac{f(\theta_{0}X_{i}'1_{m})^{2}}{F(\theta_{0}X_{i}'1_{m})(1-F(\theta_{0}X_{i}'1_{m}))}X_{i}X_{i}'\right]1_{m},
    \]
    and $\delta_{j}=g_{j}(\theta_{0})+\theta_{0}\pi_{j}$ is the $j$th reduced form. This expansion holds uniformly over the $O_{p}(n^{-1/2})$ neighborhood of $\theta_{0}$ that contains $\hat{\theta}$.
    
    One can form a VIV-like estimator by replacing $E_{0}$ with the corresponding sample averages and $\theta_{0}$ with the GMM estimate $\hat{\theta}$. Doing so yields estimated first stages $\hat{\pi}$ and reduced forms $\hat{\delta}=\hat{g}(\hat{\theta})+\hat{\theta}\,\hat{\pi}$. 
    As in Example~\ref{exa:VIV}, one can think of GMM as fitting a line to a scatter of moments via GLS and the bias correction as allowing for an intercept. With identity weights ($W=I_m$), the through-origin slope that arises from regressing $\hat{\delta}_{j}$ on $\hat{\pi}_{j}$ agrees with the GMM estimate $\hat{\theta}$ up to $o_p(n^{-1/2})$. Including an intercept yields a slope that agrees with $\hat{\theta}_{BC}$ to the same order. In contrast to Example~\ref{exa:VIV}, these equivalences are first-order rather than exact because the binary-choice moment is nonlinear and the line is fit to its linearization.
\end{example}

\subsection{Convergence rate of the BC estimator \protect\label{subsec:bc-rate}}

The following regularity conditions ensure that $\hat{\theta}_{BC}$
is well behaved as $n$ grows large.
\begin{assumption}
\label{assu:rate}(Rate conditions). The Jacobian $G$ and weighting
matrix $W$ satisfy:

i) $\left\Vert \Lambda1_{m}\right\Vert =O(1)$.

ii) $\left\Vert \Lambda\right\Vert _{2}=O(1)$.
\end{assumption}
Condition (i) bounds the magnitude of the bias correction
associated with a unit perturbation of all the moments. Condition
(ii) imposes an operator-norm bound on the GMM sensitivity matrix
$\Lambda$, ensuring that no single moment exerts unbounded influence
on the parameter estimate.

The following proposition establishes that $\hat{\theta}_{BC}$ converges
to $\theta_{0}$ at the usual parametric rate.
\begin{prop}
\label{prop:bc-rate}Under Assumptions~\ref{assu:random_b},
\ref{assu:regime}, and \ref{assu:rate}, the feasible bias-corrected
estimator satisfies 
\[
\hat{\theta}_{BC}-\theta_{0}=O_{p}\left(\frac{1}{\sqrt{n}}\right).
\]
\end{prop}
\begin{proof}
See appendix.
\end{proof}
The BC adjustment $-(\hat{\mu}/\sqrt{n})\hat{\Lambda}1_{m}$ is $O_{p}(1/\sqrt{n})$, the order of the sampling noise, because $\hat{\mu}=O_{p}(1)$ by Lemma~\ref{lem:mu_consistency} and $\|\hat{\Lambda}1_{m}\|=O_{p}(1)$ as shown in the proof. Since $\hat{\mu}-\mu=o_{p}(1)$, replacing $\mu$ with $\hat{\mu}$ perturbs the adjustment by only $o_{p}(1/\sqrt{n})$, ensuring asymptotic equivalence with an estimator using the infeasible correction $-(\mu/\sqrt n)\hat\Lambda 1_m$. 

\subsection{Empirical Bayes estimator \protect\label{subsec:eb-shrinkage}}

The BC estimator removes the leading bias of GMM but leaves a first-order
error with predictable structure. Removing this predictable component of the error can improve precision.

The proof of Lemma~\ref{lem:mu_consistency} establishes that $\hat{\mu}=w'r_{n}+o_{p}(1)$, where $w:=\left\Vert M1_{m}\right\Vert ^{-2}M'M1_{m}$ is a vector of population weights.
Consequently, the BC estimation error can be written
\[
\sqrt{n}(\hat{\theta}_{BC}-\theta_{0})=\Lambda r_{n}-\hat{\mu}\Lambda1_{m}+o_{p}(1)=\Lambda M_{\mu}r_{n}+o_{p}(1)\overset{d}{\rightarrow}\Lambda M_{\mu}r,
\]
where $M_{\mu}:=I_{m}-1_{m}w'$ is the population centering matrix. Since $w'1_m=1$, $M_\mu$ annihilates the constant vector ($M_{\mu}1_{m}=0$). The BC estimation error is governed by the centered moment $\dot{r}:=M_{\mu}r$.
The GMM fitting process masks this quantity, revealing only $\ddot{r} :=M \dot{r}$, which has mean zero and covariance $M\Sigma_{\mu}M'$ with $\Sigma_{\mu}:=M_{\mu}\Sigma M_{\mu}'$.  

The best linear predictor (BLP) of $\dot{r}$ given $\ddot{r}$ is $\Pi\ddot{r}$, where
\[
    \Pi:=\Sigma_{\mu}M'(M\Sigma_{\mu}M')^{+}M.
\]
The $+$ superscript denotes the Moore--Penrose inverse, which is required here because $M\Sigma_{\mu}M'$ is singular with rank $m-p-1$. Since the residual maker $M$ is idempotent, $\Pi M=\Pi$ and $\Pi\ddot{r}=\Pi\dot{r}=\Pi M_{\mu}r$. If $r$ is normally distributed, then $\Pi\ddot{r}$ is also the minimum mean squared error predictor of $\dot{r}$. 

The BLP $\Pi\ddot{r}$ can be thought of as the fitted values from an infeasible regression of $\dot{r}$ on $\ddot{r}$ with no intercept. As noted by \textcite{stigler19901988}, these fitted values offer a shrunken, less variable image of the dependent variable $\dot{r}$. In this sense, $\Pi$ is a shrinkage operator. The residual vector $\dot{r}-\Pi\ddot{r}$ from this regression has a smaller variance matrix than $\dot{r}$ whenever $\Sigma_\mu M'\neq0$ (i.e., whenever $\ddot{r}$ and $\dot{r}$ are correlated). 

The corresponding BLP of the leading-term error $\Lambda\dot{r}$ is $\Lambda\Pi\ddot{r}$. Applying the population predictor to the finite-sample centered moment $M_{\mu}r_{n}$ and subtracting the result, scaled by $n^{-1/2}$, from the BC estimator yields an infeasible corrected estimator
\[
\hat{\theta}_{BC}-\frac{1}{\sqrt{n}}\Lambda\Pi M_{\mu}r_{n}=\theta_{0}+\frac{1}{\sqrt{n}}\Lambda(I_{m}-\Pi)M_{\mu}r_{n}+o_{p}(1/\sqrt{n}).
\]
Since $\dot{r}-\Pi\ddot{r}=(I_{m}-\Pi)M_{\mu}r$ has a weakly smaller variance matrix than $\dot{r}$, this infeasible estimator has weakly lower first-order variance than $\hat{\theta}_{BC}$.

To construct a feasible version of this estimator, I replace population quantities with sample analogues and follow the empirical Bayes principle of plugging in the estimated hyperparameters $(\hat \mu, \widehat{\sigma^2})$ for their unknown counterparts. The matrix $M_{\mu}$ is replaced with $\hat{M}_{\mu}:=I_{m}-1_{m}\hat{w}'$, where
$\hat{w}:=\hat{M}'\hat{M}1_{m}/\|\hat{M}1_{m}\|^{2}$. Likewise, $\Sigma_{\mu}$ is replaced with $\hat{\Sigma}_{\mu}:=\hat{M}_{\mu}\hat{\Sigma}\hat{M}_{\mu}'$,
where $\hat{\Sigma}=\hat{V}+\widehat{\sigma^{2}}Q$ estimates the total
moment variance. Finally, the vector $M_{\mu}r$ has feasible counterpart $\hat{r}-\hat{\mu}1_{m}$, where $\hat{r}$ estimates the masked moment $Mr$. Hence the plug-in BLP is $\hat{\Pi}(\hat{r}-\hat{\mu}1_{m})$, where
\[
\hat{\Pi}:=\hat{\Sigma}_{\mu}\hat{M}'(\hat{M}\hat{\Sigma}_{\mu}\hat{M}')^{+}\hat{M}.
\]
The trailing $\hat{M}$ in $\hat{\Pi}$ accounts for the effects of the masking.

The plug-in predictor of the BC estimator's first-order error is $\hat{\Lambda}\hat{\Pi}(\hat{r}-\hat{\mu}1_{m})/\sqrt{n}$. Subtracting the predicted first-order error from
the BC estimate yields the empirical Bayes (EB) estimator
\[
\hat{\theta}_{EB}:=\hat{\theta}_{BC}-\frac{1}{\sqrt{n}}\hat{\Lambda}\hat{\Pi}(\hat{r}-\hat{\mu}1_{m}).
\]
The EB estimator is a regression adjustment of BC, or equivalently, of its influence function. Intuitively, removing the predictable component of the BC estimator's first-order error should improve its precision. I formalize this intuition in Section \ref{subsec:ordering}.

\begin{rem}[EB when $\mu$ is unidentified]\label{rem:eb-pure-shrinkage}
When $\|\hat{M}1_{m}\|^{2}\approx0$, the mean $\mu$ is unidentified. Once $\|\hat{M}1_{m}\|^{2}$ falls below a small threshold I set $\hat{\mu}:=0$ and $\hat{w}:=0$, replacing $\hat{M}_{\mu}$ by $I_{m}$. The BC estimator then coincides with GMM. The EB estimator remains well-defined, with shrinkage operator $\hat{\Pi}=\hat{\Sigma}\hat{M}'(\hat{M}\hat{\Sigma}\hat{M}')^{+}\hat{M}$, yielding a pure shrinkage correction with no mean adjustment. The theoretical results below will assume $\mu$ is identified.
\end{rem}

\subsection{Convergence rate of the EB estimator}

The EB estimator requires additional regularity conditions to deal with the possibility that some moment directions become weakly identified after centering. These conditions concern the behavior of the shrinkage operator under different hypothetical levels $s\ge0$ of the specification-error variance. 

Define $\Sigma_{\mu}(s):=M_{\mu}(V+sQ)M_{\mu}'$ and $\Pi(s):=\Sigma_{\mu}(s)M'(M\Sigma_{\mu}(s)M')^{+}M$. Let $\hat{\Sigma}_{\mu}(s)$ and $\hat{\Pi}(s)$ be their sample analogues, built from $\hat{M}$, $\hat{M}_{\mu}$, and $\hat{V}+sQ$. Thus, $\Pi=\Pi(\sigma^{2})$ and $\hat{\Pi}=\hat{\Pi}(\widehat{\sigma^{2}})$. The following assumption restricts these objects.
\begin{assumption}[Empirical Bayes regularity]\label{assu:eb-reg} The functions
$\Sigma_{\mu}(s)$, $\Pi(s)$, and $\hat{\Pi}(s)$ satisfy:

i) $\mathrm{tr}\big(((M\Sigma_{\mu}(0)M')^{+})^{2}\big)=O(m)$;

ii) $E\big[\sup_{s\ge0}\|(\hat{\Pi}(s)-\Pi(s))M_{\mu}r_{n}\|^{2}\big]=O(m^{2}/n)$ and $E\big[\sup_{s\ge0}\|(\hat{\Pi}(s)-\Pi(s))\Sigma_{\mu}^{1/2}\|_{F}^{2}\big]=O(m^{2}/n)$;

iii) $\sup_{s\ge0}\|(I_{m}-\Pi(s))M_{\mu}\|_{2}=O(1)$.
\end{assumption}
Condition (i) bounds the sum of the squared inverse eigenvalues of the centered covariance $M\Sigma_{\mu}(s)M'$ to order $m$. This caps the weak identification left after centering: a growing number of mildly weak directions is allowed, but no single direction is allowed to vanish faster than $m^{-1/2}$. The condition is stated at $s=0$, where it binds. 

Condition (ii) requires the sample projection to track its population counterpart in mean square, uniformly in the variance argument $s$. The uniformity is needed because the feasible projection is evaluated at the estimate $\widehat{\sigma^{2}}$. The condition holds when the first-stage error avoids the weakly identified directions. The second part of the condition states the same requirement for the covariance $\Sigma_{\mu}$ of $M_{\mu}r_{n}$. Condition (iii) ensures the population shrinkage residual operator is bounded, again uniformly in the variance argument.

The following proposition establishes the approximation error of the
plug-in predictor $\hat{\Pi}(\hat{r}-\hat{\mu}1_{m})$ and the convergence rate of the empirical
Bayes estimator.
\begin{prop}
\label{prop:eb-rate}Under Assumptions~\ref{assu:random_b},
\ref{assu:regime}, \ref{assu:rate}, and \ref{assu:eb-reg}:\\
 (i) The feasible predictor satisfies
\[
\frac{1}{\sqrt{m}}\|\hat{\Pi}(\hat{r}-\hat{\mu}1_{m})-\Pi M_{\mu}r_{n}\|=O_{p}\left(\frac{1}{\sqrt{m}}\right)+O_{p}\left(\sqrt{\frac{m}{n}}\right).
\]
\\
(ii) The empirical Bayes estimator satisfies
\[
\hat{\theta}_{EB}-\theta_{0}=O_{p}\left(\frac{1}{\sqrt{n}}\right).
\]
\end{prop}
\begin{proof}
See appendix.
\end{proof}
Part (i) bounds the approximation error of the feasible predictor relative to an oracle that knows the population matrices $\left(\Pi,M_{\mu}\right)$ and the random vector $r_n$. This error converges at the same rate as the hyperparameters described in Lemmas~\ref{lem:mu_consistency} and \ref{lem:sigma-consistency}. As in those results, the approximation error scaled by $m^{-1/2}$ vanishes when $m$ grows more slowly than $n$.

Part (ii) establishes that $\hat{\theta}_{EB}$ converges at the parametric
rate, matching the feasible bias-corrected estimator of
Proposition~\ref{prop:bc-rate}. The hyperparameter errors $\hat{\mu}-\mu$
and $\widehat{\sigma^{2}}-\sigma^{2}$ are $o_{p}(1)$ by
Lemmas~\ref{lem:mu_consistency} and \ref{lem:sigma-consistency}. Since $\hat{M}_{\mu}'\hat{M}'\hat{M}1_{m}=0$, the vector $\hat{M}1_{m}$ lies in the null space of $\hat{M}\hat{\Sigma}_{\mu}\hat{M}'$. Consequently, $\hat{\Pi}1_{m}=0$ and the EB adjustment reduces to $-\hat{\Lambda}\hat{\Pi}\hat{r}/\sqrt{n}$. 
The $1/\sqrt{n}$ scaling ensures these hyperparameter errors do not compromise the
parametric rate. 

\subsection{Variance ordering \protect \label{subsec:ordering}}

Define the leading-term risk of an estimator
as the MSE of the leading linear term in its $\sqrt{n}$-scaled estimation
error. Part (i) of the
following proposition establishes that the EB estimator weakly improves on the leading-term risk of the BC estimator. Part (ii) characterizes the leading term of the EB estimator under a mild nondegeneracy condition.
\begin{prop}
\label{prop:variance-ordering}Under Assumptions~\ref{assu:random_b} and \ref{assu:regime}:\\
(i) The bias-corrected and empirical Bayes estimators have leading-term
risks $V_{BC}=\Lambda\Sigma_{\mu}\Lambda'$ and
$V_{EB}=\Lambda(I_{m}-\Pi)\Sigma_{\mu}(I_{m}-\Pi)'\Lambda'$. The difference
\[
V_{BC}-V_{EB}=\Lambda\Sigma_{\mu}M'(M\Sigma_{\mu}M')^{+}M\Sigma_{\mu}\Lambda'
\]
is positive semidefinite. The leading-term risk improvement
is nonzero if and only if $M\Sigma_{\mu}\Lambda'\neq0$. When $p=1$,
nonzeroness and positive definiteness coincide, and positive definiteness
in general requires $m\geq2p+1$.\\
(ii) If, in addition, $\lambda_{\min}(V)>0$, then the leading term of $\hat \theta_{EB}$ equals the influence function of efficient GMM applied to the centered moments $M_{\mu}g(\theta)$:
\[
\Lambda(I_{m}-\Pi)M_{\mu}=-\left(G'\Sigma_{\mu}^{+}G\right)^{-1}G'\Sigma_{\mu}^{+}M_{\mu},
\]
and consequently
\[
V_{EB}=\left(G'\Sigma_{\mu}^{+}G\right)^{-1}=\left(G'\Sigma^{-1}G-\frac{G'\Sigma^{-1}1_{m}1_{m}'\Sigma^{-1}G}{1_{m}'\Sigma^{-1}1_{m}}\right)^{-1}.
\]
Neither expression depends on the weighting matrix $W$.
\end{prop}
\begin{proof}
See appendix.
\end{proof}
The risk improvement offered by the EB estimator is governed by the matrix $M\Sigma_{\mu}\Lambda'$. Writing $\Sigma_{\mu}=M_{\mu}VM_{\mu}'+\sigma^{2}M_{\mu}M_{\mu}'$ splits this matrix into a noise term and a dispersion term. Each of these terms provides a potential source of risk improvement. For the dispersion term to contribute to EB's advantage over BC, the model must be misspecified ($\sigma^{2}>0$) with $M\Sigma_{\mu}\Lambda'\neq0$. 

Since the noise term survives at $\sigma^{2}=0$, the proposition implies EB can improve on BC even under correct specification. Part (ii) of the proposition shows EB is asymptotically equivalent to efficient GMM applied to the centered moment conditions $M_{\mu}g(\theta)$ at every $\sigma^{2}\ge0$, regardless of the weighting matrix $W$. As a result, inefficient weighting inflates the leading-term variance of GMM and BC but not EB, widening the scope for risk improvements. 

Another potential source of improvements under correct specification stems from the structure of the moment noise matrix $V$. With efficient GMM weighting, the EB risk improvement vanishes under homoscedastic moment noise ($V\propto I_m$), while heteroscedastic moment noise generally leaves it nonzero. This pattern reflects that the choice $W=V^{-1}$ is suboptimal for the centered moments unless the noise is homoscedastic. 

A closely related variance ordering underlies the inadmissibility result of \textcite{brown1990ancillarity}, who considered the problem of estimating the intercept of a linear regression with random (ancillary) controls. In that context, as here, exploiting high-dimensional ancillary controls can reduce the risk of a low-dimensional target, even a scalar one. Indeed, the centered sample moment $\hat{r}-\hat{\mu}1_{m}$ is asymptotically ancillary for $\theta_0$, providing a formal connection with the framework of \textcite{brown1990ancillarity}.

\begin{rem}[Uniform integrability and moments]
Proposition~\ref{prop:variance-ordering} ranks $\hat{\theta}_{BC}$ and
$\hat{\theta}_{EB}$ by leading-term risk. When the errors $\sqrt{n}(\hat{\theta}_{EB}-\theta_{0})$ and
$\sqrt{n}(\hat{\theta}_{BC}-\theta_{0})$ are uniformly square-integrable, the two
estimators can be ranked in limiting MSE.
\end{rem}

\begin{rem}[Linear models and global misspecification]\label{rem:linear-global-short}
    When $g$ is affine in $\theta$ and the weighting matrix does not depend on $\theta_0$, the representations
    $\hat\theta_{BC}-\theta_0=\hat\Lambda\hat M_\mu\hat g(\theta_0)$ and
    $\hat\theta_{EB}-\theta_0=\hat\Lambda(I_m-\hat\Pi)\hat M_\mu\hat g(\theta_0)$ hold exactly. Since $\hat M_\mu 1_m=0$, both errors are invariant to common shifts $\hat g(\theta_0)\mapsto
    \hat g(\theta_0)+c\,1_m$. Consequently the common-mean component of the misspecification bias is
    purged, as with the differencing estimator of Section~\ref{sec:A-differencing-approach}. The variance ordering of Proposition~\ref{prop:variance-ordering} then holds exactly, regardless
    of the scale of the specification errors. However, the
    residual idiosyncratic bias is of the same order as that scale. Hence, consistency for $\theta_{0}$ still requires the specification errors to vanish as $n$ grows large.
    \end{rem}

\section{\protect\label{sec:distribution-theory}Distribution theory}

This section sharpens the results of Propositions~\ref{prop:bc-rate}
and~\ref{prop:eb-rate} to distributional convergence, which provides
a basis for misspecification-aware inference. The distributional convergence
results are established in Section \ref{subsec:CLTs}. Section \ref{subsec:Inference} discusses standard error estimation and inference. Section \ref{subsec:Diagnostic} develops two specification tests
that can be used to assess the plausibility of the exchangeability assumption.

\subsection{Central limit theorems \protect\label{subsec:CLTs}}

Define the centered moment vector $\eta_{n}:=r_{n}-\mu1_{m}=(b-\mu1_{m})+\varepsilon_{n}$. Both the BC and EB estimators studied below are, to leading order, studentized linear functions of $\eta_{n}$ formed from a centered sensitivity matrix. I first state a general central limit
theorem for any such estimator, then specialize it to the two of interest.

\subsubsection{General result}
Consider a generic estimator $\hat{\theta}_{\star}$ whose scaled error admits a linear expansion
$\sqrt{n}(\hat{\theta}_{\star}-\theta_{0})=\Lambda_{\star}\eta_{n}+\tilde{R}_{\star}$, where $\Lambda_{\star}$
is a $p\times m$ sensitivity matrix with $\Lambda_{\star}1_{m}=0$ and $\tilde{R}_{\star}$ is a remainder.
Write $V_{\star}:=\Lambda_{\star}\Sigma \Lambda_{\star}'$ for the variance of the leading term $\Lambda_{\star}\eta_{n}$.
The following assumption restricts how this leading-term variance grows with the sample size.
\begin{assumption}
\label{assu:lindeberg}The sensitivity matrix $\Lambda_{\star}$ satisfies $\Lambda_{\star}1_{m}=0$. Its
leading-term variance $V_{\star}$ and the realized variance $\sigma^{2}$ satisfy:

i) $V_{\star}$ and $\Lambda_{\star}V\Lambda_{\star}'$ are positive definite for sufficiently large $m$, and $\left\Vert \Lambda_{\star}\right\Vert _{2}=O(1)$.

ii) $(\Lambda_{\star}V\Lambda_{\star}')^{-1/2}\Lambda_{\star}\varepsilon_{n}\overset{d}{\rightarrow}N(0,I_{p})$.

iii) $\max_{1\le j\le m}\Lambda_{\star,j}'V_{\star}^{-1}\Lambda_{\star,j}\to0$,
where $\Lambda_{\star,j}$ denotes the $j$-th column of $\Lambda_{\star}$.

iv) $\sigma^{2}=O_{p}(1)$, and $\lambda_{\min}(V_{\star})$ is bounded away from zero, where $\lambda_{\min}\left(\cdot\right)$ denotes the smallest eigenvalue of the input matrix.
\end{assumption}
Condition (i) ensures that studentization is possible. Condition (ii)
is a high-level triangular-array central limit theorem for the noise vector $\Lambda_{\star}\varepsilon_{n}$, standardized by its own variance to a fixed $N(0,I_{p})$ limit. With i.i.d.\ moment contributions obeying the standard finite moment requirements
for cross-sectional GMM \parencite[e.g.,][]{newey1994large}, the
condition can be delivered by appeal to the multivariate Lindeberg--Feller
CLT. Under weakly dependent or time-series moment contributions, alternative
central limit theorems for mixingales or martingale-difference sequences
can deliver the result under analogous moment conditions. Condition
(iii) requires that no single moment dominate the studentized sensitivity
matrix $V_{\star}^{-1/2}\Lambda_{\star}$. This condition can equivalently be written
as $\max_{j}\Vert V_{\star}^{-1/2}\Lambda_{\star,j}\Vert^{2}\to0$, which
is the standard Lindeberg-type requirement for a central limit theorem
on weighted sums of exchangeable variables. 

Condition (iv) requires the realized variance to be bounded in probability and the leading-term variance to be nondegenerate. The eigenvalue floor ensures that first-stage noise, which enters the remainders at order $\sqrt{m/n}$, is asymptotically negligible after studentization. Because $\sigma^{2}\Lambda_{\star}Q\Lambda_{\star}'$ is positive semidefinite, $\lambda_{\min}(V_{\star})\ge\lambda_{\min}(\Lambda_{\star}V\Lambda_{\star}')$ at every realized variance. When $\Lambda_{\star}$ does not depend on the realized variance, condition (iii) and the eigenvalue floor in condition (iv) can be checked using the nonrandom matrix $\Lambda_{\star}V\Lambda_{\star}'$ even though $V_{\star}$ is random.

The following theorem shows that any estimator satisfying these conditions (with a negligible remainder) is asymptotically normal after studentization.
\begin{thm}
\label{thm:clt}Suppose $\sqrt{n}(\hat{\theta}_{\star}-\theta_{0})=\Lambda_{\star}\eta_{n}+\tilde{R}_{\star}$,
where $\Lambda_{\star}$ satisfies Assumption~\ref{assu:lindeberg} and $V_{\star}^{-1/2}\tilde{R}_{\star}=o_{p}(1)$.
Then under Assumptions~\ref{assu:random_b} and \ref{assu:regime},
\[
\sqrt{n}\,V_{\star}^{-1/2}\left(\hat{\theta}_{\star}-\theta_{0}\right)\overset{d}{\rightarrow}N\left(0,I_{p}\right).
\]
\end{thm}
\begin{proof}
See appendix.
\end{proof}
The theorem establishes that the leading term
$\Lambda_{\star}\eta_{n}$ of such an estimator is asymptotically normal. This term splits into a noise part and a specification part.
The noise part $\Lambda_{\star}\varepsilon_{n}$ is covered by condition (ii). The specification
part $\Lambda_{\star}(b-\mu1_{m})$ is covered by a central limit theorem for permutations when the realized variance is bounded away from zero, and is negligible after studentization when the realized variance shrinks to zero. The permutation argument uses the exchangeability and bounded fourth moments of Assumption~\ref{assu:random_b} and the no-dominant-column condition of Assumption~\ref{assu:lindeberg}.iii.

Note that the statistic in Theorem~\ref{thm:clt} is self-normalizing---the studentizer $V_{\star}$ is formed at the realized $\sigma^{2}$---and the limit is unconditional, with the probability integrating over both $b$ and $\varepsilon_{n}$. The normal approximation does not deteriorate near the correctly specified boundary: the limit is $N(0,I_{p})$ regardless of whether the realized variance drifts toward zero.

\subsubsection{Specialization to BC and EB}
The BC and EB estimators both admit leading-term expansions of the sort described by Theorem~\ref{thm:clt}. The BC
sensitivity matrix is $\Lambda_{BC}:=\Lambda M_{\mu}$ and the EB sensitivity matrix is
$\Lambda_{EB}:=\Lambda(I_{m}-\Pi)M_{\mu}$, with leading-term variances
$V_{BC}:=\Lambda_{BC}\Sigma\Lambda_{BC}'$ and $V_{EB}:=\Lambda_{EB}\Sigma\Lambda_{EB}'$.
Both satisfy $\Lambda_{BC}1_{m}=\Lambda_{EB}1_{m}=0$.

Since $V_{BC}=\Lambda_{BC}V\Lambda_{BC}'+\sigma^{2}\Lambda_{BC}\Lambda_{BC}'$ retains the pure sampling-noise term $\Lambda_{BC}V\Lambda_{BC}'$ at every $\sigma^{2}$, it is nondegenerate whenever the noise term is. Establishing nondegeneracy of the EB variance is subtler, as shrinkage (unlike bias correction) removes variance. The following lemma shows that $V_{EB}$ is also nondegenerate at any level of the realized variance and for any admissible weighting.
\begin{lem}\label{lem:eb-floor}
Suppose Assumptions~\ref{assu:random_b} and \ref{assu:regime} hold, $\lambda_{\min}(V)\ge c_{V}>0$, and $G'V^{-1}G=O(1)$. Then, for every $\sigma^{2}\ge0$ and every weighting matrix admitted by Assumption~\ref{assu:regime},
\[
\lambda_{\min}(V_{EB})\ge\frac{1}{\lambda_{\max}(G'V^{-1}G)}.
\]
The bound is free of $\sigma^{2}$ and of the weighting matrix, and it is bounded away from zero uniformly in $m$.
\end{lem}
\begin{proof}
See appendix.
\end{proof}
The bound, which follows from standard Gauss--Markov reasoning, holds for every weighting matrix $W$, not only the efficient one. Thus, the EB estimator satisfies the nondegeneracy requirement of Assumption~\ref{assu:lindeberg}.iv.

The following corollary gives conditions under which the BC and EB expansions have an asymptotically negligible remainder, in which case Theorem~\ref{thm:clt} applies.
\begin{cor}
\label{corr:bc-eb-clt}Suppose Assumptions~\ref{assu:random_b}, \ref{assu:regime},
and \ref{assu:rate} hold.
\begin{enumerate}
\item[(i)] If $\Lambda_{BC}$ satisfies Assumption~\ref{assu:lindeberg}, then
$\sqrt{n}\,V_{BC}^{-1/2}(\hat{\theta}_{BC}-\theta_{0})\overset{d}{\rightarrow}N(0,I_{p})$.

\item[(ii)] If Assumption~\ref{assu:eb-reg} holds, $\lambda_{\min}(V)\ge c_{V}$ for a constant $c_{V}>0$, and $\Lambda_{EB}$ satisfies Assumption~\ref{assu:lindeberg}, then
$\sqrt{n}\,V_{EB}^{-1/2}(\hat{\theta}_{EB}-\theta_{0})\overset{d}{\rightarrow}N(0,I_{p})$.
\end{enumerate}
\end{cor}
\begin{proof}
See appendix.
\end{proof}
Part (ii) adds two conditions. Assumption~\ref{assu:eb-reg} controls the first-stage error in the shrinkage operator. The eigenvalue floor $\lambda_{\min}(V)\ge c_{V}$ controls the $\widehat{\sigma^{2}}$-correction. Since $\Sigma-V=\sigma^{2}Q$ is positive semidefinite, the same floor holds for $\lambda_{\min}(\Sigma)$ at every realized variance. The proof also uses the floor on $\lambda_{\min}(V_{EB})$ in Assumption~\ref{assu:lindeberg}.iv, which Lemma~\ref{lem:eb-floor} delivers uniformly in $\sigma^{2}\ge0$ when $G'V^{-1}G=O(1)$. Both conditions can be checked using sample analogs.

\subsection{Inference \protect\label{subsec:Inference}}

Feasible inference can be conducted using plug-in versions of the leading-term variances.
With $\hat{M}_{\mu}$ and $\hat{\Pi}$ as defined in Section~\ref{sec:bias-estimation},
let $\hat{\Lambda}_{BC}:=\hat{\Lambda}\hat{M}_{\mu}$
and $\hat{\Lambda}_{EB}:=\hat{\Lambda}\left(I_{m}-\hat{\Pi}\right)\hat{M}_{\mu}$.
The feasible variance estimators are $\hat{V}_{BC}:=\hat{\Lambda}_{BC}\hat{\Sigma}\hat{\Lambda}_{BC}'$
and $\hat{V}_{EB}:=\hat{\Lambda}_{EB}\hat{\Sigma}\hat{\Lambda}_{EB}'$,
with $\hat{\Sigma}:=\hat{V}+\widehat{\sigma^{2}}Q$.
\begin{cor}
\label{corr:feasible-inference}Under the assumptions of Corollary~\ref{corr:bc-eb-clt}(i),
\[
\sqrt{n}\,\hat{V}_{BC}^{-1/2}\left(\hat{\theta}_{BC}-\theta_{0}\right)\overset{d}{\rightarrow}N\left(0,I_{p}\right).
\]
Under the assumptions of Corollary~\ref{corr:bc-eb-clt}(ii),
\[
\sqrt{n}\,\hat{V}_{EB}^{-1/2}\left(\hat{\theta}_{EB}-\theta_{0}\right)\overset{d}{\rightarrow}N\left(0,I_{p}\right).
\]
\end{cor}
\begin{proof}
See appendix.
\end{proof}
The corollary adds no conditions to those of Corollary~\ref{corr:bc-eb-clt}. The proof replaces $V_{\star}$ by its plug-in estimate and shows that the estimation errors in $\hat{\Lambda}_{\star}$ and $\hat{\Sigma}$ are negligible after studentization. The eigenvalue floor in Assumption~\ref{assu:lindeberg}.iv does the work. Lemma~\ref{lem:eb-floor} supplies the EB floor. For BC, $V_{BC}-\sigma^{2}\Lambda_{BC}\Lambda_{BC}'=\Lambda_{BC}V\Lambda_{BC}'$ is positive semidefinite at every $\sigma^{2}\ge0$, and $V_{BC}$ is bounded away from zero whenever $V$ is bounded away from zero on the row space of $\Lambda_{BC}$. The same floors give $\|V_{BC}^{-1/2}\Lambda_{BC}\|_{2}=O(1)$ and $\|V_{EB}^{-1/2}\Lambda_{EB}\|_{2}=O(1)$, which control the $\widehat{\sigma^{2}}$ contribution to the variance estimates.

Asymptotic $100(1-\tau)\%$ confidence intervals for the $k$-th
component of $\theta_{0}$ take the form $\hat{\theta}_{BC,k}\pm z_{1-\tau/2}\sqrt{\hat{V}_{BC,kk}/n}$
and $\hat{\theta}_{EB,k}\pm z_{1-\tau/2}\sqrt{\hat{V}_{EB,kk}/n}$,
respectively. Under the regularity conditions of Corollary~\ref{corr:feasible-inference}, these intervals satisfy
\[
\begin{aligned}\Pr\!\left[\theta_{0,k}\in\hat{\theta}_{BC,k}\pm z_{1-\tau/2}\sqrt{\hat{V}_{BC,kk}/n}\,\Big|\,\mu,\sigma^{2}\right] & \overset{p}{\rightarrow}1-\tau,\\
\Pr\!\left[\theta_{0,k}\in\hat{\theta}_{EB,k}\pm z_{1-\tau/2}\sqrt{\hat{V}_{EB,kk}/n}\,\Big|\,\mu,\sigma^{2}\right] & \overset{p}{\rightarrow}1-\tau,
\end{aligned}
\]
for each $k=1,\ldots,p$ and $\tau\in(0,1)$, where the probability
is taken over the joint distribution of $(b,\varepsilon_{n})$ given
$(\mu,\sigma^{2})$. Conditional coverage is random because the hyperparameters are computed from $b$. The convergence holds in probability: hyperparameter values at which coverage departs from $1-\tau$ can exist but have vanishing probability. Averaging over the hyperparameters gives unconditional coverage of $1-\tau$. The contributions $\widehat{\sigma^{2}}\hat{\Lambda}_{BC}\hat{\Lambda}_{BC}'$
and $\widehat{\sigma^{2}}\hat{\Lambda}_{EB}\hat{\Lambda}_{EB}'$ to the
variances make these intervals \textit{misspecification-aware}: they
capture the conditional cross-moment variability in $b$ that standard
GMM intervals ignore.
\begin{rem}[Uniform conditional coverage]\label{rem:pointwise-conditional}
Theorem~\ref{thm:clt} delivers conditional coverage in probability. Coverage at every value of $(\mu,\sigma^{2})$ requires additional restrictions ruling out the possibility that a large $\sigma^{2}$ reflects a single large specification error. A simple sufficient condition is the conditional moment bound $E[(b_{j}-\mu)^{4}\mid\mu,\sigma^{2}]\le C$ for some constant $C$. 
\end{rem}

\subsection{Specification tests}\label{subsec:Diagnostic}

It is useful to have a diagnostic for whether the exchangeability restriction
of Assumption~\ref{assu:random_b}.ii is plausible in a given application. The machinery developed thus far suggests two simple tests. 

\subsubsection{A Hausman-type test}
A first test for violations of exchangeability examines whether the shrinkage adjustment has the magnitude one would expect from exchangeable specification errors. Under the conditions of Corollary~\ref{corr:bc-eb-clt}(ii), the difference between the EB and BC estimators can be written:
\[
\hat{\theta}_{EB}-\hat{\theta}_{BC}=-\frac{1}{\sqrt{n}}\hat{\Lambda}\hat{\Pi}\left(\hat{r}-\hat{\mu}1_{m}\right) = -\frac{1}{\sqrt{n}} \Lambda\Pi M_{\mu}\eta_{n} + o_p(1/\sqrt{n}).
\]
Exchangeability restricts the term $M_{\mu}\eta_{n}$ to have mean zero, since $E[M_{\mu}b\mid\mu,\sigma^{2}]=\mu M_{\mu}1_{m}=0$. The leading-term variance of this contrast is the difference term of Proposition~\ref{prop:variance-ordering}(i):
\[
\Lambda\Pi\Sigma_{\mu}\Pi'\Lambda'=V_{BC}-V_{EB}.
\]
Hence, as in \textcite{hausman1978specification}, the leading-term variance of the difference equals the difference of the variances. This property follows from the EB estimator's representation in Proposition~\ref{prop:variance-ordering}(ii) as efficient GMM on the centered moment conditions.

These observations motivate the Hausman-type test statistic
\begin{align} \label{stat:Hausman}
    n\left(\hat{\theta}_{EB}-\hat{\theta}_{BC}\right)'\left(\hat{V}_{BC}-\hat{V}_{EB}\right)^{-1}\left(\hat{\theta}_{EB}-\hat{\theta}_{BC}\right),
\end{align}
where $\hat{V}_{BC}$ and $\hat{V}_{EB}$ are the variance estimators of Section~\ref{subsec:Inference}. The estimated variance difference is positive semidefinite by construction. When Assumption~\ref{assu:lindeberg} holds for the difference sensitivity $-\Lambda\Pi M_{\mu}$, Theorem~\ref{thm:clt} implies the statistic is approximately $\chi^{2}(p)$ under exchangeability. The test requires $V_{BC}-V_{EB}$ to be nondegenerate.

The statistic in \eqref{stat:Hausman} isolates the restriction that EB exploits beyond BC. Under exchangeability the shrinkage adjustment subtracts predictable noise and nothing else. In contrast, when the entries of $E[b]$ are not all equal, the shrinkage adjustment shifts the difference by $-\Lambda\Pi M_{\mu}E[b]/\sqrt{n}$ to leading order. Hence, the null being tested is $H_0:\Lambda\Pi M_{\mu}E[b]=0$, while the alternative is $H_1:\Lambda\Pi M_{\mu}E[b]\neq 0$. Note that DGPs in $H_1$ will not induce a leading-term bias in $\hat \theta_{BC}$ if $\Lambda M_{\mu}E[b]=0$. 

While \eqref{stat:Hausman} offers a computationally convenient diagnostic, the family of violations entertained by $H_1$ is rather narrow, indicating that this test will often be unable to detect interesting violations of exchangeability. I therefore consider a second test of exchangeability based on observed moment features.

\subsubsection{A test using moment features}
Exchangeability implies that the specification errors $b_{j}$ bear no systematic
relationship to observed features of the moment conditions that vary across $j$.
For example, in cell-based IV designs the cells typically differ in size. A significant covariance between the specification errors and cell size would therefore violate the implication of Assumption~\ref{assu:random_b}.ii that $E[b]$ is permutation invariant. The
distribution theory of Section~\ref{subsec:CLTs} suggests a simple test of this
sort of restriction.

Let $\Xi$ be an $m\times d$ matrix collecting $d$ observed moment-level features,
its $j$-th row holding the features of moment $j$, centered so that
$\Xi'1_{m}=0$. Under Assumption~\ref{assu:random_b}, conditioning on the features
leaves
\[
E\left[\Xi'b\mid\Xi\right]=\mu\,\Xi'1_{m}=0.
\]
The sample counterpart is $\Xi'\tilde{r}$, where
$\tilde{r}:=\hat{r}-\hat{\mu}\hat{M}1_{m}=\hat{M}\hat{M}_{\mu}\hat{r}$ is the
residual $\hat{r}$ after removing the estimated mean along $\hat{M}1_{m}$. This
vector should be close to zero under exchangeability. 

Theorem~\ref{thm:clt} and
the arguments used in its proof imply that $\Xi'\tilde{r}$ is asymptotically
normal with variance $\Omega:=\Xi'MM_{\mu}\Sigma M_{\mu}'M'\Xi$ provided that no single moment dominates the weights $\Xi'MM_{\mu}$. This observation motivates the
test statistic
\begin{align}\label{stat:feature}
    \left(\Xi'\tilde{r}\right)'\hat{\Omega}^{-1}\left(\Xi'\tilde{r}\right),\qquad
    \hat{\Omega}:=\Xi'\hat{M}\hat{M}_{\mu}\hat{\Sigma}\hat{M}_{\mu}'\hat{M}'\Xi.    
\end{align}
Under the null this statistic is approximately $\chi^{2}(d)$, and each feature
$a\in\{1,\dots,d\}$ can be assessed individually via the $t$-statistic
$\left(\Xi'\tilde{r}\right)_{a}/\sqrt{\hat{\Omega}_{aa}}$. The variance
$\hat{\Omega}$ treats $\Xi$ as fixed: this is exact when the features are
non-random and a leading-order approximation when they are precisely estimated
design quantities.

\begin{rem}[What the tests cannot detect]\label{rem:test-power}
Because $MM_{\mu}$ annihilates $\mathrm{span}(1_{m},G)$, the statistic $\Xi'\tilde{r}$ is blind to any component of the specification error lying in the span of the constant and the columns of the Jacobian. The statistic \eqref{stat:Hausman} shares this blind spot, as $M_{\mu}1_{m}=0$, $M_{\mu}G=G$, and $\Pi G=0$. A specification error proportional to a column of
$G$ is observationally equivalent to a shift in $\theta$. Hence, in the linear IV setting, neither test has power to detect a linear relationship between the specification errors and the first stage $\pi$.
\end{rem}

\section{Monte Carlo Evidence \protect\label{sec:monte-carlo}}

This section presents Monte Carlo evidence on the finite-sample performance
of the estimators developed above. The simulations are based on the
overidentified instrumental variables model of Example \ref{exa:IV_example}.
The results detail the relative performance of the GMM, BC, and EB
estimators under local misspecification. 

\subsection{Simulation design\protect\label{subsec:Simulation-design}}

For $i=1 \dots n$, the data generating process (DGP) is 
\begin{align*}
Y_{i} & =\alpha_{0}+\beta_{0}T_{i}+\frac{1}{\sqrt{n}}\left(\sum_{j=1}^{m}b_{j}Z_{ij}+b_{m+1}\right)+e_{i},\\
T_{i} & =\mathbf{1}\{Z_{i}^{\prime}\pi^{*}+\nu_{i}>0\},\quad e_{i}=\sigma_{e}\left(\rho\nu_{i}+\sqrt{1-\rho^{2}}\eta_{i}\right),\\
\left(b_{1},\dots,b_{m+1}\right)' & \sim t_{5}(\bar{\mu}\mathbf{1}_{m+1},\,\tfrac{3}{5}\bar{\sigma}^{2}I_{m+1}),\quad\nu_{i}\overset{\mathrm{i.i.d.}}{\sim}N(0,1),\quad\eta_{i}\overset{\mathrm{i.i.d.}}{\sim}N(0,1),\\
 & \quad\quad\quad\quad\quad\quad\quad\nu_{i}\perp\eta_{i}\perp Z_{i}\perp\left(b_{1},\dots,b_{m+1}\right).
\end{align*}
The parameter  $\rho=\mathrm{Corr}(e_{i},\nu_{i})$ governs the degree
of endogeneity of  $T_{i}$, while the scale parameter  $\sigma_{e}$
controls the variance of the structural error. The vector $Z_{i}=(Z_{i1},\ldots,Z_{im})'$
is comprised of mutually exclusive binary instruments $Z_{ij}\in\{0,1\}$
obeying  $Z_{ij}Z_{i\ell}=0$ for  $j\neq\ell$. I specify
\[
\left(\begin{array}{c}
1-\sum_{j=1}^{m}Z_{ij}\\
Z_{i}
\end{array}\right)\sim\text{Multinomial}\left(1,\,\left(\begin{array}{c}
1-\sum_{j=1}^{m}q_{j}\\
q
\end{array}\right)\right),\quad q_{j}>0\text{ for \ensuremath{j=1,\dots,m}}.
\]

This design can be thought of as a setting where a single fundamental
instrument $\sum_{j=1}^{m}Z_{ij}$ has been fully interacted with
 $m$ group indicators. The term  $b_{m+1}$ in the DGP is a common baseline direct effect. Thus, the complete instrument set for the parameter vector $(\alpha_{0},\beta_{0})$ would be $\left(1,Z_{i}'\right)'=\left(1,Z_{i1},\ldots,Z_{im}\right)'$. To simplify the analysis, I will treat the intercept $\alpha_{0}$ as nuisance and target the slope $\beta_{0}=\theta_0$, making $p=1$.

Note that, unlike the judge-IV design discussed in Section \ref{subsec:Intercepts-in-linear}
--- where every observation belongs to exactly one of the $m$ instrument
categories and $\sum_{j=1}^{m}Z_{ij}=1$ --- here $\sum_{j=1}^{m}Z_{ij}\in\{0,1\}$
because the instruments are interactions of a binary fundamental instrument
with group indicators. Observations with $\sum_{j=1}^{m}Z_{ij}=0$
break the adding-up constraint.

The excludability
violations $\left(b_{1},\dots,b_{m+1}\right)$ follow a multivariate  $t$ distribution with 5 degrees
of freedom, location  $\bar{\mu}\mathbf{1}_{m+1}$, and scale matrix $\frac{3}{5}\bar{\sigma}^{2}I_{m+1}$, implying covariance matrix $\bar{\sigma}^{2}I_{m+1}$. Therefore, these violations have exactly four moments. The
excludability violations are exchangeable but not  i.i.d.: the diagonal covariance matrix implies the entries are not correlated but allows for dependence at higher moments.

Both $\pi^{*}$ and  $q$ are deterministic functions of  $m$ and are
held fixed across all Monte Carlo replications. Group membership probabilities
are set to quantiles of a log-normal distribution restricted to a fixed range. I first generate  $q_{j}^{\mathrm{raw}}=\exp\!\bigl(\sigma_{q}\,\Phi^{-1}(u_{j})\bigr)$
with $\Phi$ the standard normal CDF and  $\sigma_{q}=0.83$, where the
 $u_{j}$ are evenly spaced over the fixed interval  $[\Phi(-2.231),\Phi(2.231)]$.
I then normalize the raw values so that  $\sum_{j=1}^{m}q_{j}=0.9$.
Because the range is held fixed, the cell-size spread  $\max_{j}q_{j}/\min_{j}q_{j}\approx40{:}1$
is the same for every  $m$. This is an infill design: as  $m$ grows, cells fill in
within a fixed range of sizes rather than spreading toward ever more extreme values.

The $\pi^{*}$ coefficients
use the same quantile scheme. I first generate  $\pi_{j}^{*,\mathrm{raw}}=\exp\!\bigl(\sigma_{\pi^{*}}\,\Phi^{-1}(u_{j})\bigr)$
with  $\sigma_{\pi^{*}}=0.67$, giving a fixed  $\approx20{:}1$ spread. I then rescale
the raw values so that  $\|\pi^{*}\|=3\sqrt{m/40}$,
which preserves average instrument strength as  $m$ grows. To break
correlation with group sizes,  $\pi^{*}$ values are assigned in an interleaved
pattern: odd-indexed groups receive  $\pi^{*}$ values in ascending order
and even-indexed groups in descending order. 

In the baseline parameterization, I set  $\alpha_{0}=1$,  $\beta_{0}=1$,
 $\rho=0.1$,  $\bar{\mu}=8$,  $\bar{\sigma}=8$, and  $\sigma_{e}^{2}=1$.
Recall that $b$ is scaled by  $n^{-1/2}$. Hence, the coefficients
capturing excludability violations have mean  $8/\sqrt{n}$ and variance
 $64/n$. 

\subsection{Moment conditions and estimator}\label{subsec:sim_moment}
To target the slope parameter $\beta_0$, I project out the nuisance intercept by demeaning the outcome, treatment, and instruments in sample to form $\tilde{Y}_{i}$, $\tilde{T}_{i}$, and $\tilde{Z}_{i}=Z_{i}-\hat q$, where $\hat q=n^{-1}\sum_{i=1}^n Z_i$. Continuing the convention introduced in Section \ref{sec:Examples}, let $E_{n}[\cdot]=E[\cdot \mid b_{1},\dots,b_{m+1}]$ denote expectations under the DGP conditional on the specification errors. The population moment condition I leverage for estimation can be written
\[
g\left(\theta\right)=S^{-1}E_{n}[(\tilde{Y}_{i}-\theta\tilde{T}_{i})\tilde{Z}_{i}], \qquad S:=E_{n}[\tilde{Z}_{i}\tilde{Z}_{i}'].
\]
Note that, following Example \ref{exa:IV_example}, I have rescaled by the inverse of the second moment matrix $S$ of residualized instruments, in order to avoid the requirement that this matrix takes a permutation equivariant form. 

It is useful to rewrite this condition as $g\left(\theta\right)=\delta - \theta \pi$, where $\delta:=S^{-1}E_n[\tilde Z_i \tilde Y_i]$ is a vector of reduced form coefficients and $\pi:=S^{-1} E_n[\tilde Z_i \tilde T_i]$ the corresponding vector of first-stage coefficients. Note that 
\[
E_n[\tilde Z_i \tilde Y_i] = \beta_0 E_n[\tilde Z_i \tilde T_i] + Sb / \sqrt{n}, \qquad b:=(b_{1},\dots,b_{m})',
\]
where the baseline violation $b_{m+1}$ was eliminated by demeaning. Thus,
\[
g(\beta_0) = \delta - \beta_0 \pi = b/ \sqrt{n}.
\]
The vector $b$ has a $t_5$ distribution, which is exchangeable with four finite moments. It therefore satisfies Assumption~\ref{assu:random_b}.ii. The realized hyperparameters $(\mu,\sigma^{2})$ targeted by the estimators $(\hat \mu, \widehat{\sigma^{2}})$
are the sample mean and variance of $b$.

The population Jacobian is the $m$-vector $G=-\pi$. Since the entries of $\pi$ are distinct, $1_{m}$ is not in the column space of $G$ and $\mu$ is identified by equation \eqref{eq:identification} at each $m$. The infill design of Section~\ref{subsec:Simulation-design} holds the spread of the first stage fixed as $m$ grows. The normalized residual $m^{-1}\left\Vert M1_{m}\right\Vert ^{2}$ therefore stays bounded away from zero, converging to a positive $\kappa$. This confirms Assumption~\ref{assu:regime}.vi.

In conducting GMM estimation, I use the weighting matrix $\hat{W}=\hat S := n^{-1} \sum_{i=1}^n \tilde{Z}_{i}\tilde{Z}_{i}'$, which ensures the estimates are numerically equivalent to TSLS estimation of the residualized system. By the Frisch--Waugh--Lovell theorem, the TSLS estimate of $\beta$ in the residualized system is numerically equivalent to the TSLS slope from  estimation of $(\alpha,\beta)$ using the full instrument set $(1,Z_{i}')'$. Since the structural errors $e_i$ are homoscedastic, this TSLS weighting would be efficient under proper specification.

I verify analytically in Appendix~\ref{app:verify} that the DGP satisfies
the outstanding conditions of Assumptions~\ref{assu:regime} and \ref{assu:rate} required for consistency of the BC and EB estimators when $m$ grows with $n$. The conditions of Assumptions~\ref{assu:eb-reg} and \ref{assu:lindeberg} supporting inference are checked there as well, analytically where possible and numerically otherwise. 

\subsection{Baseline findings}

Table~\ref{tab:mc-main} reports the bias, standard deviation, and
root mean squared error (RMSE) of each estimator in a first simulation
design with  $n=10{,}000$ observations and  $m=40$ cell moments.
All results report estimation performance for the slope coefficient
 $\beta$ and are averaged across  $1{,}000$ Monte Carlo replications.
The feasible estimators use the hyperparameter estimates  $\hat{\mu}$
and  $\widehat{\sigma^{2}}$, while the oracle estimators use the true
hyperparameters $\left(\mu,\sigma^{2}\right)$ based on the realized specification errors but continue to rely on the estimated matrices  $\hat{G}$,  $\hat{V}$, and  $\hat{W}$. 

\begin{table}[H]
\caption{Estimator Performance ( $m=40$,  $n=10,000$) \protect\label{tab:mc-main}}

\begin{centering}
\begin{tabular}{lcccc}
\toprule
Estimator & Bias & Std & RMSE & Rej.\tabularnewline
\midrule
GMM & 0.220 & 0.185 & 0.287 & 0.386\tabularnewline
BC & 0.040 & 0.251 & 0.254 & 0.087\tabularnewline
EB & 0.038 & 0.239 & 0.242 & 0.088\tabularnewline
Oracle BC & 0.024 & 0.181 & 0.182 & 0.043\tabularnewline
Oracle EB & 0.023 & 0.178 & 0.180 & 0.046\tabularnewline
\bottomrule
\end{tabular}
\par\end{centering}
{\small\emph{Notes:}}{\small{} Bias, standard deviation, and RMSE of
 $\hat{\beta}$ across  $1{,}000$ replications,  $n=10{,}000$,  $m=40$.
Feasible estimators use the hyperparameter estimates described in
Section~\ref{sec:estimation}. ``Oracle'' estimators use the true
 $\left(\mu,\sigma^{2}\right)$ with sample matrices  $\hat{G}$,
 $\hat{V}$,  $\hat{W}$. Mean simulated  $J=90.7$ with  $39$ degrees
of freedom. The population leading-term risk ratio of EB to
BC, $V_{EB}/V_{BC}=1-\Lambda\Sigma_{\mu}M'(M\Sigma_{\mu}M')^{+}M\Sigma_{\mu}\Lambda'/(\Lambda\Sigma_{\mu}\Lambda')=0.936$,
is computed exactly from the design parameters. Rej. is the rejection rate of a two-sided
Wald test of the null hypothesis that  $\beta_{0}=1$ at the 5\% nominal
level.}{\small\par}
\end{table}

By Proposition~\ref{prop:variance-ordering}, the EB estimator improves
on the leading-term risk of BC whenever $M\Sigma_{\mu}\Lambda'\neq0$. With $p=1$, the ratio of the EB and BC leading-term risks is governed by a quadratic form in this quantity:
\[
\frac{V_{EB}}{V_{BC}}=1-\frac{\Lambda\Sigma_{\mu}M'(M\Sigma_{\mu}M')^{+}M\Sigma_{\mu}\Lambda'}{\Lambda\Sigma_{\mu}\Lambda'}.
\]
Evaluating this expression using the population matrices $(\Lambda,M,\Sigma_{\mu})$ yields 0.936, indicating a modest efficiency advantage of EB over BC in this design. The mean simulated  $J$-statistic is
90.7, far exceeding the 5\% critical value of 54.57. Faced with this
DGP, the overidentification test will tend to reject correct specification.

The sizeable specification errors present in this DGP lead the uncorrected
GMM estimator to exhibit a bias of 0.220 that accounts for the bulk
of its RMSE of 0.287. The bias-corrected estimator  $\hat{\beta}_{BC}$
reduces the bias to 0.040 and achieves an RMSE of 0.254, 12\% below
GMM. The feasible shrinkage estimator  $\hat{\beta}_{EB}$ improves
on this, attaining an RMSE of 0.242---16\% below GMM and about 5\%
below BC---by removing the predictable component of the bias-corrected
estimation error. The ratio of EB to BC Monte Carlo variances is $(0.239/0.251)^2=0.907$, close to the leading-term variance ratio of 0.936. Evidently, the asymptotic approximation provides an empirically useful guide to finite-sample behavior under this DGP.

The oracle estimators illustrate the gains available with
known hyperparameters: oracle BC reaches an RMSE of 0.182 and oracle
EB 0.180. Both oracle estimators exhibit small biases that capture the remaining TSLS many-instruments bias that would be present even under proper specification. The narrow gap in RMSE between the two oracle estimators
suggests that shrinkage yields limited gains when $\mu$ is already known. I document below that this advantage grows when the specification errors are more dispersed.

The ``Rej.'' column of Table~\ref{tab:mc-main} reports the rejection
rate of a two-sided Wald test of the null hypothesis that  $\beta_{0}=1$
at the 5\% nominal level. Thus, rejections correspond here to type
I errors. The oracle
BC and EB rejection rates are based on the following infeasible variance
matrices 
\[
V_{OBC}=\hat{\Lambda}\left(\hat{V}+\sigma^{2}Q\right)\hat{\Lambda}',\quad V_{OEB}=\hat{\Lambda}\left(I_{m}-\hat{\Pi}(\sigma^{2})\right)\left(\hat{V}+\sigma^{2}Q\right)\left(I_{m}-\hat{\Pi}(\sigma^{2})\right)'\hat{\Lambda}',
\]
which use the sample sensitivity matrix  $\hat{\Lambda}$, the sample
moment-covariance matrix  $\hat{V}$, and the true $\sigma^{2}$. Here $\hat{\Pi}(\sigma^{2})$ is the sample projection evaluated at the true variance rather than at $\widehat{\sigma^{2}}$. Because the oracle uses the true $\mu$, the outer covariance $\hat{V}+\sigma^{2}Q$ is not centered by $\hat{M}_{\mu}$. The projection $\hat{\Pi}(\sigma^{2})$ itself is built with $\hat{M}_{\mu}$, as defined before Assumption~\ref{assu:eb-reg}.

The standard GMM Wald test over-rejects dramatically, with
a rejection rate of 38.6\%. This overrejection stems both from the
estimator's bias and use of the conventional sandwich variance estimator
that ignores misspecification. In contrast, the oracle BC and EB rejection rates are 4.3\% and 4.6\% respectively,
both close to the nominal level. The feasible BC and EB tests over-reject
modestly---8.7\% and 8.8\% respectively---reflecting the finite-sample cost of
estimating the hyperparameters  $\mu$ and  $\sigma^{2}$.

Appendix~\ref{app:differenced} reports parallel results for estimators
that difference the mean specification error away rather than estimate
it, implementing the strategy of Section~\ref{sec:A-differencing-approach},
along with the MBTSLS estimator of \textcite{kolesar2015identification}. Appendix Table \ref{tab:appc-differenced} shows that differenced GMM is nearly equivalent to levels BC and differenced EB is nearly equivalent to levels EB.  MBTSLS is biased in levels
because the nonzero mean $\bar{\mu}$ violates the orthogonality condition
it relies on. Applying MBTSLS to the differenced moment conditions yields the lowest bias of any estimator, which reflects that differenced MBTSLS removes not only the leading-term bias targeted by BC but the higher-order many-instruments bias. However, this extra bias reduction comes at a cost: the differenced MBTSLS estimator's standard deviation is nearly 50\% above that of EB, putting it at a severe RMSE disadvantage relative to both corrected estimators. 

\subsection{Altering the misspecification hyperparameters}

Table \ref{tab:mc-rmse-vs-mu} examines how the relative performance of the estimators, as measured by RMSE, varies with the marginal mean specification error  $\bar{\mu}$ while holding $\bar{\sigma}=8$ and all other parameters fixed. Theoretically,
the BC, Oracle BC, EB, and Oracle EB estimators are invariant to $\bar{\mu}$ in large samples. A unit shift in $\bar{\mu}$ moves $\hat{\mu}$ by one in expectation and shifts the GMM estimate by $(1/\sqrt{n})\hat{\Lambda}1_{m}$. The bias correction $(\hat{\mu}/\sqrt{n})\hat{\Lambda}1_{m}$ removes this shift, and the shrinkage correction is likewise unaffected. Only the uncorrected GMM estimator's RMSE genuinely depends
on $\bar{\mu}$. The mild variation visible in the corrected-estimator
rows of the table reflects a mix of finite-sample deviations from asymptotic invariance and Monte Carlo noise.

\begin{table}[H]
    \caption{RMSE by $\bar{\mu}$ ( $\bar{\sigma}=8$,  $n=10{,}000$,  $m=40$)
    \protect\label{tab:mc-rmse-vs-mu}}
    
    \begin{centering}
    \begin{tabular}{lccccc}
    \toprule
     & $\bar{\mu}=0$ & $\bar{\mu}=2$ & $\bar{\mu}=4$ & $\bar{\mu}=8$ & $\bar{\mu}=16$\tabularnewline
    \midrule
    GMM & 0.188 & 0.200 & 0.222 & 0.287 & 0.459\tabularnewline
    BC & 0.256 & 0.256 & 0.255 & 0.254 & 0.246\tabularnewline
    EB & 0.241 & 0.243 & 0.241 & 0.242 & 0.237\tabularnewline
    Oracle BC & 0.185 & 0.186 & 0.182 & 0.182 & 0.181\tabularnewline
    Oracle EB & 0.179 & 0.181 & 0.179 & 0.180 & 0.179\tabularnewline
    \midrule
    $J$-stat & 88.6 & 90.8 & 92.7 & 90.7 & 99.5\tabularnewline
    $V_{EB}/V_{BC}$ & 0.936 & 0.936 & 0.936 & 0.936 & 0.936\tabularnewline
    \bottomrule
    \end{tabular}
    \par\end{centering}
    {\small\emph{Notes: }}{\small RMSE of $\hat{\beta}$ across 1,000 Monte
    Carlo replications with $m=40$,  $n=10{,}000$, and $\bar{\sigma}=8$
    held fixed. The final two rows report the mean simulated $J$-statistic
    and the population leading-term risk ratio of EB to BC,
    $V_{EB}/V_{BC}=1-\Lambda\Sigma_{\mu}M'(M\Sigma_{\mu}M')^{+}M\Sigma_{\mu}\Lambda'/(\Lambda\Sigma_{\mu}\Lambda')$,
    computed exactly from the design parameters.}{\small\par}
\end{table}

At $\bar{\mu}=0$ the systematic component of the bias is absent. Oracle BC removes the realized mean without shrinking and barely improves on GMM, with an RMSE of 0.185. The further drop to 0.179 under Oracle EB reflects the advantages of shrinkage using the true $\sigma^2$. The feasible corrections pay the cost of
estimating the hyperparameters. This leaves GMM (0.188) ahead of BC (0.256)
and EB (0.241) by a wide margin. In contrast, at  $\bar{\mu}=8$,
the feasible corrections dominate: EB achieves 0.242 versus 0.287 for GMM.
At  $\bar{\mu}=16$, GMM degrades to 0.459 while BC and EB stay around
0.24--0.25. When the systematic bias is large, the corrected methods offer substantial improvements in RMSE.

Table~\ref{tab:mc-rmse-vs-sig2} examines how the relative performance
of the estimators depends on the marginal dispersion
$\bar{\sigma}$ of the specification errors while holding  $m=40$,
 $\bar{\mu}=8$, and all other parameters fixed. The mean simulated  $J$-statistic rises from 41.8 at  $\bar{\sigma}=0$, which is below the 5\% critical value of 54.57,
to 238.9 at  $\bar{\sigma}=16$. The final row reports the ratio of the EB and BC leading-term risks, $V_{EB}/V_{BC}$, computed exactly from the design parameters. The ratio is stable near 0.95 for $\bar{\sigma}\le8$ and falls to 0.861 at $\bar{\sigma}=16$, reflecting growing asymptotic gains from shrinkage as dispersion rises.

\begin{table}[H]
    \caption{RMSE by $\bar{\sigma}$ ( $\bar{\mu}=8$,  $n=10{,}000$,  $m=40$)
    \protect\label{tab:mc-rmse-vs-sig2}}
    
    \begin{centering}
    \begin{tabular}{lccccc}
    \toprule
     & $\bar{\sigma}=0$ & $\bar{\sigma}=2$ & $\bar{\sigma}=4$ & $\bar{\sigma}=8$ & $\bar{\sigma}=16$\tabularnewline
    \midrule
    GMM & 0.264 & 0.266 & 0.272 & 0.287 & 0.345\tabularnewline
    BC & 0.192 & 0.196 & 0.208 & 0.254 & 0.395\tabularnewline
    EB & 0.188 & 0.192 & 0.204 & 0.242 & 0.345\tabularnewline
    Oracle BC & 0.147 & 0.149 & 0.157 & 0.182 & 0.275\tabularnewline
    Oracle EB & 0.153 & 0.155 & 0.161 & 0.180 & 0.247\tabularnewline
    \midrule
    $J$-stat & 41.8 & 45.3 & 55.7 & 90.7 & 238.9\tabularnewline
    $V_{EB}/V_{BC}$ & 0.943 & 0.947 & 0.953 & 0.936 & 0.861\tabularnewline
    \bottomrule
    \end{tabular}
    \par\end{centering}
    {\small\emph{Notes: }}{\small RMSE of $\hat{\beta}$ across 1,000 Monte
    Carlo replications with $m=40$,  $n=10{,}000$, and $\bar{\mu}=8$
    held fixed. The column $\bar{\sigma}=8$ is the baseline. The final
    two rows report the mean simulated $J$-statistic and the population
    leading-term risk ratio of EB to BC,
    $V_{EB}/V_{BC}=1-\Lambda\Sigma_{\mu}M'(M\Sigma_{\mu}M')^{+}M\Sigma_{\mu}\Lambda'/(\Lambda\Sigma_{\mu}\Lambda')$,
    computed exactly from the design parameters.}{\small\par}
\end{table}

At $\bar{\sigma}=0$ (homogeneous specification error), the feasible
estimators are essentially tied (EB 0.188 vs.~BC 0.192), as are the oracle estimators (OEB 0.153 vs.~OBC 0.147). Recall from Section \ref{subsec:ordering} that with no idiosyncratic overdispersion to exploit and efficient weighting,
the precision gains from shrinkage derive entirely from heteroscedastic moment noise. The final row puts $V_{EB}/V_{BC}$ at 0.943, about a 3\% reduction in RMSE, consistent with the near tie. As the marginal variance
$\bar{\sigma}^2$ of specification errors grows, the signal from the specification errors
strengthens relative to sampling noise, making shrinkage increasingly
beneficial. At  $\bar{\sigma}=16$, the EB estimator has roughly 13\% lower RMSE than BC. In contrast, the asymptotic approximations suggest the RMSE ratio of these estimators should be driven entirely by their leading-term variance ratio, yielding a reduction of $1-\sqrt{0.861}\approx7\%$. Thus, the leading term underestimates the advantages of EB in the extreme heterogeneity design with $\bar{\sigma}=16$.

The RMSE of GMM increases monotonically with $\bar{\sigma}$. The advantage of BC over GMM decreases in $\bar{\sigma}$, as larger specification errors raise the precision cost of the first-order bias correction. At $\bar{\sigma}=16$, GMM dominates BC because the precision costs eventually outweigh the gains of bias correction. The EB estimator dampens this precision cost and is the overall best-performing feasible estimator in Table~\ref{tab:mc-rmse-vs-sig2}. However, in the presence of extreme idiosyncratic dispersion ($\bar{\sigma}=16$) its RMSE is tied with GMM.

Appendix Table \ref{tab:appc-sweeps} reports corresponding results for differenced estimators and MBTSLS. Differenced GMM and EB track levels BC and EB
closely at every setting of the hyperparameters. Evidently, estimating $\mu$
costs about as much as eliminating it. Differenced MBTSLS is dominated by differenced EB at all hyperparameter values, reflecting the more severe precision costs of correcting for higher-order many-instruments bias.

\subsection{Asymptotic behavior}\label{subsec:asymptotic}

To assess the rate predictions of Propositions~\ref{prop:bc-rate}
and \ref{prop:eb-rate}, I vary  $n$ and  $m$ jointly, setting  $m$
equal to  $n^{0.4}$ rounded to the nearest integer. This choice ensures
that  $m^{2}/n\to0$ as required by Assumption~\ref{assu:regime}.
Table~\ref{tab:mc-bias-vs-n} reports  the bias and RMSE of all five
estimators across seven sample sizes from  $n=5{,}000$ to  $n=500{,}000$. Since the specification errors are scaled by $1/\sqrt{n}$, all of the estimators are consistent. It is convenient then to multiply both the bias and RMSE by $\sqrt{n}$ so that differences in estimator performance do not collapse as the sample size grows. 

The  scaled bias of GMM grows modestly with $n$ but levels off around $25$, indicating a limiting $\sqrt{n}$ rate. In contrast, the scaled  biases of BC and EB gradually drift towards zero, falling from about  $4$  at  $n=5{,}000$  to roughly  $3$
at  $n=500{,}000$. The oracle estimators do slightly better than their feasible counterparts but exhibit the same drift towards zero bias. These gradual drops in the scaled bias of the corrected estimators reflect the many-instruments bias of TSLS that would be present even under proper specification. Since this bias is of order $m/n$, the scaled bias reported in the table is of order $m/\sqrt{n}=n^{-0.1}$.

\begin{table}[H]
    \centering
    \caption{Bias and RMSE as Functions of  $n$ with  $m=\mathrm{round}(n^{0.4})$\protect\label{tab:mc-bias-vs-n}}
    
    \begin{centering}
    \setlength{\tabcolsep}{3pt}%
    \begin{tabular}{rccccccccc ccccc}
    \toprule 
     &  &  &  & \multicolumn{5}{c}{$\sqrt{n}\times$Bias} &  & \multicolumn{5}{c}{$\sqrt{n}\times$RMSE}\tabularnewline
    \midrule 
    $n$ & $m$ & $V_{EB}/V_{BC}$ &  & GMM & BC & EB & OBC & OEB &  & GMM & BC & EB & OBC & OEB\tabularnewline
    5,000 & 30 & 0.924 &  & 21.3 & 4.2 & 4.5 & 3.4 & 3.5 &  & 28.8 & 25.0 & 23.8 & 19.1 & 18.6\tabularnewline
    10,000 & 40 & 0.936 &  & 22.0 & 4.0 & 3.8 & 2.4 & 2.3 &  & 28.7 & 25.4 & 24.2 & 18.2 & 18.0\tabularnewline
    20,000 & 53 & 0.945 &  & 24.0 & 4.7 & 4.2 & 2.9 & 2.3 &  & 29.5 & 24.2 & 23.2 & 17.3 & 17.3\tabularnewline
    50,000 & 76 & 0.953 &  & 24.5 & 3.7 & 3.4 & 2.4 & 2.0 &  & 30.3 & 23.4 & 22.3 & 17.7 & 17.4\tabularnewline
    100,000 & 100 & 0.956 &  & 24.5 & 2.3 & 2.4 & 1.7 & 1.7 &  & 29.9 & 23.4 & 22.9 & 17.1 & 17.5\tabularnewline
    200,000 & 132 & 0.958 &  & 25.5 & 3.6 & 3.3 & 2.5 & 2.2 &  & 30.6 & 22.7 & 22.2 & 16.9 & 17.4\tabularnewline
    500,000 & 190 & 0.959 &  & 25.6 & 3.2 & 3.1 & 2.2 & 2.0 &  & 30.3 & 22.3 & 21.8 & 16.3 & 16.7\tabularnewline
    \bottomrule
    \end{tabular}
    \par\end{centering}
    \raggedright{\small\emph{Notes:}}{\small{} Entries are  $\sqrt{n}$  times the bias and RMSE of  $\hat{\beta}$ across
    1,000 replications per grid point. The population leading-term risk
    ratio $V_{EB}/V_{BC}$ is computed exactly from the design parameters.}{\small\par}
\end{table}

The scaled RMSE of GMM rises slightly before stabilizing around 30. The rise reflects its scaled bias, which grows with $n$ even as the scaled variance falls. In contrast, the scaled RMSE of the corrected estimators edges down as $n$ grows before stabilizing. The bias of these estimators is already small, so this decline primarily reflects a falling scaled variance. The leveling off in RMSE of the corrected estimators is
consistent with the parametric-rate predictions of Propositions \ref{prop:bc-rate}
and \ref{prop:eb-rate}.

GMM has the highest RMSE at every sample size. By Proposition~\ref{prop:variance-ordering}, the risk ratios predict that EB improves on BC throughout the grid: $V_{EB}/V_{BC}$ rises from 0.924 at $m=30$ to 0.959 at $m=190$. Consistent with this prediction, EB exhibits lower RMSE than BC at all sample sizes, with a gap that narrows as $m$ grows.  The oracle estimators have the lowest RMSEs, and the two perform comparably to each other at large sample sizes.

Figure~\ref{fig:mc-rej-vs-m} plots rejection rates across the Table~\ref{tab:mc-bias-vs-n}
grid. GMM's bias leads to over-rejection at all sample sizes, hovering
between 38\% and 41\% type I error rates as  $m$ grows large. In
contrast, both the BC and EB rejection rates quickly approach the nominal
5\% level as $m$ grows large. This finding is consistent with Appendix Table \ref{tab:appb-diagnostics}, which verifies that the BC and EB estimators of $\beta$ satisfy the no-dominant-column condition of Assumption~\ref{assu:lindeberg}.iii. The oracle BC and EB rejection rates also hover near the nominal
5\% level across all $m$, with OBC and OEB exhibiting rejection rates of 5.3\% and 4.3\% respectively at $m=190$.

\begin{figure}[H]
    \begin{centering}
    \caption{Wald test rejection rates vs. $m$ with $m=\mathrm{round}(n^{0.4})$\protect\label{fig:mc-rej-vs-m}}
    \includegraphics[width=0.85\textwidth]{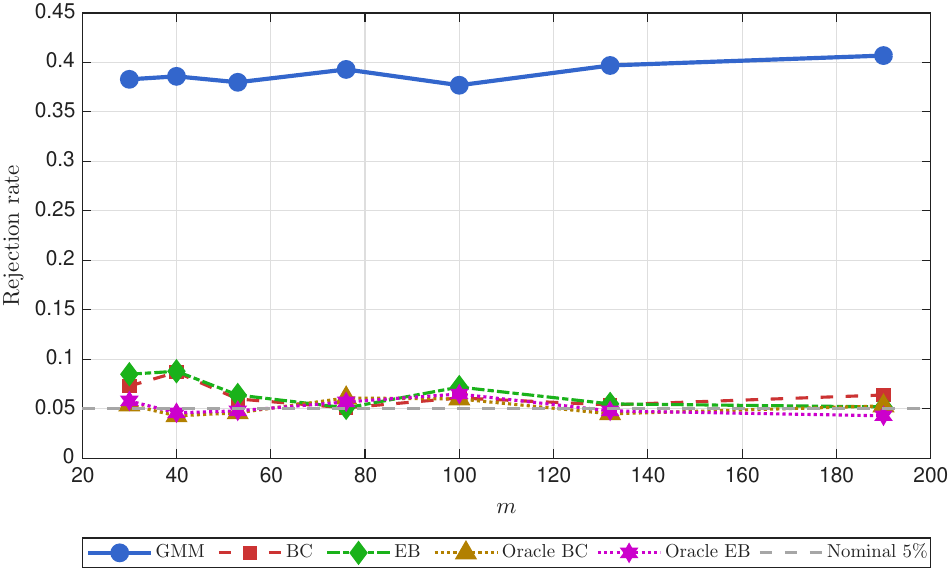}
    \par\end{centering}
    \raggedright{\small\emph{Notes:}}{\small{} Figure depicts rejection
    rate of Wald test that $\beta_{0}=1$ for each estimator. Nominal
    level of each test is 5\%. 1,000 simulations used per grid point.}{\small\par}
\end{figure}

\section{Empirical Application: Returns to Schooling \protect\label{sec:empirical}}

In this section, I revisit the work of \textcite{angrist1991does},
who used quarter-of-birth (QOB) indicators as instruments for years
of education in a log weekly wage equation to estimate the returns
to schooling. Table VII of their paper reports TSLS estimates based
on interactions between QOB and state-of-birth (SOB) as excluded instruments.
The same specification also excludes interactions between QOB and
year-of-birth (YOB), yielding 180 instruments in total. 

I will simplify this instrument set in two respects with the dual
aims of ensuring exchangeability and avoiding potential biases that
can arise under many-instrument asymptotics. First, I collapse the
three QOB indicators into a single $\mathbf{1}\left\{ \text{QOB}>1\right\} $
indicator for being born in one of the last three quarters, interacted
with SOB. \textcite{angrist1991does} work with the same binarized
version of QOB in their Table III, which they show yields similar
results to specifications leveraging the complete set of QOB indicators
as instruments via TSLS. Second, I drop the season of birth interactions
with YOB. This yields 51 instruments (one per state). 

A key identifying assumption of \textcite{angrist1991does}'s study
is that season of birth is correlated with earnings only through its
causal effect on educational attainment. This claim was questioned
early on by \textcite{bound1995problems}, who considered the potential
effects of small violations of QOB exogeneity. Later work by \textcite{buckles2013season}
provided evidence that economically disadvantaged families are more
likely to have children in the first quarter of the year, suggesting
the average earnings of children born in later quarters would be higher
even if their years of schooling were equalized. Since \textcite{angrist1991does}
establish that individuals born in later quarters tend to complete
more schooling, the family background differences highlighted by \textcite{buckles2013season}
may lead QOB-based IV estimates to overstate the returns to schooling
when instruments have a positive first stage. In what follows, I investigate
whether the BC and EB estimators reach similar conclusions about the
likely bias in TSLS.

\subsection{Estimation sample and moment conditions}

The estimation sample consists of men from the 1940--1949 birth cohorts
in the 1980 Census 5\% Public Use Micro Sample (PUMS), yielding $n=486{,}926$
observations. The instruments are $m=51$ mutually exclusive indicators
of the form $Z_{ij}=\mathbf{1}\left\{ \text{QOB}_{i}>1,\,\text{SOB}_{i}=j\right\} $
for each state $j=1,\dots,51$. Each instrument flags whether
individual $i$ was born in a non-Q1 quarter in a particular state.
Note that $\sqrt{n}\approx698$ and $m^{2}/n\approx0.005$, which
suggests the asymptotic regime of Assumption \ref{assu:regime}.i
is easily satisfied.

I consider four specifications of controls $X_{i}$ mirroring those
entertained by \textcite{angrist1991does}: (i) SOB dummies only,
(ii) SOB dummies with age (in quarters) and age-squared, (iii) SOB
dummies with individual-level covariates (indicators for race, marital
status, SMSA residence, and eight census region dummies), and (iv)
SOB dummies with all covariates plus age and age-squared. Rather than
include a constant, I always use 51 state-of-birth dummies that sum
to one. 

The endogenous variable is completed years of schooling $T_{i}$.
The scalar parameter of interest is $\theta$, which measures the
returns to an additional year of schooling. Since the controls are
nuisance parameters, I partial them out. Thus, in all four specifications
the degree of overidentification is $m-p=50$. As the simulations
of the previous section demonstrate, this level of overidentification
can yield sufficiently accurate hyperparameter estimates for corrected
estimators to generate non-trivial improvements over GMM.

Let $Y_{i}$ denote log weekly earnings, and let $\tilde{Y}_{i}$,
$\tilde{T}_{i}$, and $\tilde{Z}_{i}$ denote the residuals from OLS
regressions of $Y_{i}$, $T_{i}$, and $Z_{i}$ on $X_{i}$. I work with
sample moment conditions of the form
\[
\hat{g}\left(\theta\right)=\hat{S}^{-1}\left(\frac{1}{n}\sum_{i=1}^{n}\left(\tilde{Y}_{i}-\theta\tilde{T}_{i}\right)\tilde{Z}_{i}\right),\quad \hat{S}=\frac{1}{n}\sum_{i}\tilde{Z}_{i}\tilde{Z}_{i}'.
\]
As explained in Section \ref{subsec:sim_moment}, rescaling by the inverse of the second moment matrix of residualized instruments $\hat{S}$ ensures that these moment conditions capture excludability violations when evaluated at the true parameter $\theta_0$. Here, the excludability violations capture direct effects of being born outside
the first quarter in state $j$, net of the controls $X_{i}$. Exchangeability
across states of these direct effects is plausible ex ante: the family-composition mechanism
of \textcite{buckles2013season} could operate similarly in each state.

The sample Jacobian is the rescaled first-stage moment $\hat{G}=-\hat{S}^{-1}(\frac{1}{n}\sum_{i}\tilde{T}_{i}\tilde{Z}_{i})$. Since non-Q1 births have more education on average, most of its entries are negative.
The vector $1_{m}$ is in the column space of $\hat{G}$ only if the
first stages are common across states. Fortunately, cross-state variation
in compulsory schooling laws yields substantial variation in first-stage strength
across states, which was \textcite{angrist1991does}'s motivation for
studying SOB interactions.

\subsection{OLS and TSLS Results}

Table~\ref{tab:ak91-replication} reports OLS and TSLS estimates
for the non-Q1~$\times$~SOB instrument design. Without age controls,
the TSLS estimates of the return to education are around 4\%, below
the OLS estimates of approximately 5\%. Adding age controls raises
the TSLS estimates sharply, to 12--14\%, well above OLS. This sensitivity
to the inclusion of age reflects the mechanical correlation between
age and QOB: non-Q1 individuals are younger on average. \textcite{rosenzweig2000natural}
argued that this mechanical dependence led to excludability difficulties
as younger workers have less labor market experience and therefore
lower potential wages. The sizable increase in the TSLS estimates
after controlling for the age quadratic is consistent with this view. 

\begin{table}[H]
\caption{TSLS Estimates: Returns to Education with non-Q1~$\times$~SOB Instruments
\protect\label{tab:ak91-replication}}

\begin{centering}
\begin{tabular}{lcccc}
\toprule 
 & SOB & SOB + age & SOB + cov & SOB + cov + age\tabularnewline
\midrule 
OLS & 0.0535 & 0.0555 & 0.0495 & 0.0513\tabularnewline
 & (0.0004) & (0.0004) & (0.0003) & (0.0003)\tabularnewline
TSLS & 0.0424 & 0.1413 & 0.0360 & 0.1240\tabularnewline
 & (0.0188) & (0.0246) & (0.0193) & (0.0251)\tabularnewline
$J$ {[}d.f.{]} & 67.32 {[}50{]} & 46.64 {[}50{]} & 58.66 {[}50{]} & 45.65 {[}50{]}\tabularnewline
$p$-value & 0.052 & 0.609 & 0.188 & 0.649\tabularnewline
\bottomrule
\end{tabular}
\par\end{centering}
{\small\emph{Notes:}}{\small{} Sample is 486,926 men from the 1940--1949
birth cohorts in the 1980 Census 5\% PUMS. Instruments are 51 $\mathbf{1}\left\{ \text{QOB}>1\right\} $~$\times$~SOB
interactions. Heteroscedasticity-robust standard errors in parentheses.
 $J$-statistic degrees of freedom in brackets.}{\small\par}
\end{table}

In contrast to the Monte Carlo simulations, these data exhibit minimal
evidence of misspecification. The $J$-test borderline rejects correct
specification in one specification without age controls ($p=0.052$)
and fails to reject in the other ($p=0.188$). Both specifications
with age controls clearly fail to reject ($p=0.609$ and $p=0.649$).
I will show that misspecification-aware point estimates and standard
errors nonetheless offer additional insight in this environment.

To anchor the subsequent misspecification-aware analysis to these
TSLS specifications, I weight GMM by $\hat{W}=\hat{S}$. Because the
moments are already rescaled by $\hat{S}^{-1}$, this weight reproduces
the TSLS point estimates exactly and is efficient under correct specification
and homoscedasticity of the earnings errors. I estimate moment uncertainty
with the heteroscedasticity-robust sandwich estimator $\hat{V}$,
computed from the outer product of the rescaled moment contributions
to account for estimation of the controls.

\subsection{Scrutinizing exchangeability}\label{subsec:Scrutinizing}

To assess the plausibility of exchangeability in these data, I apply the
specification test in \eqref{stat:feature} of Section~\ref{subsec:Diagnostic} using as features the centered
inverse-root cell frequency $\hat{q}_{j}^{-1/2}$, together with centered
indicators for the four Census regions of state $j$ (Northeast, Midwest, South,
West). The first feature tests whether the specification errors covary with the
noise of the state cell moments---whose standard deviation scales with
$\hat{q}_{j}^{-1/2}$---as in Example~\ref{exa:mean}. The region indicators test
whether the direct effects of season of birth are systematically larger in some
parts of the country than others, which would arise, for example, if the
family-composition channel of \textcite{buckles2013season} operates with
systematically different intensity across regions.

\begin{table}[H]
    \caption{Exchangeability diagnostics\protect\label{tab:exch-diagnostic}}
    
    \begin{centering}
    \begin{tabular}{lcccc}
    \toprule
     & SOB & SOB + age & SOB + cov & SOB + cov + age\tabularnewline
    \midrule
    Cell noise ($\hat{q}_{j}^{-1/2}$) & $-0.60$ & $-0.49$ & $-0.68$ & $-0.52$\tabularnewline
    Northeast & 0.44 & 0.36 & 0.42 & 0.35\tabularnewline
    Midwest & $-0.76$ & $-0.68$ & $-0.87$ & $-0.74$\tabularnewline
    South & 0.29 & 0.27 & 0.65 & 0.50\tabularnewline
    West & 0.11 & 0.11 & $-0.06$ & $-0.02$\tabularnewline
    \midrule
    $\chi^{2}(4)$ & 1.04 & 0.76 & 1.50 & 0.94\tabularnewline
    $p$-value & 0.90 & 0.94 & 0.83 & 0.92\tabularnewline
    \bottomrule
    \end{tabular}
    \par\end{centering}
    {\small\emph{Notes:}}{\small{} The first five rows report individual feature
    $t$-statistics $(\Xi'\tilde{r})_{a}/\sqrt{\hat{\Omega}_{aa}}$. The cell-noise
    feature is the centered inverse-root cell frequency $\hat{q}_{j}^{-1/2}$, a proxy
    for the noise scale of the state cell moment. The four region indicators, likewise centered, identify the Census region of state $j$. The joint statistic
    $(\Xi'\tilde{r})'\hat{\Omega}^{-1}(\Xi'\tilde{r})\sim\chi^{2}(4)$ uses the
    cell-noise feature and three region contrasts (West omitted, since the four
    region indicators are collinear). Sample and instruments as in
    Table~\ref{tab:ak91-replication}.}{\small\par}
\end{table}
 
Recall from Remark~\ref{rem:test-power} that the
test cannot detect violations involving features perfectly aligned with first-stage strength. Regressing $\hat{q}_{j}^{-1/2}$ on the constant and the first-stage coefficients
$-\hat{G}$ leaves 94--97\% of its variance in the residual depending on the control specification. Corresponding regressions with each region indicator as an outcome leave 75--99\% of the variance in the residual. Thus,
basing the test on these features probes for violations in directions that should yield non-trivial power.

Table~\ref{tab:exch-diagnostic} reports the results. The cell-noise feature
shows a mild negative association that is far from significant in every
specification ($|t|\le0.68$). No region indicator approaches significance
(the largest is $0.87$ in absolute value). The joint $\chi^{2}(4)$ statistic
lies well below its degrees of freedom in all four specifications, and none
rejects at conventional levels ($p\ge0.83$). The data thus provide no evidence against exchangeability.

\subsection{Misspecification-aware estimates}

Table~\ref{tab:ak91-bceb} presents estimates of the hyperparameters
$\left(\mu,\sigma^{2}\right)$ alongside the GMM, BC, and EB estimates. Again, the GMM point estimates and standard errors coincide with the
TSLS estimates reported in Table~\ref{tab:ak91-replication} by construction.
The reported $J$-statistics, however, differ slightly from those
in Table~\ref{tab:ak91-replication}. Both statistics are quadratic forms of the type $\hat{r}'\hat{V}^{-1}\hat{r}$ evaluated at the TSLS estimate, with heteroscedasticity-consistent variance matrix $\hat{V}$. The statistics reported here use the
FWL-residualized excluded-instrument moments $\hat{r}=\sqrt{n}\hat{g}(\hat{\theta})$. The Table~\ref{tab:ak91-replication}
statistics instead use the full unresidualized system, including the
control block.
Under the null of correct specification and conditional homoscedasticity, a
setting in which the TSLS weighting is asymptotically efficient, the two are
asymptotically equivalent.

Consistent with strong identification of $\mu$ (Assumption~\ref{assu:regime}.vi), $\hat{\kappa}:=m^{-1}\|\hat{M}1_{m}\|^{2}$
is far from zero, ranging across the four specifications from $0.68$
to $0.92$. Table~\ref{tab:ak91-bceb} also reports the empirical Lindeberg
conditions that Theorem~\ref{thm:clt} requires for the BC and EB
estimators. The BC values range
from $0.0006$ to $0.0015$, while the EB values range from $0.0010$ to $0.0047$. All of these values are very small, suggesting that the
no-dominant-column condition of Assumption~\ref{assu:lindeberg}.iii holds for each estimator. Large
values would call into question the validity of the
misspecification-aware standard errors.

\begin{table}[H]
\caption{Bias-Corrected and Empirical Bayes Estimates: Returns to Education
\protect\label{tab:ak91-bceb}}

\begin{centering}
\begin{tabular}{lcccc}
\toprule 
 & SOB & SOB + age & SOB + cov & SOB + cov + age\tabularnewline
\midrule
\emph{Hyperparameters} &  &  &  & \tabularnewline
$\hat{\mu}$ & $-5.39$ & 4.49 & $-3.67$ & 4.75\tabularnewline
$\widehat{\sigma^{2}}$ & 50.97 & 0.00 & 0.10 & 0.00\tabularnewline
 &  &  &  & \tabularnewline
\midrule
\emph{Point estimates} &  &  &  & \tabularnewline
GMM & 0.0424 & 0.1413 & 0.0360 & 0.1240\tabularnewline
 & (0.0188) & (0.0246) & (0.0193) & (0.0251)\tabularnewline
BC & 0.0876 & 0.1108 & 0.0678 & 0.0927\tabularnewline
 & (0.0364) & (0.0318) & (0.0341) & (0.0314)\tabularnewline
EB & 0.0846 & 0.1103 & 0.0769 & 0.0935\tabularnewline
 & (0.0310) & (0.0270) & (0.0252) & (0.0271)\tabularnewline
 &  &  &  & \tabularnewline
\midrule
\emph{Diagnostics} &  &  &  & \tabularnewline
$J$ {[}d.f.{]} & 66.21 {[}50{]} & 46.14 {[}50{]} & 57.76 {[}50{]} & 45.14 {[}50{]}\tabularnewline
$\hat{\kappa}$ & 0.675 & 0.884 & 0.708 & 0.924\tabularnewline
$\max_{j}\hat{\Lambda}_{BC,j}'\hat{V}_{BC}^{-1}\hat{\Lambda}_{BC,j}$ & 0.0006 & 0.0012 & 0.0008 & 0.0015\tabularnewline
$\max_{j}\hat{\Lambda}_{EB,j}'\hat{V}_{EB}^{-1}\hat{\Lambda}_{EB,j}$ & 0.0010 & 0.0039 & 0.0046 & 0.0047\tabularnewline
Hausman $p$-value & 0.88 & 0.98 & 0.69 & 0.96\tabularnewline
\bottomrule
\end{tabular}
\par\end{centering}
{\small\emph{Notes:}}{\small{} Sample and instruments as in Table~\ref{tab:ak91-replication}.
All 51 non-Q1~$\times$~SOB moment conditions are treated as potentially
misspecified. GMM uses weighting matrix $\hat{W}=\hat{S}$
on $\hat{S}^{-1}$-rescaled moments, reproducing the TSLS estimator.
GMM standard errors are heteroscedasticity-robust. BC and EB standard errors are the misspecification-aware estimators of Section~\ref{subsec:Inference}, which utilize the same robust $\hat{V}$. Standard errors in parentheses. The Hausman row reports the $p$-value from comparing the test statistic in \eqref{stat:Hausman} to a $\chi^{2}(1)$ distribution.}{\small\par}
\end{table}

The estimated mean specification error $\hat{\mu}$ is negative in
the specifications without age controls, consistent with the experience
channel highlighted by \textcite{rosenzweig2000natural}: non-Q1 individuals
have less potential labor market experience, producing a negative
direct effect on earnings. With age controls,  $\hat{\mu}$ turns
positive, consistent with the family-composition mechanism of \textcite{buckles2013season}:
once potential experience differences are absorbed, the remaining
direct effect of non-Q1 birth reflects better family backgrounds,
which raise earnings. 

On the original log-wage scale, these estimates correspond to a mean
per-state direct effect $\hat{\mu}/\sqrt{n}$ ranging from $-0.0077$ to
$0.0068$ across the four specifications, implying exclusion violations
of modest magnitude (between $0.53\%$ and $0.77\%$ in absolute
terms). Aggregating across the 51 states, the bias correction $-(\hat{\mu}/\sqrt{n})\hat{\Lambda}1_{m}$
ranges from $-0.031$ to $+0.045$ across the four specifications.
In the SOB specification, for instance, $\hat{\mu}=-5.39$ generates
a BC adjustment of approximately $+0.045$ that moves the schooling
return from $0.042$ (TSLS) to $0.088$ (BC). 

Consistent with the $J$-statistics falling below their degrees of
freedom, the estimated variance $\widehat{\sigma^{2}}$ is zero in both
specifications with age controls, indicating no excess dispersion
beyond the mean shift once experience differences are absorbed. In
the specifications without age controls, the $J$-statistics exceed
their expected values but $\widehat{\sigma^{2}}$ is only meaningfully positive in the SOB-only design. Thus, there is
some heterogeneity in the direct effects across states but it also seems to be captured by the controls other than age. 

Without age controls,  $\hat{\mu}<0$ and the BC and EB corrections
adjust upward. In the specifications with age controls,  $\hat{\mu}>0$
and the BC and EB corrections adjust downward. In all specifications,
the corrected estimates converge toward the 0.07--0.11 range, substantially
narrowing the cross-specification dispersion relative to GMM/TSLS.
While the corrections consistently shift estimates in the direction
of OLS, they overshoot in the specifications without age controls,
where the estimated misspecification is largest. 

Comparing the difference between the BC and EB point estimates to the difference in their squared standard errors yields the Hausman-type test of \eqref{stat:Hausman}, reported in the bottom row of Table~\ref{tab:ak91-bceb}. Consistent with the earlier findings of Section \ref{subsec:Scrutinizing}, the test fails to reject in all four specifications. Evidently, shrinkage adjustments of this magnitude are to be expected under exchangeability given the hyperparameter estimates.

The BC standard
errors are larger than those of TSLS because they account
for two additional sources of uncertainty: the estimation of $\hat{\mu}$
and (when $\widehat{\sigma^{2}}>0$) the cross-moment dispersion of specification
errors. The increases are especially stark in the SOB-only design, where both contributions are present. 
As noted in Section \ref{subsec:ordering}, EB can improve on BC even when $\widehat{\sigma^{2}}=0$, particularly when a suboptimal weighting matrix is used or the moments are heteroscedastic. 
The EB standard errors are 14--26\% smaller than those of BC, including in the specifications with $\widehat{\sigma^{2}}=0$, reflecting the non-trivial moment heteroscedasticity driven by cell-size variation. Notably, the EB standard errors are only slightly larger than TSLS in these specifications, indicating that the ultimate precision costs of bias correction are low in this application when combined with shrinkage.
  
On net, the BC and EB corrections are consistent with season of birth
having direct effects on earnings whose sign depends on whether age
controls are included: the direct effects are negative when age controls
are omitted but positive when age controls are included. The corrections
bring TSLS estimates into closer agreement across specifications,
yielding returns to schooling in the range of 7--11\% per year. This
relative stability of the corrected estimates is reassuring, suggesting
that the misspecification-aware methods can detect and mitigate the
influence of omitted-variables bias. Likewise, the elevated standard
errors of the BC and EB estimators provide a more honest
assessment of the composite uncertainty faced by researchers.

\subsection{Visual IV}

As a final exercise, it is useful to compare the results of Table~\ref{tab:ak91-bceb} to those that would have emerged from a simpler visual IV diagnostic of the form discussed in Example~\ref{exa:VIV}. Define the reduced-form and first-stage coefficient vectors
\[
    \hat \delta = \hat S^{-1} \left(\frac{1}{n}\sum_{i=1}^n \tilde Z_i \tilde Y_i\right) \qquad \hat{\pi} = \hat S^{-1} \left(\frac{1}{n}\sum_{i=1}^n \tilde Z_i \tilde T_i\right).
\]
The reduced-form entries $\hat{\delta}_{j}$ give
the covariate-adjusted log-wage gap for non-Q1 births in state $j$, while the
first stages $\hat{\pi}_{j}$ give the adjusted schooling gap in state $j$. Figure~\ref{fig:ak91-viv} plots a scatter of these objects.

\begin{figure}[H]
    \begin{centering}
    \caption{Visual IV: reduced forms versus first stages across states\protect\label{fig:ak91-viv}}
    \par\end{centering}
    \centering{}\includegraphics[width=0.85\textwidth]{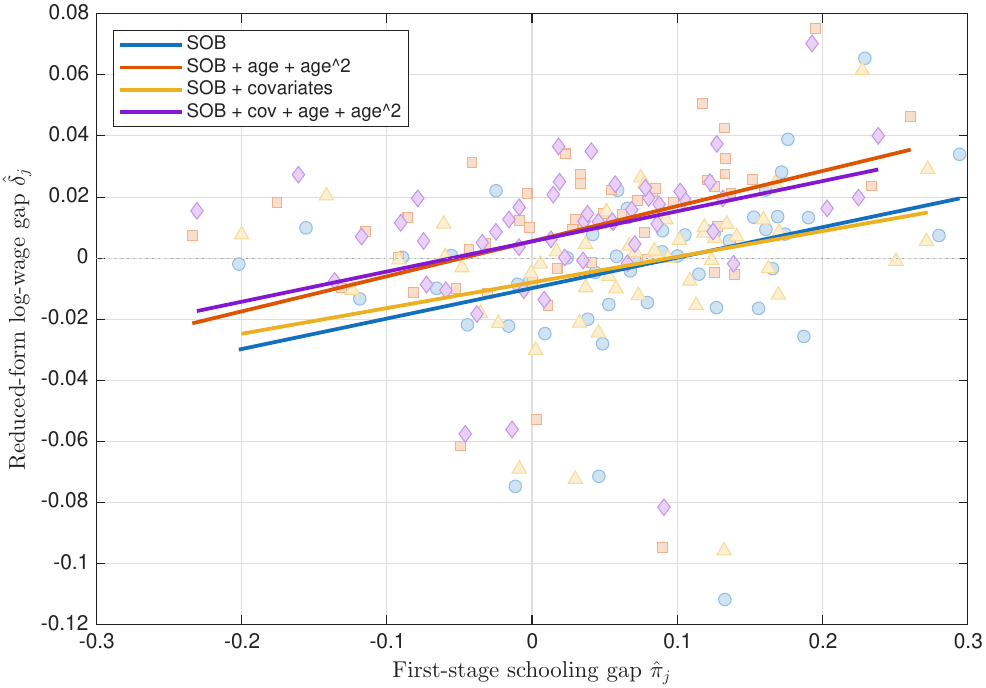}
    \par\medskip{}
    \begin{centering}
    \begin{tabular}{lcccc}
    \toprule
     & SOB & SOB + age & SOB + cov & SOB + cov + age\tabularnewline
    \midrule
    Intercept $\hat{\mu}_{W}/\sqrt{n}$ & $-$0.0098 & 0.0055 & $-$0.0080 & 0.0055\tabularnewline
     & (0.0029) & (0.0023) & (0.0027) & (0.0021)\tabularnewline
    Implied $\hat{\mu}_{W}$ & $-$6.85 & 3.87 & $-$5.55 & 3.82\tabularnewline
     & (2.05) & (1.61) & (1.89) & (1.48)\tabularnewline
    Slope $\hat{\theta}_{\mathrm{VIV}}$ & 0.0998 & 0.1150 & 0.0841 & 0.0988\tabularnewline
     & (0.0256) & (0.0247) & (0.0249) & (0.0242)\tabularnewline
    \bottomrule
    \end{tabular}
    \par\end{centering}
    \raggedright{\small\emph{Notes:}}{\small{} Each marker is a state ($m=51$). The
    horizontal axis is the first-stage coefficient $\hat{\pi}_{j}$ and the vertical axis the
    reduced-form coefficient $\hat{\delta}_{j}$, both residualized against the controls of the
    corresponding specification in Table~\ref{tab:ak91-bceb}.
    Lines are GLS fits of
    $\hat{\delta}_{j}$ on $\hat{\pi}_{j}$ with weight $\hat{S}$, allowing a non-zero intercept. The slope $\hat{\theta}_{\mathrm{VIV}}$
    approximates the bias-corrected return and the intercept $\hat{\mu}_{W}/\sqrt{n}$ the
    mean exclusion violation. GLS standard errors in parentheses.}{\small\par}
\end{figure}

A through-origin GLS fit of reduced forms to first stages with weighting matrix $\hat{S}$ reproduces the TSLS slope exactly. In contrast, Figure~\ref{fig:ak91-viv} overlays
GLS fits that allow for a non-zero intercept. The fitted intercept
$\hat{\mu}_{W}/\sqrt{n}$ estimates the weighted mean direct effect of non-Q1 birth on earnings. As in Table~\ref{tab:ak91-bceb}, the VIV indicates these effects are
negative in the specifications without age controls and positive once they are included. The implied $\hat{\mu}_{W}$ are of comparable magnitude to the $\hat{\mu}$ estimates in Table~\ref{tab:ak91-bceb} and share their signs across all four specifications.

The fitted slopes give the visual-IV returns $\hat{\theta}_{\mathrm{VIV}}$, which are modestly larger than the BC estimates of Table~\ref{tab:ak91-bceb}. Like the BC estimates, the VIV returns are less dispersed across control specifications than the TSLS estimates, ranging from 8--12\% per year. However, the VIV standard errors are smaller than those of the BC estimator in Table~\ref{tab:ak91-bceb} and comparable to those of TSLS, which is to be expected given that the GLS standard errors assume the model is properly specified. 

In this case, the impression one takes away from the VIV diagnostic about the excludability of the instruments and the magnitude of returns to schooling aligns closely with the findings from the more sophisticated BC and EB methods. A tentative conclusion is that plotting VIV fits including an intercept offers a useful diagnostic for gauging the average direction of exclusion violations. However, the standard errors from such methods fail to incorporate uncertainty due to idiosyncratic specification error, potentially leading to substantial overstatement of precision even after accounting for mean biases.

\section{Conclusion \protect\label{sec:Conclusion}}

Econometric models are almost always misspecified. The methods developed here leverage overidentifying restrictions to
correct GMM estimates of parameters of interest and their standard
errors for exchangeable misspecification. As evidenced by the SOB+age specification of Table \ref{tab:ak91-bceb},
these corrections can be non-trivial even when the realized $J$-statistics
are small. While the $J$-statistic distributes power across local
alternatives in $m-p$ directions, the hyperparameters $\left(\mu,\sigma^{2}\right)$
used in the correction pool information across these dimensions into
a low-dimensional summary. 

As the simulations and empirical application illustrate, the proposed
correction methods are not entirely automatic. The researcher must
have some sense of which aspect of the moment conditions is misspecified.
Moreover, the moment conditions must share a common structure that
makes exchangeability plausible. As the examples of Section \ref{sec:Examples}
demonstrate, the choice of moment scaling determines the exchangeability
model and is therefore a substantive modeling decision; alternative
scalings may yield different corrections. This reflects the general
observation that misspecification can only be studied relative to
larger encompassing models that are plausibly operative in the setting
under study \parencite{armstrong2025misspecification}. 

When a plausible scaling has been chosen, the empirical Bayes methods
proposed here offer a useful supplement to GMM estimates. Adoption
of these methods may carry the added benefit of improving research
transparency by reducing selective reporting of $J$-statistics and
other model diagnostics. Rather than fretting about statistical model
rejections, researchers should focus their efforts on repairing their
estimates of flawed, but ultimately useful, econometric models.

\printbibliography

\clearpage

\appendix

\section{Proofs\protect\label{app:proofs}}

\subsection*{Notation and preliminary bounds}

Throughout the appendix, $\|\cdot\|_{2}$ and $\|\cdot\|_{F}$ denote the operator and Frobenius
norms. All stochastic orders are taken under the joint asymptotics of Assumption~\ref{assu:regime},
with $m\to\infty$ and $m^{2}/n\to0$. I repeatedly make use of three elementary inequalities that are valid for
conformable matrices:
\[
\|AB\|_{F}\le\|A\|_{2}\|B\|_{F},\qquad |\mathrm{tr}(A'B)|\le\|A\|_{F}\|B\|_{F},\qquad
\|A\|_{F}\le\sqrt{m}\,\|A\|_{2}.
\]
For symmetric positive semidefinite $B$, I also use $|\mathrm{tr}(AB)|\le\|A\|_{2}\,\mathrm{tr}(B)$
and its specialization $\mathrm{tr}(A^{2})\le\|A\|_{2}\,\mathrm{tr}(A)$.

Assumptions~\ref{assu:regime}.ii and \ref{assu:regime}.v make the population matrices bounded operators:
$\|M\|_{2}$, $\|V\|_{2}$, and $\|W\|_{2}$ are $O(1)$. Their $m\times m$ traces are therefore $O(m)$. By
Assumption~\ref{assu:regime}.vi, the centering vector satisfies $\|M1_{m}\|^{2}=\kappa m(1+o(1))$
with $\kappa\in(0,\infty)$, hence $\|M1_{m}\|=O(\sqrt{m})$. Assumption~\ref{assu:regime}.iii
supplies the first-stage rates $\|\hat{M}-M\|_{2}=O_{p}(\sqrt{m/n})$ and $\|\hat{V}-V\|_{2}=O_{p}(\sqrt{m/n})$. These rates
yield $\|\hat{M}\|_{2}=O_{p}(1)$ and $\|\hat{M}\|_{F}=O_{p}(\sqrt{m})$. They also give
$\|(\hat{M}-M)1_{m}\|\le\sqrt{m}\,\|\hat{M}-M\|_{2}$, which is $O_{p}(m/\sqrt{n})$. Hence,
$\|\hat{M}1_{m}\|=O_{p}(\sqrt{m})$.

Finally, the centered moment $\eta_{n}:=r_{n}-\mu1_{m}$ has variance $\Sigma_{n}=V_{n}+\sigma^{2}Q$ with $\|\Sigma_{n}\|_{2}=O(1)$. Hence $E\|\eta_{n}\|^{2}=\mathrm{tr}(\Sigma_{n})=O(m)$ and $\|\eta_{n}\|=O_{p}(\sqrt{m})$. Write $r_{n}=\mu1_{m}+\eta_{n}$ with $|\mu|=O_{p}(1)$. Then $\|r_{n}\|=O_{p}(\sqrt{m})$. Likewise $\|Mr_{n}\|=O_{p}(\sqrt{m})$, using $\|M1_{m}\|=O(\sqrt{m})$ and $\|M\|_{2}=O(1)$. The centering matrix satisfies $\|M_{\mu}\|_{2}=O(1)$. These bounds are used below without further comment.

\begin{lemapp}[Oracle and first-stage building blocks]
\label{lem:building-blocks}
Under Assumptions~\ref{assu:random_b} and \ref{assu:regime}, the oracle mean estimator
$\mu_{\text{oracle}}:=(M1_{m})'Mr_{n}/\|M1_{m}\|^{2}$ and the masked residual
$u_{n}:=\hat{r}-\hat{M}r_{n}$ satisfy
\begin{enumerate}
\item[(i)] $\mu_{\text{oracle}}-\mu=O_{p}(1/\sqrt{m})$;
\item[(ii)] $\hat{\mu}-\mu_{\text{oracle}}=O_{p}(\sqrt{m/n})$;
\item[(iii)] $u_{n}=\hat{M}\rho_{n}$, where $\rho_{n}$ is the linearization remainder of
Assumption~\ref{assu:regime}.iii; hence $\|u_{n}\|=O_{p}(\sqrt{m/n})$.
\end{enumerate}
\end{lemapp}
\begin{proof} To avoid redundancy, I prove these claims in a different order than they are presented.

\textbf{(i)} Write $Mr_{n}=\mu M1_{m}+M(b-\mu1_{m})+M\varepsilon_{n}$. Then
\[
\mu_{\text{oracle}}-\mu=\frac{(M'M1_{m})'(b-\mu1_{m})}{\|M1_{m}\|^{2}}
+\frac{(M'M1_{m})'\varepsilon_{n}}{\|M1_{m}\|^{2}}.
\]
The first numerator is mean-zero. The second has mean $(M'M1_{m})'E[\varepsilon_{n}]$, bounded by $\|M'M1_{m}\|\,\|E[\varepsilon_{n}]\|=o(\sqrt{m})$ under Assumption~\ref{assu:random_b}.i. Both have variance $O(m)$. Note that
\[
\begin{aligned}
\mathrm{Var}\bigl((M'M1_{m})'(b-\mu1_{m})\bigr)&=\sigma^{2}(M'M1_{m})'Q(M'M1_{m})\le\sigma^{2}\|M'M1_{m}\|^{2}=O(m),\\
\mathrm{Var}\bigl((M'M1_{m})'\varepsilon_{n}\bigr)&=(M'M1_{m})'V_{n}(M'M1_{m})=O(m),
\end{aligned}
\]
where I have used $\|M'M1_{m}\|^{2}=O(m)$ and $\|V_{n}\|_{2}=O(1)$. By Chebyshev, each numerator is $O_{p}(\sqrt{m})$.
Dividing by $\|M1_{m}\|^{2}=\kappa m(1+o(1))$ gives the claim.

\textbf{(iii)} The first-order condition $\hat{G}'\hat{W}\hat{r}=0$ gives $\hat{M}\hat{r}=\hat{r}$.
Combined with $\hat{M}\hat{G}=0$, this yields
$u_{n}=\hat{r}-\hat{M}r_{n}=\hat{M}\bigl(\hat{r}-r_{n}-\hat{G}\sqrt{n}(\hat{\theta}-\theta_{0})\bigr)
=\hat{M}\rho_{n}$. Assumption~\ref{assu:regime}.iii then bounds
$\|u_{n}\|\le\|\hat{M}\|_{2}\|\rho_{n}\|=O_{p}(\sqrt{m/n})$.

\textbf{(ii)} The numerator difference decomposes exactly as
\[
(\hat{M}1_{m})'\hat{r}-(M1_{m})'Mr_{n}
=\bigl((\hat{M}-M)1_{m}\bigr)'Mr_{n}+(\hat{M}1_{m})'(\hat{M}-M)r_{n}+(\hat{M}1_{m})'u_{n}.
\]
By Cauchy--Schwarz and the bounds established above, the first term is
at most $\|(\hat{M}-M)1_{m}\|\,\|Mr_{n}\|$, which is $O_{p}(m^{3/2}/\sqrt{n})$. The second term is of the same order.
The third term is at most $\|\hat{M}1_{m}\|\,\|u_{n}\|$, which is $O_{p}(m/\sqrt{n})$. The denominator error
$\|\hat{M}1_{m}\|^{2}-\|M1_{m}\|^{2}$, which is $O_{p}(m^{3/2}/\sqrt{n})$, contributes the same order through
$(M1_{m})'Mr_{n}=O_{p}(m)$. Dividing by $\|\hat{M}1_{m}\|^{2}=\kappa m(1+o_{p}(1))$ gives the claim.
\end{proof}

The next lemma controls quadratic forms in the centered specification errors. It supplies the
concentration used in the consistency arguments directly from exchangeability and bounded fourth
moments, with no separate regularity condition. The exact permutation moments of such quadratic forms are classical \parencite{kendall1977advanced,mantel1967detection}. The lemma repackages them as an order bound.
\begin{lemapp}[Permutation quadratic forms]
\label{lem:perm-quadratic}
Let $A$ be a symmetric $m\times m$ matrix with $\|A\|_{2}\le C$ and $\mathrm{tr}(A^{2})\ge1$.
Under Assumption~\ref{assu:random_b}, conditional on the order statistics of $b$,
\[
E\left[(b-\mu1_{m})'A(b-\mu1_{m})\right]=\sigma^{2}\,\mathrm{tr}(AQ)
\]
and
\[
\mathrm{Var}\left[(b-\mu1_{m})'A(b-\mu1_{m})\right]\le C'\Big(\sigma^{4}+\tfrac{1}{m}\textstyle\sum_{l=1}^{m}(b_{l}-\mu)^{4}\Big)\mathrm{tr}(A^{2}),
\]
where $C'$ depends only on $C$. Since bounded fourth moments give $m^{-1}\sum_{l}(b_{l}-\mu)^{4}=O_{p}(1)$,
\[
(b-\mu1_{m})'A(b-\mu1_{m})=\sigma^{2}\,\mathrm{tr}(AQ)+O_{p}\!\left(\sqrt{\mathrm{tr}(A^{2})}\right).
\]
\end{lemapp}
\begin{proof}
Because $\mu$ and $\sigma^{2}$ are symmetric functions of $b$, conditioning on the order statistics
makes $b-\mu1_{m}$ a uniformly random permutation of the realized centered values $c=(c_{1},\dots,c_{m})$,
which satisfy $\sum_{l}c_{l}=0$ and $S_{2}:=\sum_{l}c_{l}^{2}=(m-1)\sigma^{2}$; write
$S_{4}:=\sum_{l}c_{l}^{4}$, and set $\tilde{c}_{j}:=c_{\phi(j)}$, where $\phi$ is a uniformly random permutation of $\{1,\dots,m\}$. The moments of a uniform sampling
without replacement are exact symmetric functions of the power sums $\sum_{l}c_{l}^{r}$
\parencite{kendall1977advanced}. Since $\sum_{l}c_{l}=0$,
\[
E[\tilde{c}_{j}^{2}]=\tfrac{S_{2}}{m},\qquad E[\tilde{c}_{j}\tilde{c}_{k}]=-\tfrac{S_{2}}{m(m-1)}\quad(j\ne k).
\]

\noindent\textbf{Mean.} Collecting these,
\[
E\Big[\textstyle\sum_{j,k}A_{jk}\tilde{c}_{j}\tilde{c}_{k}\Big]
=\tfrac{S_{2}}{m}\mathrm{tr}(A)-\tfrac{S_{2}}{m(m-1)}\big(1_{m}'A1_{m}-\mathrm{tr}(A)\big)
=\tfrac{S_{2}}{m-1}\big(\mathrm{tr}(A)-\tfrac{1}{m}1_{m}'A1_{m}\big)=\sigma^{2}\,\mathrm{tr}(AQ),
\]
using $S_{2}/(m-1)=\sigma^{2}$ and $\mathrm{tr}(A)-\tfrac{1}{m}1_{m}'A1_{m}=\mathrm{tr}(AQ)$.

\noindent\textbf{Variance.} Split $\Psi:=\sum_{j,k}A_{jk}\tilde{c}_{j}\tilde{c}_{k}=\Psi_{d}+\Psi_{o}$ with
$\Psi_{d}:=\sum_{j}A_{jj}\tilde{c}_{j}^{2}$ and $\Psi_{o}:=\sum_{j\ne k}A_{jk}\tilde{c}_{j}\tilde{c}_{k}$. Then $\mathrm{Var}(\Psi)\le2\mathrm{Var}(\Psi_{d})+2\mathrm{Var}(\Psi_{o})$. The diagonal part is a
linear statistic in the values $\{c_{l}^{2}\}$, and its permutation variance is therefore exactly
\[
\mathrm{Var}(\Psi_{d})=\tfrac{1}{m-1}\Big(\textstyle\sum_{j}(A_{jj}-\bar{A})^{2}\Big)\Big(\sum_{l}(c_{l}^{2}-\tfrac{S_{2}}{m})^{2}\Big)\le\tfrac{S_{4}}{m-1}\,\mathrm{tr}(A^{2}),
\]
with $\bar{A}:=m^{-1}\sum_{j}A_{jj}$, since $\sum_{j}(A_{jj}-\bar{A})^{2}\le\mathrm{tr}(A^{2})$ and
$\sum_{l}(c_{l}^{2}-S_{2}/m)^{2}\le S_{4}$. The off-diagonal part is a \textcite{mantel1967detection} statistic, a bilinear form $\sum_{j\ne k}A_{jk}c_{\phi(j)}c_{\phi(k)}$ in two symmetric matrices with zero diagonals. For $m\ge4$, its exact permutation variance is a polynomial in the sums $A_{0}:=\sum_{j\ne k}A_{jk}$, $A_{1}:=\sum_{j\ne k}A_{jk}^{2}$, and $A_{2}:=\sum_{j}(\sum_{k\ne j}A_{jk})^{2}$ and in the analogous sums for the matrix $(c_{k}c_{l})_{k\ne l}$, which equal $-S_{2}$, $S_{2}^{2}-S_{4}$, and $S_{4}$ because $\sum_{l}c_{l}=0$. Substituting these values into Mantel's formula gives
\begin{align*}
\mathrm{Var}(\Psi_{o})&=\frac{2A_{1}(S_{2}^{2}-S_{4})}{m(m-1)}+\frac{4(A_{2}-A_{1})(2S_{4}-S_{2}^{2})}{m(m-1)(m-2)}\\
&\quad+\frac{(A_{0}^{2}-4A_{2}+2A_{1})(3S_{2}^{2}-6S_{4})}{m(m-1)(m-2)(m-3)}-\frac{A_{0}^{2}S_{2}^{2}}{m^{2}(m-1)^{2}}.
\end{align*}
To bound the four terms, use $A_{1}\le\mathrm{tr}(A^{2})$, the Cauchy--Schwarz bound $A_{2}\le(m-1)A_{1}$, the operator-norm bounds $A_{2}\le\|A-\mathrm{diag}(A)\|_{2}^{2}\,m\le4C^{2}m$ and $A_{0}^{2}\le mA_{2}\le4C^{2}m^{2}$, the identity $S_{2}=(m-1)\sigma^{2}$, and $\mathrm{tr}(A^{2})\ge1$. The first term is at most $2\sigma^{4}\mathrm{tr}(A^{2})$. The second is at most $4mA_{1}(2S_{4}+S_{2}^{2})/(m(m-1)(m-2))\le(12S_{4}/m+8\sigma^{4})\mathrm{tr}(A^{2})$. The third and fourth are each at most a constant multiple of $C^{2}(\sigma^{4}+S_{4}/m)$, hence of $C^{2}(\sigma^{4}+S_{4}/m)\mathrm{tr}(A^{2})$. Collecting the pieces gives the stated bound with $C'$ depending only on $C$. For $m\le3$, the crude bound $\mathrm{Var}(\Psi)\le E[\Psi^{2}]\le C^{2}S_{2}^{2}\le4C^{2}\sigma^{4}\mathrm{tr}(A^{2})$ suffices.

\noindent\textbf{In probability.} Assumption~\ref{assu:random_b}.ii gives
$E[S_{4}]=\sum_{l}E[(b_{l}-\mu)^{4}]=O(m)$, giving $S_{4}/m=O_{p}(1)$. Since $\sigma^{2}=O_{p}(1)$ by Assumption~\ref{assu:random_b}.ii, the variance bound is $O_{p}(\mathrm{tr}(A^{2}))$, and Chebyshev gives
$\Psi-\sigma^{2}\mathrm{tr}(AQ)=O_{p}(\sqrt{\mathrm{tr}(A^{2})})$.
\end{proof}

\subsection*{Proof of Lemma~\ref{lem:mu_consistency}}
\begin{proof}
Combine parts (i) and (ii) of Lemma~\ref{lem:building-blocks}:
$\hat{\mu}-\mu=(\hat{\mu}-\mu_{\text{oracle}})+(\mu_{\text{oracle}}-\mu)
=O_{p}(\sqrt{m/n})+O_{p}(1/\sqrt{m})$.
Under $m^{2}/n\to0$, $\sqrt{m/n}=o(1/\sqrt{m})$, and the oracle term dominates.
\end{proof}

\subsection*{Proof of Lemma~\ref{lem:sigma-consistency}}
\begin{proof}
Decompose the error into an oracle component and a first-stage component. The oracle component uses
the population $M$ and $V$ with random $b$ and $\varepsilon_{n}$:
\[
\sigma_{\text{oracle}}^{2}=\frac{\|Mr_{n}-\mu_{\text{oracle}}M1_{m}\|^{2}-\mathrm{tr}(MVM')}{\mathrm{tr}(M'MQ)},
\]
with $\mu_{\text{oracle}}$ as in Lemma~\ref{lem:building-blocks}. The centered moment $\eta_{n}=(b-\mu1_{m})+\varepsilon_{n}$ has conditional mean $a_{n}:=E[\varepsilon_{n}]$, with $\|a_{n}\|\to0$ by Assumption~\ref{assu:random_b}.i and $E[b\mid\mu,\sigma^{2}]=\mu1_{m}$ from Section~\ref{sec:local_misspecification}. Its variance is $\Sigma_{n}=V_{n}+\sigma^{2}Q$. In the numerator of $\sigma_{\text{oracle}}^{2}$, the mean $a_{n}$ enters the quadratic form $\eta_{n}'M'M\eta_{n}$ only through $a_{n}'M'Ma_{n}=o(1)$ and the cross term $2a_{n}'M'M(\eta_{n}-a_{n})$, whose variance $4a_{n}'M'M\Sigma_{n}M'Ma_{n}$ is $o(1)$. Both are negligible against the $O_{p}(\sqrt{m})$ stochastic error established below. I therefore treat $\eta_{n}$ as mean-zero. The denominator $\mathrm{tr}(M'MQ)=\mathrm{tr}(M'M)-\tfrac{1}{m}\|M1_{m}\|^{2}$ grows at rate $m$, since $\sum_{i,j}(M_{ij}-M_{ji})^{2}\ge0$ and idempotency of $M$ give $\mathrm{tr}(M'M)\ge\mathrm{tr}(M^{2})=\mathrm{tr}\,M=m-p$, while $\mathrm{tr}(M'M)=O(m)$ and $\tfrac{1}{m}\|M1_{m}\|^{2}=O(1)$.

\noindent\textbf{Oracle error.} Since
$Mr_{n}-\mu_{\text{oracle}}M1_{m}=M\eta_{n}-(\mu_{\text{oracle}}-\mu)M1_{m}$,
\[
\|Mr_{n}-\mu_{\text{oracle}}M1_{m}\|^{2}
=\eta_{n}'M'M\eta_{n}-2(\mu_{\text{oracle}}-\mu)\,1_{m}'M'M\eta_{n}+(\mu_{\text{oracle}}-\mu)^{2}\|M1_{m}\|^{2}.
\]
By Lemma~\ref{lem:building-blocks}(i), $\mu_{\text{oracle}}-\mu=O_{p}(1/\sqrt{m})$. The linear form
$1_{m}'M'M\eta_{n}$ is mean zero with variance $(M'M1_{m})'\Sigma_{n}(M'M1_{m})=O(m)$, hence
$1_{m}'M'M\eta_{n}=O_{p}(\sqrt{m})$. With $\|M1_{m}\|^{2}=O(m)$, the last two terms are each $O_{p}(1)$. The
numerator is therefore governed by
$\eta_{n}'M'M\eta_{n}=(b-\mu1_{m})'M'M(b-\mu1_{m})+2(b-\mu1_{m})'M'M\varepsilon_{n}+\varepsilon_{n}'M'M\varepsilon_{n}$. I
bound the three terms separately.

Because $\|M'M\|_{2}=O(1)$ and $\mathrm{tr}((M'M)^{2})=O(m)$, Lemma~\ref{lem:perm-quadratic}
applies with $A=M'M$ and gives $(b-\mu1_{m})'M'M(b-\mu1_{m})=\sigma^{2}\mathrm{tr}(M'MQ)+O_{p}(\sqrt{m})$.
For $\varepsilon_{n}'M'M\varepsilon_{n}$, Assumption~\ref{assu:regime}.vii and Chebyshev give a
deviation of $O_{p}(\sqrt{\mathrm{tr}((M'MV)^{2})})=O_{p}(\sqrt{m})$ from the mean
$\mathrm{tr}(M'MV_{n})$. That mean differs from $\mathrm{tr}(MVM')$ by
$|\mathrm{tr}(M'M(V_{n}-V))|\le\|V_{n}-V\|_{2}\mathrm{tr}(M'M)=O(m^{3/2}/\sqrt{n})=o(\sqrt{m})$, using the convergence $\|V_{n}-V\|_{2}=O(\sqrt{m/n})$ of Assumption~\ref{assu:regime}.vii. Therefore
$\varepsilon_{n}'M'M\varepsilon_{n}=\mathrm{tr}(MVM')+O_{p}(\sqrt{m})$. For the cross term, condition
on $b$ and use $b\perp\varepsilon_{n}$. The conditional variance is
$4(b-\mu1_{m})'M'MV_{n}M'M(b-\mu1_{m})$, which is at most $4\|M'MV_{n}M'M\|_{2}\|b-\mu1_{m}\|^{2}=O(1)\cdot(m-1)\sigma^{2}=O_{p}(m)$ because $\|b-\mu1_{m}\|^{2}=(m-1)\sigma^{2}$ by the definition of the realized variance. Chebyshev then gives a cross term of $O_{p}(\sqrt{m})$.
Collecting,
\[
\eta_{n}'M'M\eta_{n}=\sigma^{2}\mathrm{tr}(M'MQ)+\mathrm{tr}(MVM')+O_{p}(\sqrt{m}).
\]
The $\mathrm{tr}(MVM')$ term cancels the noise subtraction in the numerator. Dividing by
$\mathrm{tr}(M'MQ)$, which grows at rate $m$, yields $\sigma_{\text{oracle}}^{2}=\sigma^{2}+O_{p}(1/\sqrt{m})$.

\noindent\textbf{First-stage error.} Substituting $(\hat{M},\hat{V},\hat{\mu})$ for
$(M,V,\mu_{\text{oracle}})$ changes the numerator and denominator. For the denominator,
$\mathrm{tr}(\hat{M}'\hat{M}Q)-\mathrm{tr}(M'MQ)=\mathrm{tr}((\hat{M}'\hat{M}-M'M)Q)$. Write
$\hat{M}'\hat{M}-M'M=(\hat{M}-M)'M+M'(\hat{M}-M)+(\hat{M}-M)'(\hat{M}-M)$. Since $\|\hat{M}-M\|_{F}\le\sqrt{m}\,\|\hat{M}-M\|_{2}=O_{p}(m/\sqrt{n})$, each term has Frobenius norm $O_{p}(m/\sqrt{n})$, giving
$|\mathrm{tr}((\hat{M}'\hat{M}-M'M)Q)|\le\|\hat{M}'\hat{M}-M'M\|_{F}\|Q\|_{F}=O_{p}(m^{3/2}/\sqrt{n})$,
using $\|Q\|_{F}=\sqrt{m-1}$. Hence
$\mathrm{tr}(\hat{M}'\hat{M}Q)=\mathrm{tr}(M'MQ)(1+O_{p}(\sqrt{m/n}))$ and
$\mathrm{tr}(\hat{M}'\hat{M}Q)^{-1}=O_{p}(1/m)$.

For the squared-residual term, recall the masked residual $u_{n}=\hat{r}-\hat{M}r_{n}$ of Lemma~\ref{lem:building-blocks}(iii), which has norm $O_{p}(\sqrt{m/n})$. The difference
\[
\hat{r}-\hat{\mu}\hat{M}1_{m}-(Mr_{n}-\mu_{\text{oracle}}M1_{m})
=(\hat{M}-M)(r_{n}-\mu_{\text{oracle}}1_{m})-(\hat{\mu}-\mu_{\text{oracle}})\hat{M}1_{m}+u_{n}
\]
has norm $O_{p}(m/\sqrt{n})$. Each term on the right is at most $O_{p}(\sqrt{m/n})\cdot O_{p}(\sqrt{m})$, by
$\|r_{n}-\mu_{\text{oracle}}1_{m}\|=O_{p}(\sqrt{m})$, $\|\hat{M}1_{m}\|=O_{p}(\sqrt{m})$, and
Lemma~\ref{lem:building-blocks}(ii). The change in a squared norm equals twice the inner product with the error plus the squared norm of the error. Since $\|Mr_{n}-\mu_{\text{oracle}}M1_{m}\|=O_{p}(\sqrt{m})$, the squared norm changes by $O_{p}(m^{3/2}/\sqrt{n})+O_{p}(m^{2}/n)=O_{p}(m^{3/2}/\sqrt{n})$.
The trace correction $\mathrm{tr}(\hat{M}\hat{V}\hat{M}')-\mathrm{tr}(MVM')$ expands into three terms.
Each is bounded by a Frobenius factor $O_{p}(m/\sqrt{n})$ times $O_{p}(\sqrt{m})$, hence is
$O_{p}(m^{3/2}/\sqrt{n})$.

\noindent\textbf{Combining.} Write
\[
\mathrm{num}_{\text{feas}}:=\|\hat{r}-\hat{\mu}\hat{M}1_{m}\|^{2}-\mathrm{tr}(\hat{M}\hat{V}\hat{M}'),\qquad
\mathrm{num}_{\text{oracle}}:=\|Mr_{n}-\mu_{\text{oracle}}M1_{m}\|^{2}-\mathrm{tr}(MVM').
\]
The untruncated feasible ratio is $\widetilde{\sigma^{2}}=\mathrm{num}_{\text{feas}}/\mathrm{tr}(\hat{M}'\hat{M}Q)$.
The difference $\mathrm{num}_{\text{feas}}-\mathrm{num}_{\text{oracle}}$ is the squared-residual
change minus the trace correction, each $O_{p}(m^{3/2}/\sqrt{n})$, hence is itself
$O_{p}(m^{3/2}/\sqrt{n})$. Substituting this and the denominator change into the identity
\[
\widetilde{\sigma^{2}}-\sigma_{\text{oracle}}^{2}
=\frac{\mathrm{num}_{\text{feas}}-\mathrm{num}_{\text{oracle}}}{\mathrm{tr}(\hat{M}'\hat{M}Q)}
-\sigma_{\text{oracle}}^{2}\cdot\frac{\mathrm{tr}(\hat{M}'\hat{M}Q)-\mathrm{tr}(M'MQ)}{\mathrm{tr}(\hat{M}'\hat{M}Q)},
\]
with $\sigma_{\text{oracle}}^{2}=O_{p}(1)$, shows that both terms equal
$O_{p}(m^{3/2}/\sqrt{n})\cdot O_{p}(1/m)=O_{p}(\sqrt{m/n})$. Adding the oracle error gives $\widetilde{\sigma^{2}}-\sigma^{2}=O_{p}(1/\sqrt{m})+O_{p}(\sqrt{m/n})$. Finally, because $\sigma^{2}\ge0$,
$|\max(0,\widetilde{\sigma^{2}})-\sigma^{2}|\le|\widetilde{\sigma^{2}}-\sigma^{2}|$, and truncation
preserves the rate:
\[
\widehat{\sigma^{2}}-\sigma^{2}=O_{p}(1/\sqrt{m})+O_{p}(\sqrt{m/n}).
\]
\end{proof}

\subsection*{Proof of Proposition \ref{prop:bc-rate}}
\begin{proof}
Under Assumptions~\ref{assu:random_b}, \ref{assu:regime},
and~\ref{assu:rate}, Lemma~\ref{lem:mu_consistency} gives
$\hat{\mu}-\mu=O_{p}(1/\sqrt{m})+O_{p}(\sqrt{m/n})=o_{p}(1)$. By
Assumption~\ref{assu:regime}.iii, $\|\hat{\Lambda}-\Lambda\|_{2}=O_{p}(1/\sqrt{n})$.
Assumption~\ref{assu:rate}.ii gives $\|\Lambda\|_{2}=O(1)$.
The triangle inequality then gives $\|\hat{\Lambda}\|_{2}=O_{p}(1)$. Likewise, Assumption~\ref{assu:rate}.i and $\|(\hat{\Lambda}-\Lambda)1_{m}\|\le\sqrt{m}\,\|\hat{\Lambda}-\Lambda\|_{2}=O_{p}(\sqrt{m/n})$ give $\|\hat{\Lambda}1_{m}\|\le\|\Lambda1_{m}\|+O_{p}(\sqrt{m/n})=O_{p}(1)$.

\noindent\textbf{Linearization.}
From Assumption~\ref{assu:regime}.iii, the linearization remainder
$\rho_{n}:=\hat{r}-r_{n}-\hat{G}\sqrt{n}(\hat{\theta}-\theta_{0})$
satisfies $\|\rho_{n}\|=O_{p}(\sqrt{m/n})$. Premultiplying by $\hat{\Lambda}$
and using the first-order condition $\hat{G}'\hat{W}\hat{r}=0$, hence $\hat{\Lambda}\hat{r}=0$,
together with $\hat{\Lambda}\hat{G}=-I_{p}$ gives 
\[
\sqrt{n}(\hat{\theta}-\theta_{0})=\hat{\Lambda}r_{n}+\hat{\Lambda}\rho_{n}.
\]
Here $r_{n}=\mu1_{m}+\eta_{n}$, with $\eta_{n}$ the centered moment.

\noindent\textbf{Combining.}
Substituting this expansion into the definition of the bias-corrected estimator $\hat{\theta}_{BC}=\hat{\theta}-(\hat{\mu}/\sqrt{n})\hat{\Lambda}1_{m}$ gives
\[
\sqrt{n}(\hat{\theta}_{BC}-\theta_{0})=\hat{\Lambda}\eta_{n}-(\hat{\mu}-\mu)\hat{\Lambda}1_{m}+\hat{\Lambda}\rho_{n}.
\]
Decompose the first term as $\hat{\Lambda}\eta_{n}=\Lambda\eta_{n}+(\hat{\Lambda}-\Lambda)\eta_{n}$.
The leading piece $\Lambda\eta_{n}$ is $O_{p}(1)$. Its mean $\Lambda E[\varepsilon_{n}]$ has norm $o(1)$ by Assumption~\ref{assu:random_b}.i, and its covariance is $\Lambda\Sigma_{n}\Lambda'$ with $\Sigma_{n}=V_{n}+\sigma^{2}Q$ by $b\perp\varepsilon_{n}$ from the same assumption. Because $\Lambda$ is nonrandom, $E\|\Lambda\eta_{n}\|^{2}=\|\Lambda E[\varepsilon_{n}]\|^{2}+\mathrm{tr}(\Lambda\Sigma_{n}\Lambda')\le o(1)+p\|\Lambda\|_{2}^{2}\|\Sigma_{n}\|_{2}=O(1)$, and Markov's inequality gives the $O_{p}(1)$ bound. The second piece
satisfies $\|(\hat{\Lambda}-\Lambda)\eta_{n}\|\le\|\hat{\Lambda}-\Lambda\|_{2}\|\eta_{n}\|=O_{p}(1/\sqrt{n})\cdot O_{p}(\sqrt{m})=o_{p}(1)$
under $m/n\to0$. The remaining terms are also negligible:
$|\hat{\mu}-\mu|\|\hat{\Lambda}1_{m}\|=o_{p}(1)\cdot O_{p}(1)=o_{p}(1)$, and
$\|\hat{\Lambda}\rho_{n}\|\le\|\hat{\Lambda}\|_{2}\|\rho_{n}\|=O_{p}(\sqrt{m/n})=o_{p}(1)$.
Hence $\hat{\theta}_{BC}-\theta_{0}=O_{p}(1/\sqrt{n})$.
\end{proof}

Several proofs below rely on the following rank identity for the centered moment covariance.
\begin{lemapp}[Rank of centered covariance]\label{lem:rank-centered}
Under Assumptions~\ref{assu:regime}.iv and \ref{assu:regime}.vi,
\[
MM_{\mu}=(I_{m}-\|M1_{m}\|^{-2}M1_{m}1_{m}'M')M,
\]
which deletes the single direction $M1_{m}$ from $\mathrm{col}(M)$. Hence $\mathrm{rank}(MM_{\mu})=m-p-1$, and $M\Sigma_{\mu}(s)M'=(MM_{\mu})(V+sQ)(MM_{\mu})'$ has rank $m-p-1$ and range $\mathrm{col}(MM_{\mu})$ for every $s\ge0$.
\end{lemapp}
\begin{proof}
Assumption~\ref{assu:regime}.vi gives $M1_{m}\neq0$, so $M_{\mu}=I_{m}-1_{m}w'$ with $w=M'M1_{m}/\|M1_{m}\|^{2}$ is well defined. Then $MM_{\mu}=M-\|M1_{m}\|^{-2}M1_{m}1_{m}'M'M=(I_{m}-P)M$, with $P=\|M1_{m}\|^{-2}(M1_{m})(M1_{m})'$ the orthogonal projector onto $\mathrm{span}(M1_{m})$. By Assumption~\ref{assu:regime}.iv, $\mathrm{rank}(M)=\mathrm{rank}(MVM')=m-p$; equivalently, $MVM'$ is positive definite on $\mathrm{col}(M)$, of dimension $m-p$. Deleting the single nonzero direction $M1_{m}\in\mathrm{col}(M)$ lowers the rank by one, hence $\mathrm{rank}(MM_{\mu})=m-p-1$. Because $sQ$ is positive semidefinite, $(MM_{\mu})(V+sQ)(MM_{\mu})'$ stays positive definite on $\mathrm{col}(MM_{\mu})$ for every $s\ge0$. Its range lies in $\mathrm{col}(MM_{\mu})$ and has the same dimension, hence equals it. Therefore $M\Sigma_{\mu}(s)M'$ has rank $m-p-1$ and range $\mathrm{col}(MM_{\mu})$ for every $s\ge0$, including $s=0$.
\end{proof}

\subsection*{Proof of Proposition~\ref{prop:eb-rate}}
\begin{proof}
Let $L:=\Pi M_{\mu}r_{n}$ be the infeasible predictor of the centered moment $M_{\mu}r_{n}$ and $\hat{L}:=\hat{\Pi}(\hat{r}-\hat{\mu}1_{m})$ its feasible counterpart. The vector $L$ is
the finite-sample analog of the predictor $\Pi\ddot{r}$ of Section~\ref{subsec:eb-shrinkage}. Since $M_{\mu}1_{m}=0$, $M_{\mu}r_{n}=M_{\mu}\eta_{n}$. Since $\hat{M}_{\mu}'\hat{M}'\hat{M}1_{m}=0$ by the definition of $\hat{w}$, the vector $\hat{M}1_{m}$ lies in the null space of $\hat{M}\hat{\Sigma}_{\mu}\hat{M}'$ and of its pseudoinverse. Hence $\hat{\Pi}1_{m}=0$ and $\hat{L}=\hat{\Pi}\hat{r}$. By Lemma~\ref{lem:building-blocks}(iii), $\hat{r}=\hat{M}r_{n}+u_{n}$ with $\|u_{n}\|=O_{p}(\sqrt{m/n})$. Since $\hat{\Pi}\hat{M}=\hat{\Pi}$, $\hat{\Pi}1_{m}=0$, and $r_{n}-M_{\mu}\eta_{n}=(\mu+w'\eta_{n})1_{m}$, it follows that $\hat{\Pi}\hat{r}=\hat{\Pi}M_{\mu}\eta_{n}+\hat{\Pi}u_{n}$. Therefore
\[
\hat{L}-L=(\hat{\Pi}(\widehat{\sigma^{2}})-\Pi(\widehat{\sigma^{2}}))M_{\mu}\eta_{n}+(\Pi(\widehat{\sigma^{2}})-\Pi(\sigma^{2}))M_{\mu}\eta_{n}+\hat{\Pi}u_{n}.
\]
The first term collects the first-stage errors in $\hat{M}$ and $\hat{V}$, evaluated at the estimated variance. The second isolates the error in $\widehat{\sigma^{2}}$ and involves population matrices only. The third is the projection of the linearization remainder $u_{n}$. I bound the three terms in turn.

\noindent\textbf{First term.}
By Assumption~\ref{assu:eb-reg}.ii and Markov's inequality, $\sup_{s\ge0}\|(\hat{\Pi}(s)-\Pi(s))M_{\mu}\eta_{n}\|=O_{p}(m/\sqrt{n})$. Evaluating at $s=\widehat{\sigma^{2}}$ shows that the first term is $O_{p}(m/\sqrt{n})$.

\noindent\textbf{Second term.}
Lemma~\ref{lem:sigma-consistency} and Assumption~\ref{assu:regime}.i give $|\widehat{\sigma^{2}}-\sigma^{2}|=O_{p}(1/\sqrt{m})$. Write $A(s):=M\Sigma_{\mu}(s)M'$. By Lemma~\ref{lem:rank-centered}, $A(s)$ has rank $m-p-1$ and range $\mathrm{range}(MM_{\mu})$ for every $s\ge0$. Since $\partial\Sigma_{\mu}(s)/\partial s=M_{\mu}QM_{\mu}'=M_{\mu}M_{\mu}'$ by $M_{\mu}1_{m}=0$, the derivative $\partial A(s)/\partial s=MM_{\mu}M_{\mu}'M'$ has range inside $\mathrm{range}(A(s))$. The difference $A(s)-A(0)=sMM_{\mu}M_{\mu}'M'$ is positive semidefinite on the common range. The nonzero eigenvalues of $A(s)$ therefore increase with $s$, and those of $A(s)^{+}$ decrease. In particular, $\|A(s)^{+}\|_{2}\le\|A(0)^{+}\|_{2}\le\mathrm{tr}((A(0)^{+})^{2})^{1/2}=O(\sqrt{m})$ and $\mathrm{tr}((A(s)^{+})^{2})\le\mathrm{tr}((A(0)^{+})^{2})=O(m)$ for every $s\ge0$, by Assumption~\ref{assu:eb-reg}.i.

Differentiating $A(s)^{+}$ uses the differential of the Moore--Penrose inverse \parencite[Theorem~8.5]{magnus2019matrix}, whose constant-rank hypothesis holds by Lemma~\ref{lem:rank-centered}. That differential has two further terms carrying the projectors $I_{m}-A(s)A(s)^{+}$ and $I_{m}-A(s)^{+}A(s)$, and both annihilate $\partial A(s)/\partial s$. Hence $\partial A(s)^{+}/\partial s=-A(s)^{+}(\partial A(s)/\partial s)A(s)^{+}$. Differentiating $\Pi(s)=\Sigma_{\mu}(s)M'A(s)^{+}M$ by the product rule, and regrouping the factor $\Sigma_{\mu}(s)M'A(s)^{+}M$ in the second piece as $\Pi(s)$, then gives
\[
\partial\Pi(s)/\partial s=(I_{m}-\Pi(s))M_{\mu}M_{\mu}'M'A(s)^{+}M.
\]
The derivative at $s=0$ is one-sided. This suffices because both $\sigma^{2}$ and $\widehat{\sigma^{2}}$ are nonnegative. Integrating along the segment between $\sigma^{2}$ and $\widehat{\sigma^{2}}$ gives
\[
\|(\Pi(\widehat{\sigma^{2}})-\Pi(\sigma^{2}))M_{\mu}\eta_{n}\|\le|\widehat{\sigma^{2}}-\sigma^{2}|\,\sup_{s}\|(\partial\Pi(s)/\partial s)M_{\mu}\eta_{n}\|,
\]
where the supremum runs over that segment. Write $v:=MM_{\mu}\eta_{n}$. By Assumption~\ref{assu:eb-reg}.iii and the bounds $\|M_{\mu}\|_{2},\|M\|_{2}=O(1)$, the vector $(\partial\Pi(s)/\partial s)M_{\mu}\eta_{n}=(I_{m}-\Pi(s))M_{\mu}M_{\mu}'M'A(s)^{+}v$ has norm at most a constant times $\|A(s)^{+}v\|$, uniformly in $s$. It remains to bound $\|A(s)^{+}v\|$ on the segment.

The map $s\mapsto A(s)^{+}v$ has derivative $-A(s)^{+}(\partial A(s)/\partial s)A(s)^{+}v$, whose norm is at most
\[
\|A(s)^{+}\|_{2}\,\|MM_{\mu}M_{\mu}'M'\|_{2}\,\|A(s)^{+}v\|=O(\sqrt{m})\cdot\|A(s)^{+}v\|.
\]
The vector $A(s)^{+}v$ is nonzero whenever $v$ is, because $v$ lies in $\mathrm{range}(A(s))$. The logarithmic derivative of $\|A(s)^{+}v\|$ is therefore bounded by $C\sqrt{m}$ in absolute value, for a constant $C$. Integrating from $\sigma^{2}$ to $s$ gives $\|A(s)^{+}v\|\le\|A(\sigma^{2})^{+}v\|\exp(C\sqrt{m}\,|s-\sigma^{2}|)$, and the exponent is $O(\sqrt{m})\cdot O_{p}(1/\sqrt{m})=O_{p}(1)$ on the segment. It remains to show $\|A(\sigma^{2})^{+}v\|=O_{p}(\sqrt{m})$. Write $A:=A(\sigma^{2})$. The vector $A^{+}v$ has covariance $A^{+}(A+\Delta_{n})A^{+}$, where $A+\Delta_{n}$ is the exact covariance $MM_{\mu}\Sigma_{n}M_{\mu}'M'$ of $v$ and $\Delta_{n}:=MM_{\mu}(V_{n}-V)M_{\mu}'M'$. Its mean squared norm is therefore
\[
E\|A^{+}v\|^{2}=\mathrm{tr}(A^{+}AA^{+})+\mathrm{tr}(A^{+}\Delta_{n}A^{+})+\|A^{+}E[v]\|^{2}.
\]
The first term equals $\mathrm{tr}(A^{+})$ because $A^{+}AA^{+}=A^{+}$. Since $E[v]=MM_{\mu}E[\varepsilon_{n}]$, the third term satisfies
\[
\|A^{+}MM_{\mu}E[\varepsilon_{n}]\|^{2}\le\|A^{+}\|_{2}^{2}\|M\|_{2}^{2}\|M_{\mu}\|_{2}^{2}\|E[\varepsilon_{n}]\|^{2}=O(m)\cdot o(1)
\]
by Assumption~\ref{assu:random_b}.i. Cauchy--Schwarz gives $(\mathrm{tr}\,A^{+})^{2}\le\mathrm{rank}(A)\,\mathrm{tr}((A^{+})^{2})\le m\,\mathrm{tr}((A^{+})^{2})=O(m^{2})$, hence $\mathrm{tr}(A^{+})=O(m)$. For the second term, the bounds $\|M\|_{2},\|M_{\mu}\|_{2}=O(1)$ and Assumption~\ref{assu:regime}.vii give $\|\Delta_{n}\|_{2}=O(\sqrt{m/n})$, hence $|\mathrm{tr}(A^{+}\Delta_{n}A^{+})|\le\|\Delta_{n}\|_{2}\,\mathrm{tr}((A^{+})^{2})=O(\sqrt{m/n})\cdot O(m)=o(m)$. Markov's inequality then gives $\|A^{+}v\|=O_{p}(\sqrt{m})$. Collecting the bounds, the second term is $O_{p}(1/\sqrt{m})\cdot O_{p}(\sqrt{m})=O_{p}(1)$.

\noindent\textbf{Third term.}
By Lemma~\ref{lem:building-blocks}(iii), $u_{n}=\hat{M}\rho_{n}$. Since $\hat{M}$ is idempotent, $\hat{M}u_{n}=u_{n}$, and the third term is $\hat{\Pi}u_{n}=\hat{\Sigma}_{\mu}\hat{M}'\hat{A}(\widehat{\sigma^{2}})^{+}u_{n}$, where $\hat{A}(s):=\hat{M}\hat{\Sigma}_{\mu}(s)\hat{M}'$. I next show that $\sup_{s\ge0}\|\hat{A}(s)^{+}\|_{2}=O_{p}(\sqrt{m})$.

The argument uses the first-stage rates $\|\hat{M}-M\|_{2}=O_{p}(\sqrt{m/n})$ and $\|\hat{V}-V\|_{2}=O_{p}(\sqrt{m/n})$ of Assumption~\ref{assu:regime}.iii, together with the rate they induce for $\hat{M}_{\mu}=I_{m}-1_{m}\hat{w}'$. Here $\hat{M}_{\mu}-M_{\mu}=-1_{m}(\hat{w}-w)'$, hence $\|\hat{M}_{\mu}-M_{\mu}\|_{2}=\sqrt{m}\,\|\hat{w}-w\|$. In $\hat{w}=\hat{M}'\hat{M}1_{m}/\|\hat{M}1_{m}\|^{2}$, the numerator error is $\|(\hat{M}'\hat{M}-M'M)1_{m}\|\le\sqrt{m}\,\|\hat{M}'\hat{M}-M'M\|_{2}=O_{p}(m/\sqrt{n})$. The proof of Lemma~\ref{lem:building-blocks}(ii) gives $\|\hat{M}1_{m}\|^{2}=\kappa m(1+o_{p}(1))$ and $\|\hat{M}1_{m}\|^{2}-\|M1_{m}\|^{2}=O_{p}(m^{3/2}/\sqrt{n})$. The denominator error multiplies $\|M'M1_{m}\|/(\|M1_{m}\|^{2}\|\hat{M}1_{m}\|^{2})=O_{p}(m^{-3/2})$. Both pieces of $\hat{w}-w$ are therefore $O_{p}(1/\sqrt{n})$, hence $\|\hat{M}_{\mu}-M_{\mu}\|_{2}=O_{p}(\sqrt{m/n})$. Combining the three rates gives $\|\hat{A}(0)-A(0)\|_{2}=O_{p}(\sqrt{m/n})$, which is $o_{p}(m^{-1/2})$ by Assumption~\ref{assu:regime}.i.

At $s=0$, Assumption~\ref{assu:eb-reg}.i gives $\lambda_{\min}^{+}(A(0))\ge\mathrm{tr}((A(0)^{+})^{2})^{-1/2}\ge c\,m^{-1/2}$ for some $c>0$. Since each eigenvalue moves by at most the operator norm of the perturbation, the $m-p-1$ largest eigenvalues of $\hat{A}(0)$ exceed $c\,m^{-1/2}/2$ with probability approaching one. The remaining eigenvalues are zero, because $\mathrm{rank}(\hat{A}(0))\le\mathrm{rank}(\hat{M}\hat{M}_{\mu})=m-p-1$ by the argument of Lemma~\ref{lem:rank-centered} applied to $\hat{M}$. Hence $\hat{A}(0)$ has range $\mathrm{range}(\hat{M}\hat{M}_{\mu})$ and $\|\hat{A}(0)^{+}\|_{2}=O_{p}(\sqrt{m})$. For $s>0$, $\hat{A}(s)$ has the same range, and $\hat{A}(s)-\hat{A}(0)=s\hat{M}\hat{M}_{\mu}\hat{M}_{\mu}'\hat{M}'$ is positive semidefinite, hence $\lambda_{\min}^{+}(\hat{A}(s))\ge\lambda_{\min}^{+}(\hat{A}(0))$. This proves the bound. With $\|\hat{\Sigma}_{\mu}\hat{M}'\|_{2}=O_{p}(1)$, the third term is $O_{p}(\sqrt{m})\cdot O_{p}(\sqrt{m/n})=O_{p}(m/\sqrt{n})$.

Combining the three terms gives $\|\hat{L}-L\|=O_{p}(1)+O_{p}(m/\sqrt{n})$.
Division by $\sqrt{m}$ gives part (i).

\noindent\textbf{Estimator rate.}
Recall $\hat{\theta}_{EB}=\hat{\theta}_{BC}-\hat{\Lambda}\hat{L}/\sqrt{n}$. The proof of Proposition~\ref{prop:bc-rate}
gives $\sqrt{n}(\hat{\theta}_{BC}-\theta_{0})=\Lambda\eta_{n}+o_{p}(1)$, and $\Lambda\eta_{n}-\Lambda M_{\mu}\eta_{n}=(w'\eta_{n})\Lambda1_{m}=(\mu_{\text{oracle}}-\mu)\Lambda1_{m}=O_{p}(1/\sqrt{m})$ by Lemma~\ref{lem:building-blocks}(i) and Assumption~\ref{assu:rate}.i. Hence $\sqrt{n}(\hat{\theta}_{BC}-\theta_{0})=\Lambda M_{\mu}\eta_{n}+o_{p}(1)$, and
\[
\sqrt{n}(\hat{\theta}_{EB}-\theta_{0})=\Lambda(M_{\mu}\eta_{n}-L)-\hat{\Lambda}(\hat{L}-L)-(\hat{\Lambda}-\Lambda)L+o_{p}(1).
\]
Each term is $O_{p}(1)$. The first equals $\Lambda(I-\Pi)M_{\mu}\eta_{n}=\Lambda_{EB}\eta_{n}$,
by $L=\Pi M_{\mu}\eta_{n}$. This is the leading term of Corollary~\ref{corr:bc-eb-clt}(ii). Its covariance $\Lambda_{EB}\Sigma_{n}\Lambda_{EB}'$ is $O(1)$ because $\|\Lambda_{EB}\|_{2}\le\|\Lambda\|_{2}\|(I-\Pi)M_{\mu}\|_{2}=O(1)$ by Assumptions~\ref{assu:rate}.ii and \ref{assu:eb-reg}.iii, and its mean is $o(1)$. Markov's inequality gives $O_{p}(1)$.
The second satisfies $\|\hat{\Lambda}(\hat{L}-L)\|\le\|\hat{\Lambda}\|_{2}\|\hat{L}-L\|=O_{p}(1)$,
using $\|\hat{L}-L\|=O_{p}(1)+O_{p}(m/\sqrt{n})=O_{p}(1)$ by Assumption~\ref{assu:regime}.i and $\|\hat{\Lambda}\|_{2}=O_{p}(1)$ from the proof of Proposition~\ref{prop:bc-rate}.
The third is $o_{p}(1)$, since $\|\hat{\Lambda}-\Lambda\|_{2}=O_{p}(1/\sqrt{n})$ by Assumption~\ref{assu:regime}.iii
and $\|L\|\le\|\Pi M_{\mu}\|_{2}\|\eta_{n}\|=O_{p}(\sqrt{m})$, where $\|\Pi M_{\mu}\|_{2}\le\|M_{\mu}\|_{2}+\|(I-\Pi)M_{\mu}\|_{2}=O(1)$. Hence $\sqrt{n}(\hat{\theta}_{EB}-\theta_{0})=O_{p}(1)$,
i.e.\ $\hat{\theta}_{EB}-\theta_{0}=O_{p}(1/\sqrt{n})$. This establishes part (ii).
\end{proof}

\subsection*{Proof of Proposition~\ref{prop:variance-ordering}}
\begin{proof}
By Lemma~\ref{lem:rank-centered}, the centered matrix $M\Sigma_{\mu}M'=(MM_{\mu})\Sigma(MM_{\mu})'$ has rank $m-p-1$. Since this matrix is symmetric and positive semidefinite, it is positive definite on its range, which equals $\mathrm{range}(MM_{\mu})$. The columns of $M\Sigma_{\mu}=(MM_{\mu})\Sigma M_{\mu}'$ lie in this range.

\noindent\textbf{Part (i): risk gap.} The $\sqrt{n}$-scaled bias-corrected error is $\Lambda M_{\mu}\eta_{n}$.
Its variance is $V_{BC}=\Lambda M_{\mu}\Sigma M_{\mu}'\Lambda'=\Lambda\Sigma_{\mu}\Lambda'$.
The EB error is $\Lambda(I-\Pi)M_{\mu}\eta_{n}$. Here
$\Pi=\Sigma_{\mu}M'(M\Sigma_{\mu}M')^{+}M$ is the best-linear-predictor
coefficient for the centered moment $M_{\mu}\eta_{n}$ given its masked image
$MM_{\mu}\eta_{n}$, whose covariance is $M\Sigma_{\mu}M'$. The matrix $\Pi$ is idempotent: $\Pi^{2}=\Pi$ follows from the Moore--Penrose property $(M\Sigma_{\mu}M')^{+}(M\Sigma_{\mu}M')(M\Sigma_{\mu}M')^{+}=(M\Sigma_{\mu}M')^{+}$.
It also satisfies $\Pi\Sigma_{\mu}=\Sigma_{\mu}\Pi'=\Sigma_{\mu}M'(M\Sigma_{\mu}M')^{+}M\Sigma_{\mu}$, which follows from the definition of $\Pi$ and the symmetry of $\Sigma_{\mu}$ and of $(M\Sigma_{\mu}M')^{+}$. Thus, $\Pi$ is self-adjoint in the $\Sigma_{\mu}$ inner product.
Hence, the residual covariance is
\[
(I-\Pi)\Sigma_{\mu}(I-\Pi)'=\Sigma_{\mu}-\Sigma_{\mu}M'(M\Sigma_{\mu}M')^{+}M\Sigma_{\mu}.
\]
Therefore
\[
V_{BC}-V_{EB}=\Lambda\Sigma_{\mu}M'(M\Sigma_{\mu}M')^{+}M\Sigma_{\mu}\Lambda'=(M\Sigma_{\mu}\Lambda')'(M\Sigma_{\mu}M')^{+}(M\Sigma_{\mu}\Lambda').
\]
This matrix is positive semidefinite because $(M\Sigma_{\mu}M')^{+}$ is. It is nonzero if and only if
$M\Sigma_{\mu}\Lambda'\neq0$. The forward direction is trivial. For the reverse direction, note that every column of $M\Sigma_{\mu}\Lambda'$ is a combination of columns of $M\Sigma_{\mu}$ and therefore lies in $\mathrm{range}(M\Sigma_{\mu}M')$, where $(M\Sigma_{\mu}M')^{+}$ is positive definite. Hence $a_{j}'(M\Sigma_{\mu}M')^{+}a_{j}>0$ for any nonzero column $a_{j}$. The difference $V_{BC}-V_{EB}$ is positive definite if and only
if the rank of $M\Sigma_{\mu}\Lambda'$ equals $p$. The $p$ columns of $M\Sigma_{\mu}\Lambda'$ lie in
$\mathrm{range}(M\Sigma_{\mu}M')$, a subspace of dimension $m-p-1$. They can be linearly
independent only when $p\le m-p-1$, that is, $m\geq2p+1$.
When $p=1$, the nonzero and positive-definite conditions coincide.

\noindent\textbf{Strict improvement under correct specification.} With $\sigma^{2}=0$ and efficient weighting ($W=V^{-1}$), the matrix in the strict-improvement condition becomes
\[
M\Sigma_{\mu}\Lambda'=M\left[(w'Vw)1_{m}-Vw\right](\Lambda1_{m})',
\]
where I have used $\Sigma_{\mu}=M_{\mu}VM_{\mu}'$ together with $MV\Lambda'=0$ and $w'V\Lambda'=0$, both of which follow from $MG=0$ and $w'G=0$. This matrix vanishes only if $\Lambda1_{m}=0$ or $(w'Vw)1_{m}-Vw\in\mathrm{col}(G)$, the latter holding under homoscedastic noise. Hence, EB strictly improves on BC under correct specification whenever $\Lambda1_{m}\neq0$ and $(w'Vw)1_{m}-Vw\notin\mathrm{col}(G)$. Homoscedastic noise is the leading case in which the latter fails. Heteroscedasticity makes strict improvement typical but does not guarantee it.

\noindent\textbf{Part (ii): efficiency characterization.} Suppose $\lambda_{\min}(V)>0$. Then $\Sigma=V+\sigma^{2}Q$ is positive definite. Consider the class of $p\times m$ matrices $B$ satisfying $BG=-I_{p}$ and $B1_{m}=0$. The EB sensitivity matrix $\Lambda_{EB}=\Lambda(I-\Pi)M_{\mu}$ belongs to this class. Since $MG=0$, both $\Pi G=0$ and $w'G=0$ hold. Therefore $M_{\mu}G=G$ and $\Lambda_{EB}G=\Lambda G=-I_{p}$. The identity $M_{\mu}1_{m}=0$ gives $\Lambda_{EB}1_{m}=0$. Every member of the class satisfies $B=BM_{\mu}$ because $B1_{m}w'=0$. Hence, every member's variance takes the centered form $B\Sigma B'=BM_{\mu}\Sigma M_{\mu}'B'=B\Sigma_{\mu}B'$. The Jacobian $G=M_{\mu}G$ lies in $\mathrm{col}(M_{\mu})$, which equals $\mathrm{col}(\Sigma_{\mu})$ because $\Sigma$ is positive definite. Therefore $\Sigma_{\mu}\Sigma_{\mu}^{+}G=G$, and $G'\Sigma_{\mu}^{+}G$ is positive definite. Because $\Sigma_{\mu}^{+}$ is positive semidefinite, the matrix
\[
\left(B\Sigma_{\mu}+\left(G'\Sigma_{\mu}^{+}G\right)^{-1}G'\right)\Sigma_{\mu}^{+}\left(B\Sigma_{\mu}+\left(G'\Sigma_{\mu}^{+}G\right)^{-1}G'\right)'
\]
is positive semidefinite for every member $B$. Expanding the product using $\Sigma_{\mu}\Sigma_{\mu}^{+}\Sigma_{\mu}=\Sigma_{\mu}$, $\Sigma_{\mu}\Sigma_{\mu}^{+}G=G$, and $BG=-I_{p}$ reduces it to $B\Sigma_{\mu}B'-\left(G'\Sigma_{\mu}^{+}G\right)^{-1}$. No member's variance falls below $\left(G'\Sigma_{\mu}^{+}G\right)^{-1}$.

The minimum is attained at $\Lambda_{EB}$, and only there. For any member $B$, the rows of $B-\Lambda_{EB}$ annihilate $G$ and $1_{m}$. By Lemma~\ref{lem:rank-centered}, $MM_{\mu}$ has rank $m-p-1$ and annihilates both $G$ and $1_{m}$. Because $M1_{m}\neq0$ puts $1_{m}$ outside $\mathrm{col}(G)$, the orthogonal complement of $\mathrm{col}([G\;1_{m}])$ also has dimension $m-p-1$. The rows of $MM_{\mu}$ span this complement. Hence $B=\Lambda_{EB}+\tilde{C}MM_{\mu}$ for some matrix $\tilde{C}$. As noted at the outset, the columns of $M\Sigma_{\mu}$ lie in $\mathrm{range}(M\Sigma_{\mu}M')$. The orthogonal projector $(M\Sigma_{\mu}M')^{+}(M\Sigma_{\mu}M')$ acts as the identity operator on this range: $(M\Sigma_{\mu}M')^{+}(M\Sigma_{\mu}M')M\Sigma_{\mu}=M\Sigma_{\mu}$. Transposing yields $\Pi\Sigma_{\mu}M'=\Sigma_{\mu}M'(M\Sigma_{\mu}M')^{+}(M\Sigma_{\mu}M')=\Sigma_{\mu}M'$. Thus $(I-\Pi)\Sigma_{\mu}M'=0$. The shrinkage residual is uncorrelated with the masked regressor. The cross term $\Lambda_{EB}\Sigma M_{\mu}'M'\tilde{C}'=\Lambda(I-\Pi)\Sigma_{\mu}M'\tilde{C}'$ therefore vanishes, and
\[
B\Sigma B'=V_{EB}+\tilde{C}\left(M\Sigma_{\mu}M'\right)\tilde{C}'.
\]
Thus $\Lambda_{EB}$ attains the class minimum. Any other minimizer has $\tilde{C}(M\Sigma_{\mu}M')\tilde{C}'=0$. Since $M\Sigma_{\mu}M'=(MM_{\mu})\Sigma(MM_{\mu})'$ with $\Sigma$ positive definite, it follows that $\tilde{C}MM_{\mu}=0$. The minimizer is therefore unique. The matrix $-(G'\Sigma_{\mu}^{+}G)^{-1}G'\Sigma_{\mu}^{+}M_{\mu}$ is also a member of the class because $M_{\mu}G=G$ and $M_{\mu}1_{m}=0$. Since $M_{\mu}$ is idempotent, $M_{\mu}\Sigma_{\mu}=\Sigma_{\mu}$, and the Moore--Penrose identity $\Sigma_{\mu}^{+}\Sigma_{\mu}\Sigma_{\mu}^{+}=\Sigma_{\mu}^{+}$ then shows that its variance equals $(G'\Sigma_{\mu}^{+}G)^{-1}$. This variance meets the lower bound, implying the member is a minimizer. Uniqueness delivers the influence identity $\Lambda_{EB}=-(G'\Sigma_{\mu}^{+}G)^{-1}G'\Sigma_{\mu}^{+}M_{\mu}$ and, in turn, $V_{EB}=(G'\Sigma_{\mu}^{+}G)^{-1}$.

Finally, the generalized least squares rule $-[I_{p}\;\,0]\left([G\;1_{m}]'\Sigma^{-1}[G\;1_{m}]\right)^{-1}[G\;1_{m}]'\Sigma^{-1}$ is also a member of the class, since multiplying it by $[G\;1_{m}]$ gives $-[I_{p}\;\,0]$. The matrix $[G\;1_{m}]$ has full column rank because $M1_{m}\neq0$. By the Gauss--Markov theorem, the variance of this rule is the class minimum. That variance is the leading $p\times p$ block of $([G\;1_{m}]'\Sigma^{-1}[G\;1_{m}])^{-1}$. Partitioned inversion yields the profiled expression for this block. The profiled expression involves only $G$, $\Sigma$, and $1_{m}$. It follows that neither closed form depends on the weighting matrix.
\end{proof}

\subsection*{Proof of Theorem~\ref{thm:clt}}
\begin{proof}
I show the studentized leading term $V_{\star}^{-1/2}\Lambda_{\star}\eta_{n}$ converges to $N(0,I_{p})$ by splitting on whether the realized variance exceeds a small threshold. The high-variance event is covered by a permutation central limit theorem. The low-variance event is covered by a second-moment bound.

Write $\tilde{\Lambda}_{\star,j}:=V_{\star}^{-1/2}\Lambda_{\star,j}$ for the studentized columns and $\tilde{\Lambda}_{\star}:=V_{\star}^{-1/2}\Lambda_{\star}$ for the studentized sensitivity, which is well-defined for sufficiently large $m$ by Assumption~\ref{assu:lindeberg}.i. Then $\tilde{\Lambda}_{\star}\Sigma\tilde{\Lambda}_{\star}'=V_{\star}^{-1/2}V_{\star}V_{\star}^{-1/2}=I_{p}$.

I condition on the order statistics of $b$. Because $\mu$ and $\sigma^{2}$ are symmetric functions of $b$, the conditioning holds them fixed. The realized variance depends on $n$ through $b$. I suppress this dependence in the notation. By the exchangeability in Assumption~\ref{assu:random_b}, $b-\mu1_{m}$ is a uniformly random permutation of its realized values, with conditional mean zero and conditional covariance $\sigma^{2}Q$. Since $r_{n}=b+\varepsilon_{n}$ with $b\perp\varepsilon_{n}$, the studentized leading term splits as
\[
\zeta_{n}:=V_{\star}^{-1/2}\Lambda_{\star}\eta_{n}=\zeta_{n}^{b}+\zeta_{n}^{\varepsilon},\qquad \zeta_{n}^{b}:=\tilde{\Lambda}_{\star}(b-\mu1_{m}),\quad \zeta_{n}^{\varepsilon}:=\tilde{\Lambda}_{\star}\varepsilon_{n},
\]
into a specification piece and a noise piece that are independent at each $n$ given the order statistics. Because $\Lambda_{\star}1_{m}=0$, the identity $\Lambda_{\star}Q\Lambda_{\star}'=\Lambda_{\star}\Lambda_{\star}'$ holds. The conditional covariance of the specification piece is therefore $\Sigma_{b}:=\sigma^{2}\tilde{\Lambda}_{\star}\tilde{\Lambda}_{\star}'$. Write $\Sigma_{\varepsilon}:=\tilde{\Lambda}_{\star}V\tilde{\Lambda}_{\star}'$ for the studentized noise covariance formed with the limiting $V$. The two sum to $\Sigma_{b}+\Sigma_{\varepsilon}=\tilde{\Lambda}_{\star}\Sigma\tilde{\Lambda}_{\star}'=I_{p}$ at every $\sigma^{2}\ge0$. Both are positive semidefinite, hence each is bounded above by $I_{p}$.

\noindent\textbf{The noise piece.}
By Assumption~\ref{assu:lindeberg}.ii, $\xi_{n}:=(\Lambda_{\star}V\Lambda_{\star}')^{-1/2}\Lambda_{\star}\varepsilon_{n}\overset{d}{\rightarrow}N(0,I_{p})$. Write $\zeta_{n}^{\varepsilon}=\Gamma_{n}\xi_{n}$ with $\Gamma_{n}:=V_{\star}^{-1/2}(\Lambda_{\star}V\Lambda_{\star}')^{1/2}$, giving $\Gamma_{n}\Gamma_{n}'=\Sigma_{\varepsilon}$. Given the order statistics, $\Gamma_{n}$ is fixed and $\xi_{n}$ retains its unconditional law. Hence the conditional characteristic function of $\zeta_{n}^{\varepsilon}$ at $t$ is $\varphi_{\xi_{n}}(\Gamma_{n}'t)$, where $\varphi_{\xi_{n}}$ is the characteristic function of $\xi_{n}$. The argument $\Gamma_{n}'t$ obeys $\|\Gamma_{n}'t\|^{2}=t'\Sigma_{\varepsilon}t\le\|t\|^{2}$ and stays in a fixed ball. Convergence of $\xi_{n}$ in distribution makes $\varphi_{\xi_{n}}$ converge to the standard normal characteristic function uniformly on that ball, and the normal characteristic function depends only on the norm of its argument. Hence
\[
\varphi_{\xi_{n}}(\Gamma_{n}'t)=\exp\!\big(-\tfrac12\|\Gamma_{n}'t\|^{2}\big)+o(1)=\exp\!\big(-\tfrac12\,t'\Sigma_{\varepsilon}t\big)+o(1).
\]
The display gives the conditional characteristic function of $\zeta_{n}^{\varepsilon}$ at every argument, up to an error that vanishes with $n$. The noise piece enters the rest of the argument only through this display. Its finite-sample covariance $V_{n}$ plays no role.

\noindent\textbf{The event decomposition.}
Fix a threshold $h>0$ and a unit vector $c\in\mathbb{R}^{p}$. Write $\psi_{n}:=E[\exp(ic'\zeta_{n})\mid\text{order statistics}]$ for the conditional characteristic function of the leading term at $c$. Likewise, let $\psi_{n}^{b}$ and $\psi_{n}^{\varepsilon}$ denote the conditional characteristic functions of the specification and noise pieces. The target is $\exp(-\tfrac12\,c'c)$. I split the deviation from the target according to the realized variance:
\[
\psi_{n}-\exp\!\big(-\tfrac12\,c'c\big)=\mathbf{1}\{\sigma^{2}>h\}\big(\psi_{n}-\exp\!\big(-\tfrac12\,c'c\big)\big)+\mathbf{1}\{\sigma^{2}\le h\}\big(\psi_{n}-\exp\!\big(-\tfrac12\,c'c\big)\big).
\]
Throughout, a random quantity vanishes in probability on an event if its product with the indicator of that event vanishes in probability. On the first event the realized variance is bounded away from zero, and a permutation central limit theorem shows that the first product vanishes in probability. On the second event the specification piece is small in mean square, and a direct bound shows that the second product is at most a nonrandom multiple of $\sqrt{h}+h$, up to a term that vanishes with $n$. Sending $h\to0$ after $n\to\infty$ will then give $\psi_{n}\overset{p}{\rightarrow}\exp(-\tfrac12\,c'c)$.

\noindent\textbf{The event $\{\sigma^{2}>h\}$.}
By the Cram\'er--Wold device, consider the scalar permutation statistic $c'\zeta_{n}^{b}=\sum_{j}a_{j}(b_{j}-\mu)$ with weights $a_{j}:=c'\tilde{\Lambda}_{\star,j}$. Because $\Lambda_{\star}1_{m}=0$, both weights and values are centered, and the variance over the uniform permutation is
\[
s_{n}^{2}=\frac{1}{m-1}\Big(\sum_{j}a_{j}^{2}\Big)\Big(\sum_{k}(b_{k}-\mu)^{2}\Big)=\sigma^{2}\,c'\tilde{\Lambda}_{\star}\tilde{\Lambda}_{\star}'c,
\]
using $\sum_{k}(b_{k}-\mu)^{2}=(m-1)\sigma^{2}$. Since $1=c'\tilde{\Lambda}_{\star}\Sigma\tilde{\Lambda}_{\star}'c\le\|\Sigma\|_{2}\,c'\tilde{\Lambda}_{\star}\tilde{\Lambda}_{\star}'c$, the variance satisfies $s_{n}^{2}\ge\sigma^{2}/\|\Sigma\|_{2}$, with $\|\Sigma\|_{2}\le\|V\|_{2}+\sigma^{2}=O_{p}(1)$ by Assumptions~\ref{assu:regime}.ii and~\ref{assu:lindeberg}.iv. On the event, $s_{n}^{2}\ge h/\|\Sigma\|_{2}>0$. The Wald--Wolfowitz--Noether combinatorial central limit theorem \parencite{hoeffding1951combinatorial}, in the form given by \textcite[Theorem~4.1]{hajek1961sampling}, gives $c'\zeta_{n}^{b}/s_{n}\overset{d}{\rightarrow}N(0,1)$, conditional on the order statistics, provided the two ratios
\[
\frac{\max_{j}a_{j}^{2}}{\sum_{j}a_{j}^{2}},\qquad\frac{\max_{k}(b_{k}-\mu)^{2}}{\sum_{k}(b_{k}-\mu)^{2}}
\]
vanish and the standardized products satisfy a Lindeberg condition: for every $\omega>0$,
\[
\frac{1}{m}\sum_{j,k:\,|d_{jk}|>\omega}d_{jk}^{2}\rightarrow0,\qquad d_{jk}:=\frac{a_{j}(b_{k}-\mu)}{\big[\tfrac{1}{m}\big(\sum_{l}a_{l}^{2}\big)\big(\sum_{l}(b_{l}-\mu)^{2}\big)\big]^{1/2}}.
\]
On the event, all three quantities vanish in probability. The first ratio vanishes in probability: $\max_{j}a_{j}^{2}\le\max_{j}\|\tilde{\Lambda}_{\star,j}\|^{2}\to0$ by Assumption~\ref{assu:lindeberg}.iii, while $\sum_{j}a_{j}^{2}\ge1/\|\Sigma\|_{2}$ is bounded away from zero in probability. For the second, the denominator $(m-1)\sigma^{2}$ exceeds $(m-1)h$ on the event. Assumption~\ref{assu:random_b}.ii bounds the fourth moments of the entries of $b$, and $E[\mu^{4}]\le E[b_{1}^{4}]$ by Jensen's inequality and exchangeability, hence $E[(b_{k}-\mu)^{4}]\le8(E[b_{k}^{4}]+E[\mu^{4}])=O(1)$ and $E[\sum_{k}(b_{k}-\mu)^{4}]=O(m)$. Therefore $\max_{k}(b_{k}-\mu)^{2}\le(\sum_{k}(b_{k}-\mu)^{4})^{1/2}=O_{p}(\sqrt{m})$, and the second ratio is $O_{p}(1/\sqrt{m})$ at each fixed $h$. For the Lindeberg condition, the bound $z^{2}\mathbf{1}\{|z|>\omega\}\le z^{4}/\omega^{2}$ gives
\[
\frac{1}{m}\sum_{j,k:\,|d_{jk}|>\omega}d_{jk}^{2}\le\frac{1}{m\omega^{2}}\sum_{j,k}d_{jk}^{4}=\frac{m\big(\sum_{j}a_{j}^{4}\big)\big(\sum_{k}(b_{k}-\mu)^{4}\big)}{\omega^{2}\big(\sum_{j}a_{j}^{2}\big)^{2}\big(\sum_{k}(b_{k}-\mu)^{2}\big)^{2}}\le\frac{\max_{j}a_{j}^{2}}{\sum_{j}a_{j}^{2}}\cdot\frac{m\sum_{k}(b_{k}-\mu)^{4}}{\omega^{2}(m-1)^{2}\sigma^{4}},
\]
using $\sum_{j}a_{j}^{4}\le(\max_{j}a_{j}^{2})\sum_{j}a_{j}^{2}$. On the event, the second factor is $O_{p}(1)$ at each fixed $h$: the fourth-moment bound makes the numerator $O_{p}(m^{2})$, and the denominator exceeds $\omega^{2}(m-1)^{2}h^{2}$. The first factor vanishes in probability as shown above. Hence the Lindeberg sum vanishes in probability on the event.

The H\'ajek result applies to deterministic weights and values. Here the realized values are random. The two ratios and the Lindeberg sum vanish in probability, and only on the event. Recall that a sequence of random variables converges to zero in probability if and only if every subsequence contains a further subsequence converging to zero almost surely. Given any subsequence, extract a further subsequence along which the indicator-weighted ratios and the indicator-weighted Lindeberg sum converge to zero for almost every realization of the order statistics. Fix such a realization. Along the $n$ in the further subsequence at which $\sigma^{2}>h$, the weights and values are deterministic arrays satisfying the conditions of the H\'ajek result. If there are infinitely many such $n$, the result gives $\psi_{n}^{b}-\exp(-\tfrac12\,s_{n}^{2})\to0$ along them, because $\psi_{n}^{b}$ is the characteristic function of $c'\zeta_{n}^{b}/s_{n}$ evaluated at $s_{n}$, characteristic functions converge uniformly on bounded sets, and $s_{n}^{2}=c'\Sigma_{b}c\le1$. If there are finitely many, the indicator is eventually zero. In either case $\mathbf{1}\{\sigma^{2}>h\}\big(\psi_{n}^{b}-\exp(-\tfrac12\,c'\Sigma_{b}c)\big)\to0$ along the further subsequence, almost surely. Since every subsequence contains such a further subsequence, this indicator-weighted difference vanishes in probability. No limits of $\sigma^{2}$ or of the studentized covariances are extracted. Given the order statistics, the two pieces are independent, and the conditional characteristic functions multiply: $\psi_{n}=\psi_{n}^{b}\psi_{n}^{\varepsilon}$. The noise result above gives $\psi_{n}^{\varepsilon}=\exp(-\tfrac12\,c'\Sigma_{\varepsilon}c)+o(1)$. Because $\Sigma_{b}+\Sigma_{\varepsilon}=I_{p}$ at every $n$,
\[
\mathbf{1}\{\sigma^{2}>h\}\big(\psi_{n}-\exp\!\big(-\tfrac12\,c'c\big)\big)\overset{p}{\rightarrow}0.
\]

\noindent\textbf{The event $\{\sigma^{2}\le h\}$.}
Conditional on the order statistics, the specification piece has mean zero and covariance $\Sigma_{b}$, hence
\[
E\big[\|\zeta_{n}^{b}\|^{2}\mid\text{order statistics}\big]=\mathrm{tr}(\Sigma_{b})\le p\,\sigma^{2}\,\|\Lambda_{\star}\|_{2}^{2}/\lambda_{\min}(V_{\star}).
\]
By Assumptions~\ref{assu:lindeberg}.i and~\ref{assu:lindeberg}.iv, the right side is a bounded multiple of $\sigma^{2}$, at most a bounded multiple of $h$ on the event. The bound $|e^{iz}-1|\le|z|$ and the Cauchy--Schwarz inequality give
\[
|\psi_{n}-\psi_{n}^{\varepsilon}|\le E\big[|c'\zeta_{n}^{b}|\mid\text{order statistics}\big]\le\big(\mathrm{tr}(\Sigma_{b})\big)^{1/2}.
\]
The noise result gives $\psi_{n}^{\varepsilon}=\exp(-\tfrac12\,c'\Sigma_{\varepsilon}c)+o(1)$, and $|\exp(-\tfrac12\,c'\Sigma_{\varepsilon}c)-\exp(-\tfrac12\,c'c)|\le\tfrac12\,c'\Sigma_{b}c\le\tfrac12\,\mathrm{tr}(\Sigma_{b})$, using $\Sigma_{\varepsilon}=I_{p}-\Sigma_{b}$. On the event, $\psi_{n}$ therefore differs from $\exp(-\tfrac12\,c'c)$ by at most a bounded multiple of $\sqrt{h}+h$, plus a term that vanishes with $n$.

\noindent\textbf{Combining the events.}
For every fixed $h>0$, the first product in the event decomposition vanishes in probability, and the second is at most a nonrandom multiple of $\sqrt{h}+h$, plus a term that vanishes with $n$. Because $h$ can be chosen arbitrarily small,
\[
\psi_{n}\overset{p}{\rightarrow}\exp\!\big(-\tfrac12\,c'c\big).
\]
The two ratios and the Lindeberg sum are invariant to the scale of $c$, and the low-variance bounds are finite multiples of $\sqrt{h}$ and $h$ at each fixed argument. The conditional characteristic function of $\zeta_{n}$ therefore converges in probability to $\exp(-\tfrac12\|t\|^{2})$ at every $t\in\mathbb{R}^{p}$.

\noindent\textbf{Conclusion.}
The conditional characteristic function is bounded, and convergence in probability upgrades to convergence in mean. Iterated expectations then give $E[\exp(it'\zeta_{n})]\rightarrow\exp(-\tfrac12\|t\|^{2})$ at every $t$, hence $\zeta_{n}\overset{d}{\rightarrow}N(0,I_{p})$. With the hypothesis $V_{\star}^{-1/2}\tilde{R}_{\star}=o_{p}(1)$ and Slutsky's theorem,
\[
\sqrt{n}\,V_{\star}^{-1/2}(\hat{\theta}_{\star}-\theta_{0})\overset{d}{\rightarrow}N(0,I_{p}),
\]
as claimed.

\noindent\textbf{Conditional statements.}
The hyperparameters are functions of the order statistics. Taking conditional expectations given $(\mu,\sigma^{2})$ cannot increase the mean deviation from the limit. Hence $E[\exp(it'\zeta_{n})\mid\mu,\sigma^{2}]$ converges in mean, and in probability, to $\exp(-\tfrac12\|t\|^{2})$. This is convergence in probability of a random conditional characteristic function. It does not by itself deliver convergence at every value of $(\mu,\sigma^{2})$. Remark~\ref{rem:pointwise-conditional} gives a condition that does. More is true when the realized variance vanishes. Under the coarser conditioning the specification piece still has mean zero and covariance $\Sigma_{b}$, by Lemma~\ref{lem:cond-moments}. The noise piece retains its unconditional law, because $b\perp\varepsilon_{n}$. The bound on $\mathrm{tr}(\Sigma_{b})$ is a nonrandom multiple of $\sigma^{2}$ and uses no other feature of the realized values. The low-variance argument therefore applies unchanged along every admissible sequence of hyperparameter values with $\sigma^{2}\to0$.
\end{proof}

\subsection*{Proof of Lemma~\ref{lem:eb-floor}}
\begin{proof}
The requirement $\lambda_{\min}(V)>0$ of Proposition~\ref{prop:variance-ordering}(ii) is covered by the hypothesis $\lambda_{\min}(V)\ge c_{V}$. The proposition gives
\[
V_{EB}=\left(G'\Sigma^{-1}G-\frac{G'\Sigma^{-1}1_{m}1_{m}'\Sigma^{-1}G}{1_{m}'\Sigma^{-1}1_{m}}\right)^{-1}.
\]
The matrices $G'\Sigma^{-1}G$ and $G'V^{-1}G$ are positive definite because $G$ has full column rank, as the definition of $\Lambda$ presumes, and the inverted matrix is positive definite because Proposition~\ref{prop:variance-ordering}(ii) inverts it. I use repeatedly that inversion reverses the positive semidefinite order on positive definite matrices. First, the subtracted matrix is positive semidefinite, hence $G'\Sigma^{-1}G$ minus the inverted matrix is positive semidefinite. Inverting, $V_{EB}-(G'\Sigma^{-1}G)^{-1}$ is positive semidefinite. Second, $\Sigma-V=\sigma^{2}Q$ is positive semidefinite. Inverting, $V^{-1}-\Sigma^{-1}$ is positive semidefinite, which makes $G'V^{-1}G-G'\Sigma^{-1}G$ positive semidefinite. Inverting again, $(G'\Sigma^{-1}G)^{-1}-(G'V^{-1}G)^{-1}$ is positive semidefinite. Combining the two, $\lambda_{\min}(V_{EB})\ge\lambda_{\min}\big((G'V^{-1}G)^{-1}\big)=1/\lambda_{\max}(G'V^{-1}G)$. The hypothesis $G'V^{-1}G=O(1)$ bounds $\lambda_{\max}(G'V^{-1}G)$, hence the right side is bounded away from zero uniformly in $m$. No step depends on $\sigma^{2}$. The weighting matrix enters only through Proposition~\ref{prop:variance-ordering}(ii), whose conclusion is free of it. The bound is therefore uniform in both.
\end{proof}

\subsection*{Proof of Corollary~\ref{corr:bc-eb-clt}}
\begin{proof}
\noindent\textbf{Part (i): bias correction.}
Standard GMM linearization gives $\sqrt{n}(\hat{\theta}-\theta_{0})=\hat{\Lambda}(\mu1_{m}+\eta_{n})+R$,
where $R=\hat{\Lambda}\rho_{n}$ is formed from the moment-equation linearization remainder
$\rho_{n}$ of Assumption~\ref{assu:regime}.iii. Combined with the BC definition
$\hat{\theta}_{BC}:=\hat{\theta}-\hat{\mu}\hat{\Lambda}1_{m}/\sqrt{n}$,
\[
\sqrt{n}(\hat{\theta}_{BC}-\theta_{0})=\hat{\Lambda}\eta_{n}+R-(\hat{\mu}-\mu)\hat{\Lambda}1_{m}.
\]
By definition of $w$, $\mu_{\text{oracle}}-\mu=w'\eta_{n}$, where $\mu_{\text{oracle}}$ is the oracle mean
of Lemma~\ref{lem:building-blocks}, and Lemma~\ref{lem:building-blocks}(ii) gives
$\hat{\mu}-\mu_{\text{oracle}}=O_{p}(\sqrt{m/n})$. Substituting $\hat{\mu}-\mu=w'\eta_{n}+(\hat{\mu}-\mu_{\text{oracle}})$
and using $\Lambda_{BC}=\Lambda-(\Lambda1_{m})w'=\Lambda M_{\mu}$ yields
\[
\sqrt{n}(\hat{\theta}_{BC}-\theta_{0})=\Lambda_{BC}\eta_{n}+\tilde{R}_{BC}.
\]
The remainder $\tilde{R}_{BC}$ collects four pieces: the sensitivity estimation error $(\hat{\Lambda}-\Lambda)\eta_{n}$, the mean-feasibility gap $-(\hat{\mu}-\mu_{\text{oracle}})\Lambda1_{m}$, the cross term $-(\hat{\mu}-\mu)(\hat{\Lambda}-\Lambda)1_{m}$, and the GMM linearization remainder $R$. Each is $O_{p}(\sqrt{m/n})$. The first is bounded by $\|\hat{\Lambda}-\Lambda\|_{2}\|\eta_{n}\|=O_{p}(1/\sqrt{n})\cdot O_{p}(\sqrt{m})$, using Assumption~\ref{assu:regime}.iii. The second is bounded by $|\hat{\mu}-\mu_{\text{oracle}}|\,\|\Lambda1_{m}\|$, with $\|\Lambda1_{m}\|=O(1)$ by Assumption~\ref{assu:rate}.i. The third is bounded by $|\hat{\mu}-\mu|\,\sqrt{m}\,\|\hat{\Lambda}-\Lambda\|_{2}$, with $\hat{\mu}-\mu=o_{p}(1)$ by Lemma~\ref{lem:mu_consistency}. The fourth is bounded by $\|\hat{\Lambda}\|_{2}\|\rho_{n}\|$, with $\|\hat{\Lambda}\|_{2}=O_{p}(1)$ from the proof of Proposition~\ref{prop:bc-rate} and $\|\rho_{n}\|=O_{p}(\sqrt{m/n})$ by Assumption~\ref{assu:regime}.iii. Hence $\|\tilde{R}_{BC}\|=O_{p}(\sqrt{m/n})$. The floor in Assumption~\ref{assu:lindeberg}.iv applied to $\Lambda_{BC}$ gives $\|V_{BC}^{-1/2}\|_{2}=O(1)$, and $m/n\to0$ by Assumption~\ref{assu:regime}.i, hence $V_{BC}^{-1/2}\tilde{R}_{BC}=o_{p}(1)$.
Because $\Lambda_{BC}$ satisfies Assumption~\ref{assu:lindeberg}, Theorem~\ref{thm:clt} applied with
$\Lambda_{\star}=\Lambda_{BC}$ gives $\sqrt{n}\,V_{BC}^{-1/2}(\hat{\theta}_{BC}-\theta_{0})\overset{d}{\rightarrow}N(0,I_{p})$.

\noindent\textbf{Part (ii): empirical Bayes.}
Recall $\hat{\theta}_{EB}=\hat{\theta}_{BC}-\hat{\Lambda}\hat{L}/\sqrt{n}$, with $L:=\Pi M_{\mu}\eta_{n}$ and $\hat{L}:=\hat{\Pi}(\hat{r}-\hat{\mu}1_{m})$ as in the proof of Proposition~\ref{prop:eb-rate}. Subtracting $\hat{\Lambda}\hat{L}$ from the expansion of part~(i) and using $\Lambda_{BC}\eta_{n}-\Lambda L=\Lambda(I_{m}-\Pi)M_{\mu}\eta_{n}=\Lambda_{EB}\eta_{n}$ gives $\sqrt{n}(\hat{\theta}_{EB}-\theta_{0})=\Lambda_{EB}\eta_{n}+\tilde{R}_{EB}$ with
\[
\tilde{R}_{EB}=\tilde{R}_{BC}-(\hat{\Lambda}-\Lambda)L-\hat{\Lambda}(\hat{L}-L).
\]
The proof of Proposition~\ref{prop:eb-rate} decomposes $\hat{L}-L$ exactly into three terms: the first-stage error $(\hat{\Pi}(\widehat{\sigma^{2}})-\Pi(\widehat{\sigma^{2}}))M_{\mu}\eta_{n}$, the $\widehat{\sigma^{2}}$-correction $(\Pi(\widehat{\sigma^{2}})-\Pi(\sigma^{2}))M_{\mu}\eta_{n}$, and the projected linearization remainder $\hat{\Pi}u_{n}$. It shows that the first and third terms are $O_{p}(m/\sqrt{n})$, that the second is $O_{p}(1)$, and that $\|L\|=O_{p}(\sqrt{m})$. I studentize each piece of $\tilde{R}_{EB}$ in turn.

Two bounds recur. The floor in Assumption~\ref{assu:lindeberg}.iv applied to $\Lambda_{EB}$ gives $\|V_{EB}^{-1/2}\|_{2}=O(1)$. Combined with Assumption~\ref{assu:rate}.ii, it also gives $\|V_{EB}^{-1/2}\Lambda\|_{2}\le\lambda_{\min}(V_{EB})^{-1/2}\|\Lambda\|_{2}=O(1)$. The first bound and $m/n\to0$ from Assumption~\ref{assu:regime}.i studentize $\tilde{R}_{BC}$ and $(\hat{\Lambda}-\Lambda)L$, both $O_{p}(\sqrt{m/n})$, to $o_{p}(1)$. The first bound also gives $\|V_{EB}^{-1/2}(\hat{\Lambda}-\Lambda)\|_{2}=O_{p}(1/\sqrt{n})=o_{p}(1/\sqrt{m})$, which studentizes $(\hat{\Lambda}-\Lambda)(\hat{L}-L)$ to $o_{p}(1)$, because $\|\hat{L}-L\|=O_{p}(1)$ by Assumption~\ref{assu:regime}.i. It remains to studentize $\Lambda(\hat{L}-L)$. The second bound handles the first and third terms of $\hat{L}-L$, since $\|V_{EB}^{-1/2}\Lambda\|_{2}\cdot O_{p}(m/\sqrt{n})=o_{p}(1)$ by Assumption~\ref{assu:regime}.i. The $\widehat{\sigma^{2}}$-correction is $O_{p}(1)$ before studentization and requires a more detailed argument.

The studentized correction is $-V_{EB}^{-1/2}\Lambda(\Pi(\widehat{\sigma^{2}})-\Pi(\sigma^{2}))M_{\mu}\eta_{n}$, with $|\widehat{\sigma^{2}}-\sigma^{2}|=O_{p}(1/\sqrt{m})$ by Lemma~\ref{lem:sigma-consistency} and Assumption~\ref{assu:regime}.i. The proof of Proposition~\ref{prop:eb-rate} gives $\partial\Pi(s)/\partial s=(I_{m}-\Pi(s))M_{\mu}M_{\mu}'M'A(s)^{+}M$ with $A(s):=M\Sigma_{\mu}(s)M'$. The correction therefore equals the integral of $-V_{EB}^{-1/2}\Lambda(\partial\Pi(s)/\partial s)M_{\mu}\eta_{n}$ over $s$ between $\sigma^{2}$ and $\widehat{\sigma^{2}}$, and its norm is at most $|\widehat{\sigma^{2}}-\sigma^{2}|$ times the supremum of the integrand's norm on that segment. The argument uses only the rate of $\widehat{\sigma^{2}}$. No limiting distribution for it is required.

At $s=\sigma^{2}$, grouping the projection factors gives
\[
\Lambda(\partial\Pi(s)/\partial s)M_{\mu}=\Lambda_{EB}K,\qquad K:=M_{\mu}'M'(M\Sigma_{\mu}M')^{+}MM_{\mu}.
\]
Because $(M\Sigma_{\mu}M')^{+}$ is the pseudoinverse of $M\Sigma_{\mu}M'=(MM_{\mu})\Sigma(MM_{\mu})'$, the matrix $K$ is positive semidefinite and satisfies $K\Sigma K=K$. The floor $\lambda_{\min}(V)\ge c_{V}$ bounds $K$. Since $\Sigma-V=\sigma^{2}Q$ is positive semidefinite, so is $\Sigma-c_{V}I_{m}$ at every realized variance. Hence $K-c_{V}K^{2}=K(\Sigma-c_{V}I_{m})K$ is positive semidefinite, every eigenvalue $\lambda$ of $K$ satisfies $c_{V}\lambda^{2}\le\lambda$, and $\|K\|_{2}\le1/c_{V}$. The same floor makes $V_{EB}-c_{V}\Lambda_{EB}\Lambda_{EB}'=\Lambda_{EB}(\Sigma-c_{V}I_{m})\Lambda_{EB}'$ positive semidefinite, hence $\|V_{EB}^{-1/2}\Lambda_{EB}\|_{2}\le c_{V}^{-1/2}$ and $\mathrm{tr}(V_{EB}^{-1}\Lambda_{EB}\Lambda_{EB}')\le p/c_{V}$. Combining the two bounds,
\[
\mathrm{tr}(V_{EB}^{-1}\Lambda_{EB}K\Lambda_{EB}')\le c_{V}^{-1}\,\mathrm{tr}(V_{EB}^{-1}\Lambda_{EB}\Lambda_{EB}')\le p/c_{V}^{2}.
\]
The constant $c_{V}$ does not depend on $m$, so these bounds are uniform in the sample size and in the realized variance. The integrand at $s=\sigma^{2}$ is $-V_{EB}^{-1/2}\Lambda_{EB}K\eta_{n}$. Its covariance is $V_{EB}^{-1/2}\Lambda_{EB}K\Sigma_{n}K\Lambda_{EB}'V_{EB}^{-1/2}$. Writing $\Sigma_{n}=\Sigma+(V_{n}-V)$ and using $K\Sigma K=K$, the trace of this covariance equals $\mathrm{tr}(V_{EB}^{-1}\Lambda_{EB}K\Lambda_{EB}')$ plus $\mathrm{tr}(V_{EB}^{-1}\Lambda_{EB}K(V_{n}-V)K\Lambda_{EB}')$. The second trace is at most $\|V_{n}-V\|_{2}\|K\|_{2}\,\mathrm{tr}(V_{EB}^{-1}\Lambda_{EB}K\Lambda_{EB}')=O(\sqrt{m/n})$ by Assumption~\ref{assu:regime}.vii, using that $\|K\|_{2}K-K^{2}$ is positive semidefinite. The mean $V_{EB}^{-1/2}\Lambda_{EB}KE[\varepsilon_{n}]$ has norm at most $c_{V}^{-3/2}\|E[\varepsilon_{n}]\|=o(1)$ by Assumption~\ref{assu:random_b}.i. Markov's inequality gives $\|V_{EB}^{-1/2}\Lambda_{EB}K\eta_{n}\|=O_{p}(1)$.

Away from $s=\sigma^{2}$, write $v:=MM_{\mu}\eta_{n}$ and differentiate the integrand in $s$. Both pieces of the product rule coincide, and the derivative equals
\[
2V_{EB}^{-1/2}\Lambda(I_{m}-\Pi(s))M_{\mu}M_{\mu}'M'A(s)^{+}(\partial A(s)/\partial s)A(s)^{+}v.
\]
By the bound $\|V_{EB}^{-1/2}\Lambda\|_{2}=O(1)$ and Assumption~\ref{assu:eb-reg}.iii, its norm is at most a constant times $\|A(s)^{+}(\partial A(s)/\partial s)A(s)^{+}v\|$. The floor $\lambda_{\min}(V)\ge c_{V}$ makes $c_{V}^{-1}A(0)-\partial A(s)/\partial s=c_{V}^{-1}MM_{\mu}(V-c_{V}I_{m})M_{\mu}'M'$ positive semidefinite, and $A(s)-A(0)$ is positive semidefinite on the common range. Hence
\begin{align*}
\|A(s)^{+}(\partial A(s)/\partial s)^{1/2}\|_{2}^{2}&=\|A(s)^{+}(\partial A(s)/\partial s)A(s)^{+}\|_{2}\\
&\le c_{V}^{-1}\|A(s)^{+}A(s)A(s)^{+}\|_{2}=c_{V}^{-1}\|A(s)^{+}\|_{2}=O(\sqrt{m}),
\end{align*}
while $\|(\partial A(s)/\partial s)^{1/2}A(s)^{+}v\|\le\|MM_{\mu}M_{\mu}'M'\|_{2}^{1/2}\|A(s)^{+}v\|=O_{p}(\sqrt{m})$ by the bound on $\|A(s)^{+}v\|$ in the proof of Proposition~\ref{prop:eb-rate}. The derivative is therefore $O(m^{1/4})\cdot O_{p}(\sqrt{m})=O_{p}(m^{3/4})$, the integrand is $O_{p}(1)+O_{p}(1/\sqrt{m})\cdot O_{p}(m^{3/4})=O_{p}(m^{1/4})$ on the segment, and the studentized correction is $O_{p}(1/\sqrt{m})\cdot O_{p}(m^{1/4})=o_{p}(1)$. Hence $V_{EB}^{-1/2}\tilde{R}_{EB}=o_{p}(1)$.

Because $\Lambda_{EB}$ satisfies Assumption~\ref{assu:lindeberg} and $V_{EB}^{-1/2}\tilde{R}_{EB}=o_{p}(1)$, Theorem~\ref{thm:clt} applied with $\Lambda_{\star}=\Lambda_{EB}$ gives
$\sqrt{n}\,V_{EB}^{-1/2}(\hat{\theta}_{EB}-\theta_{0})\overset{d}{\rightarrow}N(0,I_{p})$.
\end{proof}

\subsection*{Proof of Corollary~\ref{corr:feasible-inference}}
\begin{proof}
For $\star\in\{BC,EB\}$, write $\delta\Lambda_{\star}:=\hat{\Lambda}_{\star}-\Lambda_{\star}$. Expanding $\hat{V}_{\star}=\hat{\Lambda}_{\star}\hat{\Sigma}\hat{\Lambda}_{\star}'$ around $\Lambda_{\star}$ gives the exact identity
\[
\hat{V}_{\star}-V_{\star}=\delta\Lambda_{\star}\hat{\Sigma}\Lambda_{\star}'+\Lambda_{\star}\hat{\Sigma}\delta\Lambda_{\star}'+\delta\Lambda_{\star}\hat{\Sigma}\delta\Lambda_{\star}'+\Lambda_{\star}(\hat{\Sigma}-\Sigma)\Lambda_{\star}'.
\]
Write
\[
a_{\star}:=\|V_{\star}^{-1/2}\delta\Lambda_{\star}\hat{\Sigma}^{1/2}\|_{2},\qquad
b_{\star}:=\|V_{\star}^{-1/2}\Lambda_{\star}(\hat{\Sigma}-\Sigma)\Lambda_{\star}'V_{\star}^{-1/2}\|_{2}.
\]
Since $V_{\star}^{-1/2}\Lambda_{\star}\Sigma\Lambda_{\star}'V_{\star}^{-1/2}=I_{p}$, the factor $\|\hat{\Sigma}^{1/2}\Lambda_{\star}'V_{\star}^{-1/2}\|_{2}^{2}$, which equals $\|V_{\star}^{-1/2}\Lambda_{\star}\hat{\Sigma}\Lambda_{\star}'V_{\star}^{-1/2}\|_{2}$, is at most $1+b_{\star}$. Studentizing the identity therefore gives
\[
\|V_{\star}^{-1/2}(\hat{V}_{\star}-V_{\star})V_{\star}^{-1/2}\|_{2}\le2a_{\star}(1+b_{\star})^{1/2}+a_{\star}^{2}+b_{\star}.
\]
It suffices to show $a_{\star}=o_{p}(1)$ and $b_{\star}=o_{p}(1)$. Two inputs recur. First, $\|\hat{\Sigma}-\Sigma\|_{2}\le\|\hat{V}-V\|_{2}+|\widehat{\sigma^{2}}-\sigma^{2}|=O_{p}(\sqrt{m/n})+O_{p}(1/\sqrt{m})=o_{p}(1)$, using $\|Q\|_{2}=1$, Assumption~\ref{assu:regime}.iii, Lemma~\ref{lem:sigma-consistency}, and Assumption~\ref{assu:regime}.i. Hence $\|\hat{\Sigma}\|_{2}=O_{p}(1)$. Second, $\hat{\Lambda}_{BC}-\Lambda_{BC}=(\hat{\Lambda}-\Lambda)\hat{M}_{\mu}-\Lambda1_{m}(\hat{w}-w)'$ has operator norm $O_{p}(1/\sqrt{n})$, by Assumption~\ref{assu:regime}.iii, $\|\Lambda1_{m}\|=O(1)$ from Assumption~\ref{assu:rate}.i, and $\|\hat{w}-w\|=O_{p}(1/\sqrt{n})$ from the proof of Proposition~\ref{prop:eb-rate}.

\noindent\textbf{Bias-corrected estimator.}
The floor in Assumption~\ref{assu:lindeberg}.iv gives $\lambda_{\min}(V_{BC})^{-1/2}=O(1)$, and with $\|\Lambda_{BC}\|_{2}=O(1)$ from Assumption~\ref{assu:lindeberg}.i, $\|V_{BC}^{-1/2}\Lambda_{BC}\|_{2}\le\lambda_{\min}(V_{BC})^{-1/2}\|\Lambda_{BC}\|_{2}=O(1)$. The term $b_{BC}$ is at most $\|V_{BC}^{-1/2}\Lambda_{BC}\|_{2}^{2}\|\hat{\Sigma}-\Sigma\|_{2}=o_{p}(1)$. The term $a_{BC}$ is at most $\lambda_{\min}(V_{BC})^{-1/2}\|\delta\Lambda_{BC}\|_{2}\|\hat{\Sigma}\|_{2}^{1/2}=O_{p}(1/\sqrt{n})=o_{p}(1)$.

\noindent\textbf{Empirical Bayes estimator.}
Since $\hat{\Pi}1_{m}=0$ and $\Pi1_{m}=0$, the projections satisfy $\hat{\Pi}\hat{M}_{\mu}=\hat{\Pi}$ and $\Pi M_{\mu}=\Pi$. Hence $\hat{\Lambda}_{EB}=\hat{\Lambda}\hat{M}_{\mu}-\hat{\Lambda}\hat{\Pi}$ and $\Lambda_{EB}=\Lambda M_{\mu}-\Lambda\Pi$. With $\delta\Pi:=\hat{\Pi}(\widehat{\sigma^{2}})-\Pi(\widehat{\sigma^{2}})$,
\[
\delta\Lambda_{EB}=(\hat{\Lambda}_{BC}-\Lambda_{BC})-(\hat{\Lambda}-\Lambda)\Pi-\hat{\Lambda}\,\delta\Pi-\hat{\Lambda}\big(\Pi(\widehat{\sigma^{2}})-\Pi(\sigma^{2})\big).
\]
The four pieces are the BC sensitivity and centering errors, the sensitivity error acting through the shrinkage operator, the first-stage error in the shrinkage operator, and its $\widehat{\sigma^{2}}$-dependence. Three further inputs are used. The floor $\lambda_{\min}(V)\ge c_{V}$ holds for $\Sigma$ as well, hence for any matrix $B$, in operator or Frobenius norm, $\|B\hat{\Sigma}^{1/2}\|^{2}\le\|B\Sigma^{1/2}\|^{2}+\|\hat{\Sigma}-\Sigma\|_{2}\|B\|^{2}\le(1+\|\hat{\Sigma}-\Sigma\|_{2}/c_{V})\|B\Sigma^{1/2}\|^{2}$. Each piece may therefore be bounded with $\Sigma^{1/2}$ in place of $\hat{\Sigma}^{1/2}$. Next, $\|V_{EB}^{-1/2}\Lambda_{EB}\|_{2}\le c_{V}^{-1/2}$ and $\|V_{EB}^{-1/2}\Lambda\|_{2}=O(1)$ from the proof of Corollary~\ref{corr:bc-eb-clt}(ii), and $\|V_{EB}^{-1/2}\hat{\Lambda}\|_{2}\le\|V_{EB}^{-1/2}\Lambda\|_{2}+\lambda_{\min}(V_{EB})^{-1/2}\|\hat{\Lambda}-\Lambda\|_{2}=O_{p}(1)$ by the floor in Assumption~\ref{assu:lindeberg}.iv. Finally, $\sup_{s\ge0}\|\Pi(s)\|_{2}\le\|M_{\mu}\|_{2}+\sup_{s\ge0}\|(I_{m}-\Pi(s))M_{\mu}\|_{2}=O(1)$ by Assumption~\ref{assu:eb-reg}.iii, using $\Pi(s)=\Pi(s)M_{\mu}$.

The term $b_{EB}$ is at most $c_{V}^{-1}\|\hat{\Sigma}-\Sigma\|_{2}=o_{p}(1)$. For $a_{EB}$, the first two pieces of $\delta\Lambda_{EB}$ have operator norm $O_{p}(1/\sqrt{n})$, and their contributions are bounded exactly as $a_{BC}$ was. The third piece contributes at most $\|V_{EB}^{-1/2}\hat{\Lambda}\|_{2}\|\delta\Pi\,\hat{\Sigma}^{1/2}\|_{F}$, and $\|\delta\Pi\,\Sigma^{1/2}\|_{F}=\|\delta\Pi\,\Sigma_{\mu}^{1/2}\|_{F}$ because $\delta\Pi=\delta\Pi M_{\mu}$. The second part of Assumption~\ref{assu:eb-reg}.ii and Markov's inequality give $\sup_{s\ge0}\|(\hat{\Pi}(s)-\Pi(s))\Sigma_{\mu}^{1/2}\|_{F}=O_{p}(m/\sqrt{n})$, which covers $s=\widehat{\sigma^{2}}$ even though $\widehat{\sigma^{2}}$ depends on the data. The contribution is $O_{p}(1)\cdot O_{p}(m/\sqrt{n})=o_{p}(1)$ by Assumption~\ref{assu:regime}.i.

The fourth piece is the $\widehat{\sigma^{2}}$-correction. Splitting $\hat{\Lambda}=\Lambda+(\hat{\Lambda}-\Lambda)$, the part with $\hat{\Lambda}-\Lambda$ contributes at most $\lambda_{\min}(V_{EB})^{-1/2}\|\hat{\Lambda}-\Lambda\|_{2}\cdot2\sup_{s\ge0}\|\Pi(s)\|_{2}\,\|\hat{\Sigma}\|_{2}^{1/2}=o_{p}(1)$, again by the floor. The part with $\Lambda$ contributes at most a $1+o_{p}(1)$ multiple of $\|V_{EB}^{-1/2}\Lambda(\Pi(\widehat{\sigma^{2}})-\Pi(\sigma^{2}))\Sigma^{1/2}\|_{2}$. Integrating the derivative $\partial\Pi(s)/\partial s$ of Proposition~\ref{prop:eb-rate}'s proof over $s$ between $\sigma^{2}$ and $\widehat{\sigma^{2}}$ gives
\[
\|V_{EB}^{-1/2}\Lambda(\Pi(\widehat{\sigma^{2}})-\Pi(\sigma^{2}))\Sigma^{1/2}\|_{2}\le|\widehat{\sigma^{2}}-\sigma^{2}|\,\sup_{s}\|V_{EB}^{-1/2}\Lambda(\partial\Pi(s)/\partial s)\Sigma^{1/2}\|_{2},
\]
with the supremum over that segment. At $s=\sigma^{2}$ the operator is $V_{EB}^{-1/2}\Lambda_{EB}K\Sigma^{1/2}$, with $K$ as in the proof of Corollary~\ref{corr:bc-eb-clt}(ii), and because $K\Sigma K=K$,
\[
\|V_{EB}^{-1/2}\Lambda_{EB}K\Sigma^{1/2}\|_{2}^{2}=\|V_{EB}^{-1/2}\Lambda_{EB}K\Lambda_{EB}'V_{EB}^{-1/2}\|_{2}\le\mathrm{tr}(V_{EB}^{-1}\Lambda_{EB}K\Lambda_{EB}')\le p/c_{V}^{2}
\]
by the bound established there. Away from $s=\sigma^{2}$, the derivative of the operator in $s$ is
\[
2V_{EB}^{-1/2}\Lambda(I_{m}-\Pi(s))M_{\mu}M_{\mu}'M'A(s)^{+}(\partial A(s)/\partial s)A(s)^{+}MM_{\mu}\Sigma^{1/2},
\]
as in that proof, with $A(s):=M\Sigma_{\mu}(s)M'$. Its norm is at most a constant times the product of $\|A(s)^{+}(\partial A(s)/\partial s)^{1/2}\|_{2}$ and $\|A(s)^{+}MM_{\mu}\Sigma^{1/2}\|_{2}$, and the first factor is $O(m^{1/4})$ by the floor $\lambda_{\min}(V)\ge c_{V}$, as shown there. For the second factor, $\|A(s)^{+}MM_{\mu}\Sigma^{1/2}\|_{2}^{2}=\|A(s)^{+}A(\sigma^{2})A(s)^{+}\|_{2}$, and the floor makes $(1+|s-\sigma^{2}|/c_{V})A(s)-A(\sigma^{2})$ positive semidefinite on the common range, hence $\|A(s)^{+}A(\sigma^{2})A(s)^{+}\|_{2}\le(1+|s-\sigma^{2}|/c_{V})\|A(s)^{+}\|_{2}=O(\sqrt{m})$. The derivative is therefore $O(m^{1/4})\cdot O(m^{1/4})=O(\sqrt{m})$, the supremum is $O(1)+O_{p}(1/\sqrt{m})\cdot O(\sqrt{m})=O_{p}(1)$, and the fourth piece contributes $O_{p}(1/\sqrt{m})$ to $a_{EB}$. Hence $a_{EB}=o_{p}(1)$.

\noindent\textbf{Conclusion.}
In both cases $\|V_{\star}^{-1/2}(\hat{V}_{\star}-V_{\star})V_{\star}^{-1/2}\|_{2}=o_{p}(1)$. The eigenvalues of $V_{\star}$ are bounded below by Assumption~\ref{assu:lindeberg}.iv and above by $\|\Lambda_{\star}\|_{2}^{2}\|\Sigma\|_{2}=O_{p}(1)$, by Assumptions~\ref{assu:lindeberg}.i, \ref{assu:regime}.ii, and \ref{assu:lindeberg}.iv. Hence $\|\hat{V}_{\star}-V_{\star}\|_{2}=o_{p}(1)$, and with probability approaching one both matrices lie in a compact set of positive definite matrices, on which the inverse square root is uniformly continuous. Therefore $\hat{V}_{\star}^{-1/2}V_{\star}^{1/2}-I_{p}=(\hat{V}_{\star}^{-1/2}-V_{\star}^{-1/2})V_{\star}^{1/2}=o_{p}(1)$. Slutsky's theorem combined with Corollary~\ref{corr:bc-eb-clt} yields both statements.
\end{proof}

\clearpage

\section{Verifying asymptotic assumptions\protect\label{app:verify}}

In this appendix, I verify that the design of Section~\ref{subsec:Simulation-design} satisfies the
outstanding conditions of Assumptions~\ref{assu:regime}--\ref{assu:lindeberg} when $m$ grows with $n$. Assumptions~\ref{assu:regime} and \ref{assu:rate} are verified analytically, as are condition (i) of Assumption~\ref{assu:eb-reg} and conditions (i), (ii), and (iv) of Assumption~\ref{assu:lindeberg}. The remaining conditions of Assumptions~\ref{assu:eb-reg} and \ref{assu:lindeberg} are checked numerically along the sequence $m=\mathrm{round}(n^{0.4})$ considered in Section~\ref{subsec:asymptotic}.

\subsection*{Assumption~\ref{assu:regime}: array regularity}

\subsubsection*{Normalized moment variance}

The residualized cell moments have asymptotic variance $V=\sigma_{e}^{2}S^{-1}$, where
$S=E[\tilde{Z}_{i}\tilde{Z}_{i}']=\mathrm{diag}(q)-qq'$. Although this has the form of a
multinomial covariance, here it is positive definite, because $\sum_{j}q_{j}=0.9<1$:
the Sherman--Morrison identity gives $S^{-1}=\mathrm{diag}(1/q_{j})+10\cdot1_{m}1_{m}'$,
with the factor $10=1/(1-\sum_{j}q_{j})$. The operator norm of $V$ grows in $m$, so I
check Assumption~\ref{assu:regime}.ii on the normalized matrix $V/\|V\|_{2}$. The extreme
cell $j_{\star}$ with $q_{j_{\star}}=q_{\min}$ drives the largest row sum of $S^{-1}$,
$1/q_{\min}+10m$, and its operator norm, $\|S^{-1}\|_{2}\ge(S^{-1})_{j_{\star}j_{\star}}=1/q_{\min}+10$.
Their ratio is at most $1+10mq_{\min}\le10$, using $q_{\min}\le0.9/m$. The normalized row
sums of $V$ are therefore $O(1)$. The weight $W=S$ has absolute row sums $\sum_{k}|S_{jk}|=q_{j}(1.9-2q_{j})\le1.9q_{\max}$. Its operator norm is of the same order, $\|S\|_{2}\ge q_{\max}-q_{\max}^{2}$, from the diagonal entry of the largest cell. The normalized row sums are $O(1)$ as well. Therefore
Assumption~\ref{assu:regime}.ii holds.

\subsubsection*{First-stage convergence}

Assumption~\ref{assu:regime}.iii requires the first-stage estimators to converge at rate $O_{p}(\sqrt{m/n})$. The estimator $\hat{V}$ attains this rate in normalized operator norm. The cell-frequency
errors $\hat{q}_{j}-q_{j}=O_{p}(\sqrt{q_{j}/n})$ scale with cell size. The
near-singular inversion at the extreme cell is therefore controlled whenever
$nq_{\min}\to\infty$, which $m^{2}/n\to0$ ensures. The Jacobian estimate $\hat{G}=-\hat{\pi}$ and the weight $\hat{S}$ are smooth functions of the cell frequencies and sample cross-moments. Hence, the residual maker $\hat{M}=I_{m}-\hat{G}(\hat{G}'\hat{S}\hat{G})^{-1}\hat{G}'\hat{S}$ inherits the rate $\|\hat{M}-M\|_{2}=O_{p}(\sqrt{m/n})$. The sensitivity is the $1\times m$ row $\hat{\Lambda}=-(\hat{G}'\hat{S}\hat{G})^{-1}\hat{G}'\hat{S}$. Its operator norm is its Euclidean norm. The scalar $G'SG$ is bounded away from zero, as shown below. Its plug-in $\hat{G}'\hat{S}\hat{G}$ carries an $O_{p}(1/\sqrt{n})$ stochastic error and an $O(m/n)$ overfitting bias. The bias is $o(1/\sqrt{n})$ under $m^{2}/n\to0$. Hence $(\hat{G}'\hat{S}\hat{G})^{-1}-(G'SG)^{-1}=O_{p}(1/\sqrt{n})$. The row $G'S$ has $j$th entry $q_{j}(G_{j}-q'G)$ and norm $O(1/\sqrt{m})$. Its plug-in error has mean square $O(q_{j}/n)$ entry by entry. The cell weights $q_{j}=O(1/m)$ keep these errors small in the aggregate: $E\|\hat{G}'\hat{S}-G'S\|^{2}=\sum_{j}O(q_{j}/n)=O(1/n)$ because $\sum_{j}q_{j}\le1$, hence $\|\hat{G}'\hat{S}-G'S\|=O_{p}(1/\sqrt{n})$. Combining the scalar and row errors gives $\|\hat{\Lambda}-\Lambda\|_{2}=O_{p}(1/\sqrt{n})$. The shrinking cell weights offset the $\sqrt{m}$ dimension penalty that governs $\hat{V}$ and $\hat{M}$. The residualized moment $g(\theta)=\delta-\theta\pi$ is affine in $\theta$. Hence, its linearization remainder is identically zero.

\subsubsection*{The eigenvalue floor}

For Assumption~\ref{assu:regime}.iv I bound the condition number of $S$. For any $v$, Cauchy--Schwarz gives
\[
v'S v=\sum_{j}q_{j}v_{j}^{2}-\Bigl(\sum_{j}q_{j}v_{j}\Bigr)^{2}\ge\Bigl(1-\sum_{j}q_{j}\Bigr)\sum_{j}q_{j}v_{j}^{2}=0.1\sum_{j}q_{j}v_{j}^{2}\ge0.1\,q_{\min}\|v\|^{2},
\]
so $\lambda_{\min}(S)\ge0.1\,q_{\min}$. Dropping the subtracted term, $v'S v\le\sum_{j}q_{j}v_{j}^{2}\le q_{\max}\|v\|^{2}$, hence $\lambda_{\max}(S)\le q_{\max}$. The condition number is therefore bounded by the fixed cell-size spread, $\lambda_{\max}(S)/\lambda_{\min}(S)\le10\,(q_{\max}/q_{\min})=O(1)$, uniformly in $m$.

With $W=S$ and $MG=0$, the residual maker obeys $MS^{-1}M'=S^{-1}-G(G'S G)^{-1}G'$, which is $S^{-1}$ less a rank-one positive semidefinite matrix. Its eigenvalues therefore interlace those of $S^{-1}$, and the smallest nonzero one obeys $\lambda_{\min}^{+}(MS^{-1}M')\ge\lambda_{\min}(S^{-1})=1/\lambda_{\max}(S)$. Hence the unnormalized $\lambda_{\min}^{+}(MVM')\ge\sigma_{e}^{2}/\lambda_{\max}(S)$, and dividing by $\|V\|_{2}=\sigma_{e}^{2}/\lambda_{\min}(S)$ leaves the normalized floor $\lambda_{\min}^{+}(MVM')\ge\lambda_{\min}(S)/\lambda_{\max}(S)$, bounded away from zero. The rate condition $\lambda_{\min}^{+}(MVM')\cdot n/m\to\infty$ of Assumption~\ref{assu:regime}.iv then holds with room to spare. Because $\sigma^{2}Q$ is positive semidefinite, the same floor carries to $M\Sigma M'$ for every $\sigma^{2}\ge0$.

This floor also reaches the centered covariance $M\Sigma_{\mu}M'$ of Assumption~\ref{assu:eb-reg}. Lemma~\ref{lem:rank-centered} shows $MM_{\mu}=(I_{m}-P)M$, with $P$ the orthogonal projector onto $\mathrm{span}(M1_{m})$. The centered covariance is then $M\Sigma_{\mu}M'=(I_{m}-P)(M\Sigma M')(I_{m}-P)$, the compression of $M\Sigma M'$ to the subspace orthogonal to $M1_{m}\in\mathrm{col}(M)$. This compression drops one direction from the range. The nonzero eigenvalues interlace, and $\lambda_{\min}^{+}(M\Sigma_{\mu}M')\ge\lambda_{\min}^{+}(M\Sigma M')$. The normalized floor carries over. All $m-p-1$ nonzero eigenvalues are bounded away from zero. The squared trace is then $\mathrm{tr}\big(((M\Sigma_{\mu}M')^{+})^{2}\big)=\sum_{j}\lambda_{j}^{-2}\le(m-p-1)/\lambda_{\min}^{+}(M\Sigma_{\mu}M')^{2}=O(m)$, which is Assumption~\ref{assu:eb-reg}.i at $s=0$. The floor on $M\Sigma M'$ holds for every $\sigma^{2}\ge0$, and the nonzero eigenvalues are smallest at $\sigma^{2}=0$, where the condition is stated.

\subsubsection*{The Jacobian and residual maker}

Assumption~\ref{assu:regime}.v bounds the entries of $G$ and requires $\|M\|_{2}=O(1)$. The first-stage covariance is $\mathrm{Cov}_{n}(T_{i},Z_{i})=\mathrm{diag}(q)(\Phi(\pi^{*})-\bar{T}1_{m})$, with $\bar{T}=q'\Phi(\pi^{*})+0.05$, the $0.05$ being the omitted category's mass $0.1$ times its first stage $\tfrac12$. Since $S^{-1}\mathrm{diag}(q)=I+10\cdot1_{m}q'$, the Jacobian is
\[
G=-S^{-1}\mathrm{Cov}_{n}(T_{i},Z_{i})=-\bigl(\Phi(\pi^{*})-\tfrac12 1_{m}\bigr),
\]
with entries in $[-\tfrac12,\tfrac12]$. This gives the bounded-$G$ clause of Assumption~\ref{assu:regime}.v, and Assumption~\ref{assu:rate} for the scalar slope.

It remains to verify that $\|M\|_{2}=O(1)$. The matrix $H:=I-M=G(G'S G)^{-1}G'S$ has rank one, so $\|M\|_{2}=\|H\|_{2}=\|G\|\,\|S G\|/(G'S G)$. Here $\|G\|^{2}=\sum_{j}(\Phi(\pi^{*}_{j})-\tfrac12)^{2}=O(m)$. The $j$th entry of $S G$ is $q_{j}(G_{j}-q'G)$, so $\|S G\|^{2}=\sum_{j}q_{j}^{2}(G_{j}-q'G)^{2}\le(\max_{j}q_{j})\sum_{j}q_{j}(G_{j}-q'G)^{2}=O(1/m)$. The denominator decomposes as $G'S G=\sum_{j}q_{j}(G_{j}-q'G)^{2}+(1-\sum_{j}q_{j})(q'G)^{2}$. The first term is the cell-weighted first-stage variance, bounded away from zero. The second term is nonnegative. Hence $\|M\|_{2}=O(\sqrt{m})\cdot O(1/\sqrt{m})/O(1)=O(1)$.

\subsubsection*{The moment-noise covariance}

Assumption~\ref{assu:regime}.vii has two clauses. The residualized cell-moment noise is
\[
\varepsilon_{n}=n^{-1/2}\sum_{i}S^{-1}e_{i}\tilde{Z}_{i},
\]
built from the Gaussian error $e_{i}=\sigma_{e}(\rho\nu_{i}+\sqrt{1-\rho^{2}}\eta_{i})$ and the bounded demeaned instrument $\tilde{Z}_{i}=Z_{i}-\hat{q}$, with $\hat{q}=n^{-1}\sum_{i}Z_{i}$ the sample cell frequencies. The products $e_{i}Z_{i}$ are i.i.d.\ and mean-zero. The sample demeaning subtracts the common term $\hat{q}$, removing one degree of freedom. A direct calculation gives $\mathrm{Var}(n^{-1/2}\sum_{i}e_{i}\tilde{Z}_{i})=\sigma_{e}^{2}\tfrac{n-1}{n}S$ with $S=\mathrm{diag}(q)-qq'$. The residualized noise inherits this factor: $\mathrm{Var}(\varepsilon_{n})=\sigma_{e}^{2}S^{-1}\tfrac{n-1}{n}$. In the normalization $\|V\|_{2}=1$ of Assumption~\ref{assu:regime}.ii, the limiting covariance is $V\propto\sigma_{e}^{2}S^{-1}$ and $V_{n}=\tfrac{n-1}{n}V$. The difference $V_{n}-V=-V/n$ has norm $\|V\|_{2}/n=1/n$, which is $o(\sqrt{m/n})$. Hence, the first clause holds. The error $e_{i}$ is Gaussian, hence $\varepsilon_{n}$ is Gaussian conditional on the instruments, with covariance $\Sigma_{n}^{Z}=\sigma_{e}^{2}S^{-1}\hat{S}S^{-1}$ for $\hat{S}=n^{-1}\sum_{i}\tilde{Z}_{i}\tilde{Z}_{i}'$. Conditional on the instruments, $\varepsilon_{n}'M'M\varepsilon_{n}$ has exact variance $2\,\mathrm{tr}((M'M\Sigma_{n}^{Z})^{2})$, with no fourth-cumulant term. The large cell entries of $S^{-1}\tilde{Z}_{i}$ enter only through $\Sigma_{n}^{Z}$, which tends to $V$ at the rate of Assumption~\ref{assu:regime}.iii. This conditional variance is the relevant one, since the noise is Gaussian given the instruments. Hence, the second clause holds with $C=2$. This rests on the moments of $\varepsilon_{n}$, not those of the specification errors $b$.

\subsection*{Assumption~\ref{assu:rate}: rate and studentization conditions}

\subsubsection*{The BC estimator}

The residualized scalar moment makes $\hat{\beta}$ single-parameter TSLS on the cell moments. Standard TSLS algebra gives the leading-order forms
\[
V_{BC}\ge\frac{\sigma_{e}^{2}}{\sum_{k}q_{k}(\Phi(\pi^{*}_{k})-\bar{\Phi})^{2}},\qquad\Lambda_{BC,j}=\frac{q_{j}(\Phi(\pi^{*}_{j})-\bar{\Phi})}{\sum_{k}q_{k}(\Phi(\pi^{*}_{k})-\bar{\Phi})^{2}},
\]
where $\bar{\Phi}:=\sum_{k}q_{k}\Phi(\pi^{*}_{k})/\sum_{k}q_{k}$ is the cell-size--weighted
mean first stage, and the demeaning ensures $\Lambda_{BC}1_{m}=0$. The floor on $V_{BC}$ in Assumption~\ref{assu:lindeberg}.iv holds because the scalar $V_{BC}$ is bounded away from zero: its noise component $\sigma_{e}^{2}/\sum_{k}q_{k}(\Phi(\pi^{*}_{k})-\bar{\Phi})^{2}$ is at least $\sigma_{e}^{2}$ because the denominator is at most one. The studentized sensitivity that the proof of Corollary~\ref{corr:feasible-inference} bounds satisfies $\|V_{BC}^{-1/2}\Lambda_{BC}\|_{2}^{2}=\sum_{j}\Lambda_{BC,j}^{2}/V_{BC}\le\max_{j}q_{j}/\sigma_{e}^{2}=O(1/m)$.

\subsubsection*{The EB estimator}

The floor on $V_{EB}$ in Assumption~\ref{assu:lindeberg}.iv follows from a lower bound on the scalar $V_{EB}$, which Lemma~\ref{lem:eb-floor} supplies once its two design conditions are checked. First, $G'V^{-1}G=\sigma_{e}^{-2}G'S G=O(1)$, the entries of $G$ being bounded. Second, the moment covariance $V=\sigma_{e}^{2}S^{-1}$ has $\lambda_{\min}(V)$ bounded away from zero because $S$ is well-conditioned. This is also the floor $\lambda_{\min}(V)\ge c_{V}$ that Corollary~\ref{corr:bc-eb-clt}(ii) imposes. The lemma then gives $V_{EB}\ge1/(G'V^{-1}G)$, bounded away from zero uniformly in $m$ and $\sigma^{2}\ge0$.

\subsection*{Assumptions~\ref{assu:eb-reg} and \ref{assu:lindeberg}: design diagnostics}

The remaining conditions concern the shrinkage operators $\Pi(s)$ and the plug-in errors $\hat{\Pi}(s)-\Pi(s)$. Assumption~\ref{assu:eb-reg} asks that the centered covariance be well conditioned on average in squared trace at $s=0$ (i), that the plug-in track its population value in mean square at rate $m^{2}/n$ uniformly in $s$ (ii), and that the shrinkage residual operator stay bounded uniformly in $s$ (iii). Condition (i) was established analytically in the preceding subsection. The table below supplies numerical checks for the remaining conditions. 

I evaluate the relevant population quantities along the simulation path $m=\mathrm{round}(n^{0.4})$ of Table~\ref{tab:mc-bias-vs-n}. Table~\ref{tab:appb-diagnostics} reports the results. Every entry except the plug-in row is exact. The plug-in row is a Monte Carlo average, as the table notes describe.

\setcounter{table}{0}
\renewcommand{\thetable}{B.\arabic{table}}%
\begin{table}[H]
\centering
\caption{Design diagnostics along the simulation path $m=\mathrm{round}(n^{0.4})$\protect\label{tab:appb-diagnostics}}
\begin{centering}
\small
\setlength{\tabcolsep}{3pt}%
\begin{tabular}{lccccccc}
\toprule
$n$ & 5{,}000 & 10{,}000 & 20{,}000 & 50{,}000 & 100{,}000 & 200{,}000 & 500{,}000\tabularnewline
$m$ & 30 & 40 & 53 & 76 & 100 & 132 & 190\tabularnewline
\midrule
$\mathrm{tr}\big(((M\Sigma_{\mu}(0)M')^{+})^{2}\big)/m$ & 209 & 209 & 210 & 207 & 206 & 205 & 205\tabularnewline
$(n/m^{2})\,E\big[\sup_{s}\|(\hat{\Pi}(s)-\Pi(s))M_{\mu}r_{n}\|^{2}\big]$ & 4.27 & 2.81 & 1.86 & 1.17 & 0.87 & 0.63 & 0.45\tabularnewline
$\max_{s}\|(I_{m}-\Pi(s))M_{\mu}\|_{2}$ & 2.29 & 2.33 & 2.36 & 2.38 & 2.39 & 2.40 & 2.40\tabularnewline
\midrule
$\max_{j}\Lambda_{BC,j}^{2}/V_{BC}\ (\times10^{3})$ & 0.491 & 0.284 & 0.172 & 0.0915 & 0.0560 & 0.0338 & 0.0172\tabularnewline
$\max_{j}\Lambda_{EB,j}^{2}/V_{EB}\ (\times10^{3})$ & 0.889 & 0.477 & 0.273 & 0.136 & 0.0800 & 0.0470 & 0.0232\tabularnewline
\bottomrule
\end{tabular}
\par\end{centering}
\raggedright{\small\emph{Notes:} Each entry is computed from the cell-IV design at the $(n,m)$ pairs of Table~\ref{tab:mc-bias-vs-n}. All rows except the plug-in row are exact. The squared-trace and plug-in rows are reported in the normalization $\|V\|_{2}=1$ of Assumption~\ref{assu:regime}.ii. In the design scale, with $\sigma_{e}=1$ and $V=S^{-1}$, the operator norm $\|V\|_{2}=1/\lambda_{\min}(S)$ equals 390, 513, 673, 955, 1251, 1645, and 2360 along the path. Because $\Pi(s)$ is invariant to the scale of the moments and $M_{\mu}r_{n}$ carries that scale, the squared-trace entries computed in the design scale are multiplied by $\|V\|_{2}^{2}$ and the plug-in entries divided by $\|V\|_{2}$. The remaining rows are scale-free. The squared-trace row is evaluated at $s=0$, where Assumption~\ref{assu:eb-reg}.i is stated. The shrinkage-residual row reports the maximum over the variance argument $s$ on the grid $\{0\}\cup\{1,2,4,\dots,2^{20}\}$ together with the limit $s\to\infty$, at which both projections converge to the projection built from $M_{\mu}M_{\mu}'$ alone. It attains its maximum at $s=0$. The plug-in row reports the expectation of the supremum over the same grid and limit, with the covariance of $M_{\mu}r_{n}$ held at the design value $\sigma^{2}=64$ throughout. The leverage rows are evaluated at $\sigma^{2}=64$. The squared-trace, shrinkage-residual, and leverage rows are direct functions of the population matrices. The plug-in row is simulated. I linearize $\hat{\Pi}(s)$ by the delta method in its estimated inputs, the cell frequencies and the cell means of the treatment, including the omitted cell, and draw these first-stage errors jointly with the centered moment $M_{\mu}r_{n}$ from their first-order Gaussian law, which carries the design's correlation $\rho=0.1$ between first-stage and structural errors. Each of $4{,}000$ draws yields the maximum over $s$, and the row averages them. Monte Carlo standard errors are below two percent of the entries. The Frobenius form in the second part of Assumption~\ref{assu:eb-reg}.ii, computed from the same draws, gives 3.63, 2.32, 1.50, 0.90, 0.62, 0.43, and 0.28. The maximum over $s$ of the expectation, the weaker quantity, gives 3.70, 2.40, 1.61, 1.01, 0.73, 0.55, and 0.40. Entries in rows marked $(\times10^{3})$ have been multiplied by the indicated factor.}{\small\par}
\end{table}

The first three rows of Table~\ref{tab:appb-diagnostics} concern Assumption~\ref{assu:eb-reg}. The scaled squared trace $\mathrm{tr}\big(((M\Sigma_{\mu}(0)M')^{+})^{2}\big)/m$ stays between 205 and 210 along the path, consistent with condition (i) and with the normalized floor established above. The shrinkage residual $\|(I_{m}-\Pi(s))M_{\mu}\|_{2}$ stays below $2.4$ at every $s$ as $m$ grows from $30$ to $190$, the numerical check for condition (iii). The scaled plug-in constant $(n/m^{2})\,E[\sup_{s}\|(\hat{\Pi}(s)-\Pi(s))M_{\mu}r_{n}\|^{2}]$ declines from 4.27 to 0.45, and its Frobenius part, the second part of condition (ii), from 3.63 to 0.28. Both satisfy the stated rate with room to spare. The design's endogeneity adds a mean component to the projection error. The first stage and the moments are estimated from the same observations. The first-stage error in a cell is therefore correlated with that cell's moment noise, and $(\hat{\Pi}(s)-\Pi(s))M_{\mu}r_{n}$ acquires a nonzero mean. The squared norm of that mean is at most about $0.13\,m^{2}/n$ along the entire path, within the rate of condition (ii), and its share of the row grows as the remaining part declines. Its image under the sensitivity row, divided by $\sqrt{n}$, is about $0.09$ standard errors of $\hat{\theta}_{EB}$ at every point of the path. This is the mean of the first-stage remainder in the proof of Corollary~\ref{corr:bc-eb-clt}(ii), which is $O_{p}(m/\sqrt{n})$, and $m/\sqrt{n}=n^{-0.1}$ declines slowly along $m=\mathrm{round}(n^{0.4})$.

Assumption~\ref{assu:lindeberg} asks that the leading-term variances be positive definite (i), that a normal limit hold for the scaled noise (ii), that no single moment dominate (iii), and that $\sigma^{2}=O_{p}(1)$ with $\lambda_{\min}(V_{\star})$ bounded away from zero (iv). Condition (i) follows from the variance bounds above. Condition (iv) holds because the realized variance of the $t_{5}$ specification draws converges in distribution, hence is bounded in probability, and the floors on $V_{BC}$ and $V_{EB}$ established above are uniform in $\sigma^{2}$.

Condition (ii) is the sampling central limit theorem for the noise term. The noise $\Lambda\varepsilon_{n}=n^{-1/2}\sum_{i}e_{i}(\Lambda S^{-1}\tilde{Z}_{i})$ is a scalar sum of i.i.d.\ observation-level contributions. With Gaussian $e_{i}$ it is exactly normal conditional on the instruments. Assumption~\ref{assu:lindeberg}.ii therefore holds.

The last two rows of Table~\ref{tab:appb-diagnostics} report the per-moment leverage $\max_{j}\Lambda_{\star,j}^{2}/V_{\star}$ of Assumption~\ref{assu:lindeberg}.iii for the BC and EB estimators. Both fall from about $5\times10^{-4}$ to $2\times10^{-5}$ across the grid; no single cell moment dominates. For BC the ratio is at most $\max_{j}q_{j}/\sigma_{e}^{2}=O(1/m)$, and the EB ratio falls at the same rate. Assumption~\ref{assu:lindeberg}.iii therefore holds.

\clearpage

\section{Differenced and k-class estimators\protect\label{app:differenced}}

This appendix implements the differencing strategy of Section~\ref{sec:A-differencing-approach} in the Monte Carlo design of Section~\ref{sec:monte-carlo} and compares it with the MBTSLS estimator of \textcite{kolesar2015identification}. Both approaches avoid estimating the mean specification error $\mu$. Differencing eliminates it. MBTSLS relies on an orthogonality condition that fails in this design when $\bar{\mu}\neq0$.

\subsection*{Differenced moment conditions}

Let $D$ denote the $(m-1)\times m$ differencing matrix of Section~\ref{sec:A-differencing-approach}, with $D1_{m}=0$ and $DD'=I_{m-1}$. Every quantity below is invariant to the choice of $D$: any two such matrices differ by an orthogonal rotation of the rows, which cancels from the estimators and their variance estimators. Applying $D$ to the rescaled sample moment condition of Section~\ref{subsec:sim_moment} yields
\[
g_{D}(\theta):=D\hat{g}(\theta)=\delta_{D}-\theta\pi_{D},\qquad\delta_{D}:=D\hat{\delta},\quad\pi_{D}:=D\hat{\pi},
\]
where $\hat{\delta}=\hat{S}^{-1}n^{-1}\sum_{i}\tilde{Z}_{i}\tilde{Y}_{i}$ and $\hat{\pi}=\hat{S}^{-1}n^{-1}\sum_{i}\tilde{Z}_{i}\tilde{T}_{i}$ are the sample reduced-form and first-stage coefficients.

Because $D1_{m}=0$, the differenced specification errors $Db$ have mean zero regardless of $\mu$. Two exact simplifications follow. First, no mean centering is needed: the differenced analogs of $\hat{\mu}$ and $\hat{M}_{\mu}$ are $0$ and $I_{m-1}$. Second, $Q$ acts as the identity on the orthogonal complement of $1_{m}$, so $DQD'=I_{m-1}$: the dispersion of the specification errors contributes $\sigma^{2}I_{m-1}$ to the variance of the differenced moments, as in Section~\ref{sec:A-differencing-approach}.

I weight the differenced moments with $W_{D}:=(D\hat{S}^{-1}D')^{-1}$. Because the level moments were rescaled by $\hat{S}^{-1}$, the matrix $D\hat{S}^{-1}D'$ is the second-moment matrix of the differenced instruments $D\hat{S}^{-1}\tilde{Z}_{i}$. GMM with this weighting therefore reproduces TSLS on the differenced system, continuing the convention of Section~\ref{subsec:sim_moment}. Note that $W_{D}$ differs from $(D\hat{S}D')^{-1}$.

\subsection*{Differenced GMM and EB}

The differenced GMM estimator, its sensitivity, and its residual maker are
\[
\hat{\beta}_{D}=\frac{\pi_{D}'W_{D}\delta_{D}}{\pi_{D}'W_{D}\pi_{D}},\qquad\Lambda_{D}=-(G_{D}'W_{D}G_{D})^{-1}G_{D}'W_{D},\qquad M_{D}=I_{m-1}+G_{D}\Lambda_{D},
\]
with $G_{D}:=-\pi_{D}$. I report two variance estimators for $\hat{\beta}_{D}$. The naive estimator $\Lambda_{D}V_{D}\Lambda_{D}'$ uses $V_{D}:=D\hat{V}D'$, where $\hat{V}$ is the heteroscedasticity-robust moment covariance estimate built from residuals at $\hat{\beta}_{D}$, and ignores misspecification. The misspecification-aware estimator is $\Lambda_{D}\Sigma_{D}\Lambda_{D}'$, where $\Sigma_{D}:=V_{D}+\widehat{\sigma_{D}^{2}}I_{m-1}$ and
\[
\widehat{\sigma_{D}^{2}}:=\max\left(0,\ \frac{r_{D}'r_{D}-\mathrm{tr}(M_{D}V_{D}M_{D}')}{\mathrm{tr}(M_{D}M_{D}')}\right),\qquad r_{D}:=\sqrt{n}\,g_{D}(\hat{\beta}_{D}).
\]
This is the dispersion estimator of Section~\ref{sec:estimation} with $Q$ replaced by $I_{m-1}$ and no mean subtraction.

The differenced EB estimator is the pure-shrinkage construction of Remark~\ref{rem:eb-pure-shrinkage} applied to the differenced system:
\[
\hat{\beta}_{D,EB}:=\hat{\beta}_{D}-\frac{1}{\sqrt{n}}\Lambda_{D}\Pi_{D}r_{D},\qquad\Pi_{D}:=\Sigma_{D}M_{D}'(M_{D}\Sigma_{D}M_{D}')^{+}M_{D}.
\]
There is no bias-correction step and no recentering of $r_{D}$: both collapse because the mean has been differenced out. The misspecification-aware variance estimator is $\Lambda_{D,EB}\Sigma_{D}\Lambda_{D,EB}'$ with $\Lambda_{D,EB}:=\Lambda_{D}(I_{m-1}-\Pi_{D})$. The results below therefore double as a demonstration of the estimator available when $\mu$ is unidentified.

\subsection*{MBTSLS}

The MBTSLS estimator of \textcite{kolesar2015identification} is the k-class estimator
\[
\hat{\beta}_{MB}:=\frac{(1-k)\,n^{-1}\tilde{T}'\tilde{Y}+k\,\hat{\pi}'\hat{S}\hat{\delta}}{(1-k)\,n^{-1}\tilde{T}'\tilde{T}+k\,\hat{\pi}'\hat{S}\hat{\pi}},\qquad k:=\frac{1-1/n}{1-m/n-1/n}.
\]
The terms $\hat{\pi}'\hat{S}\hat{\pi}$ and $\hat{\pi}'\hat{S}\hat{\delta}$ are the TSLS building blocks; the choice of $k$ removes the many-instrument bias. Standard errors use the variance formula in Theorem 2 of \textcite{kolesar2015identification}, with degrees-of-freedom-corrected reduced-form covariance estimates and the residual-concentration plug-in truncated at zero.

Consistency of MBTSLS under excludability violations requires the violations to be orthogonal to the first-stage coefficients. In the rescaled system the leading bias of $\hat{\beta}_{MB}$ is proportional to $\pi'Sb$, and the orthogonality condition requires this cross-moment to vanish in expectation. Since $E[b]=\bar{\mu}1_{m}$ and $b$ is independent of the instruments,
\[
E[\pi'Sb]=\bar{\mu}\,\pi'S1_{m}.
\]
The orthogonality condition fails in this design whenever $\bar{\mu}\neq0$: the multinomial cell structure gives $S1_{m}=(1-\sum_{j}q_{j})\,q$, so $\pi'S1_{m}=0.1\,\pi'q\neq0$. Had the instruments instead summed to one, as in the judge-IV designs of Section~\ref{subsec:Intercepts-in-linear}, demeaning would force $S1_{m}=0$, and orthogonality would reduce to non-correlation between the centered violations $b-\bar{\mu}1_{m}$ and the first stage. Exchangeability delivers this non-correlation in expectation. The realized cross-moment $\pi'S(b-\bar{\mu}1_{m})$ vanishes almost surely when $\bar{\sigma}^{2}=0$. 

The differenced MBTSLS estimator applies the identical construction to the differenced system: $m-1$ replaces $m$ in $k$ and in the degrees-of-freedom corrections, the projected blocks are computed from $(\pi_{D},\delta_{D},W_{D})$, and the unprojected blocks $n^{-1}\tilde{T}'\tilde{T}$ and $n^{-1}\tilde{T}'\tilde{Y}$ are unchanged. Differencing restores the orthogonality condition in expectation: $Db$ has mean zero for every $\bar{\mu}$.

\subsection*{Results}

Table~\ref{tab:appc-differenced} reports baseline performance and Table~\ref{tab:appc-sweeps} reports RMSE across the two hyperparameter sweeps, both computed on the same simulated data sets as Tables~\ref{tab:mc-main}--\ref{tab:mc-rmse-vs-sig2} and with the level MBTSLS estimator included for comparison. Four findings stand out.

First, differencing and estimating the mean are nearly equivalent. Differenced GMM tracks BC (baseline RMSE 0.253 versus 0.254) and differenced EB tracks EB (0.243 versus 0.242), with agreement across both hyperparameter sweeps. Both members of each pair remove the same mean direction, so whether $\mu$ is estimated or eliminated matters little for performance. Valid inference still requires estimating the dispersion, however: the naive standard errors for differenced GMM reject a true null 15.8\% of the time at the 5\% level, while the misspecification-aware standard errors reject 7.1\%.

\setcounter{table}{0}
\renewcommand{\thetable}{C.\arabic{table}}%
\begin{table}[H]
\caption{Differenced and k-class estimators: baseline performance ($\bar{\mu}=8$, $\bar{\sigma}=8$, $n=10{,}000$, $m=40$)\protect\label{tab:appc-differenced}}

\begin{centering}
\begin{tabular}{lccccc}
\toprule
 & Bias & Std & RMSE & Rej. (naive) & Rej. (aware)\tabularnewline
\midrule
Diff. GMM & 0.039 & 0.250 & 0.253 & 0.158 & 0.071\tabularnewline
Diff. EB & 0.033 & 0.241 & 0.243 & -- & 0.080\tabularnewline
Diff. MBTSLS & $-0.011$ & 0.357 & 0.357 & -- & 0.101\tabularnewline
MBTSLS & 0.234 & 0.230 & 0.328 & -- & 0.090\tabularnewline
\bottomrule
\end{tabular}
\par\end{centering}
{\small\emph{Notes:}}{\small{} Same design and Monte Carlo draws as Table~\ref{tab:mc-main}, with 1,000 replications. Rej.\ columns report rejection rates of a two-sided Wald test of $\beta_{0}=1$ at the 5\% nominal level. The naive column uses the heteroscedasticity-robust standard errors that ignore misspecification, reported for differenced GMM. The aware column uses the misspecification-aware standard errors for differenced GMM and EB and the Theorem-2 standard errors of \textcite{kolesar2015identification} for both MBTSLS estimators. The MBTSLS row is estimated in levels; all other rows use the differenced system.}{\small\par}
\end{table}

Second, the mean channel accounts for the entire bias of MBTSLS in this design. In levels, its baseline bias of 0.234 is comparable to GMM's 0.220, and its RMSE rises steeply with $\bar{\mu}$, from 0.231 at $\bar{\mu}=0$ to 0.534 at $\bar{\mu}=16$: the violations share the mean $\bar{\mu}$, so they are not orthogonal to the first stage. Differencing eliminates the bias ($-0.011$ at baseline) and renders the RMSE flat in $\bar{\mu}$. 

Third, the repair is expensive. The differenced MBTSLS standard deviation at baseline is 0.357, against 0.241 for differenced EB, and its RMSE reaches 0.572 at $\bar{\sigma}=16$. Differencing sacrifices part of the identifying variation, and the k-class correction amplifies noise as the first stage weakens. The EB estimator removes the same mean bias at a far smaller variance cost.

\begin{table}[H]
\caption{Differenced and k-class estimators: RMSE by $\bar{\mu}$ and $\bar{\sigma}$ ($n=10{,}000$, $m=40$)\protect\label{tab:appc-sweeps}}

\begin{centering}
\begin{tabular}{lccccc}
\toprule
 & \multicolumn{5}{c}{RMSE by $\bar{\mu}$ ($\bar{\sigma}=8$)}\tabularnewline
\cmidrule{2-6}
 & $\bar{\mu}=0$ & $\bar{\mu}=2$ & $\bar{\mu}=4$ & $\bar{\mu}=8$ & $\bar{\mu}=16$\tabularnewline
\midrule
Diff. GMM & 0.254 & 0.255 & 0.254 & 0.253 & 0.246\tabularnewline
Diff. EB & 0.243 & 0.245 & 0.243 & 0.243 & 0.236\tabularnewline
Diff. MBTSLS & 0.353 & 0.359 & 0.356 & 0.357 & 0.343\tabularnewline
MBTSLS & 0.231 & 0.239 & 0.257 & 0.328 & 0.534\tabularnewline
\midrule
 & \multicolumn{5}{c}{RMSE by $\bar{\sigma}$ ($\bar{\mu}=8$)}\tabularnewline
\cmidrule{2-6}
 & $\bar{\sigma}=0$ & $\bar{\sigma}=2$ & $\bar{\sigma}=4$ & $\bar{\sigma}=8$ & $\bar{\sigma}=16$\tabularnewline
\midrule
Diff. GMM & 0.187 & 0.192 & 0.204 & 0.253 & 0.406\tabularnewline
Diff. EB & 0.189 & 0.193 & 0.205 & 0.243 & 0.348\tabularnewline
Diff. MBTSLS & 0.259 & 0.267 & 0.283 & 0.357 & 0.572\tabularnewline
MBTSLS & 0.296 & 0.299 & 0.306 & 0.328 & 0.407\tabularnewline
\bottomrule
\end{tabular}
\par\end{centering}
{\small\emph{Notes:}}{\small{} Same design and Monte Carlo draws as Tables~\ref{tab:mc-rmse-vs-mu} and \ref{tab:mc-rmse-vs-sig2}, with 1,000 replications per grid point. The MBTSLS row is estimated in levels; all other rows use the differenced system.}{\small\par}
\end{table}

Finally, MBTSLS is competitive with BC and EB in terms of RMSE only when $\bar \mu\leq 4$ and only in its levels form. When $\bar \mu = 0$, levels MBTSLS dominates both corrected estimators because it eliminates higher-order many-instruments biases, which are the only systematic biases present in the design. However, differenced MBTSLS, which also insures against nonzero mean biases, is dominated by both BC and EB under this regime because of the precision costs of removing the many-instruments bias.

\end{document}